%% file: main.tex
\documentclass[a4paper, UKenglish, cleveref, autoref, thm-restate]{lipics-v2021}
\nolinenumbers

\title{Weighted basic parallel processes\\ and combinatorial enumeration}

\author{Lorenzo Clemente}{University of Warsaw, Department of Mathematics, Mechanics, and Computer Science \and \url{http://www.myhomepage.edu} }{clementelorenzo@gmail.com}{https://orcid.org/1234-5678-9012}{}

\authorrunning{L.~Clemente} 

\Copyright{Lorenzo Clemente} 



\relatedversion{Extended version of a paper accepted for publication at CONCUR'24.} 

\funding{Supported by the ERC grant INFSYS, agreement no.~950398.}

\acknowledgements{We warmly thank Mikołaj Bojańczyk, Arka Ghosh, Filip Mazowiecki, and Paweł Parys for their comments and support at the various stages of this work.}

\EventEditors{John Q. Open and Joan R. Access}
\EventNoEds{2}
\EventLongTitle{42nd Conference on Very Important Topics (CVIT 2016)}
\EventShortTitle{CVIT 2016}
\EventAcronym{CVIT}
\EventYear{2016}
\EventDate{December 24--27, 2016}
\EventLocation{Little Whinging, United Kingdom}
\EventLogo{}
\SeriesVolume{42}
\ArticleNo{23}

\input{packages}
\input{macros}

\begin{document}

\maketitle

\begin{abstract}
    We study \emph{weighted basic parallel processes} (\WBPP),
    a nonlinear recursive generalisation of weighted finite automata
    inspired from process algebra and Petri net theory.
    %
    Our main result is an algorithm of \TWOEXPSPACE~complexity for the \WBPP~equivalence problem.
    While (unweighted) \BPP~language equivalence is undecidable,
    we can use this algorithm to decide multiplicity equivalence of \BPP\-
    and language equivalence of unambiguous \BPP, with the same complexity.
    These are long-standing open problems for the related model of weighted context-free grammars.

	Our second contribution is a connection between \WBPP,
	power series solutions of systems of polynomial differential equations, and combinatorial enumeration.
	To this end we consider \emph{constructible differentially finite} power series (\CDF),
	a class of multivariate differentially algebraic series
	introduced by Bergeron and Reutenauer in order to provide a combinatorial interpretation to differential equations.
	\CDF~series generalise rational, algebraic, and a large class of D-finite (holonomic) series,
	for which decidability of equivalence was an open problem.
    We show that \CDF~series correspond to \emph{commutative} \WBPP~series.
    As a consequence of our result on \WBPP~and commutativity,
    we show that equivalence of \CDF~power series can be decided with \TWOEXPTIME~complexity.
    
    In order to showcase the \CDF~equivalence algorithm,
    we show that \CDF~power series naturally arise from combinatorial enumeration,
    namely as the exponential generating series of \emph{constructible species of structures}.
    Examples of such species include sequences, binary trees, ordered trees, Cayley trees, set partitions, series-parallel graphs, and many others.
    As a consequence of this connection, we obtain an algorithm to decide multiplicity equivalence of constructible species,
    decidability of which was not known before.

    The complexity analysis is based on effective bounds from algebraic geometry,
    namely on the length of chains of polynomial ideals constructed by repeated application of finitely many,
    not necessarily commuting derivations of a multivariate polynomial ring.
    This is obtained by generalising a result of Novikov and Yakovenko in the case of a single derivation,
    which is noteworthy since generic bounds on ideal chains are non-primitive recursive in general.
    On the way, we develop the theory of \WBPP~series and \CDF~power series,
    exposing several of their appealing properties.   
\end{abstract}

\input{01-introduction}
\input{02-preliminaries}
\input{03-WBPP}
\input{03-zeroness}
\input{04-CDF}
\input{05-species}
\input{06-conclusions}

\newpage
\bibliographystyle{plainurl}
\bibliography{literature}

\appendix
\input{A02-preliminaries}
\input{A03-WBPP}
\input{A04-shuffle-finite}

\input{A04-zeroness}
\input{A05-CDF}

\end{document}

%% file: packages.tex
\usepackage[utf8]{inputenc}
\usepackage{stmaryrd}
\usepackage{amsthm, amsmath}
\usepackage{thm-restate}
\usepackage{xspace}
\usepackage[textwidth=2cm,textsize=small]{todonotes}
\usepackage{cleveref}
\usepackage{hyperref}
\usepackage{stackrel}
\usepackage{todonotes}
\usepackage{paralist}
\usepackage{graphicx}
\usepackage{etoolbox}
\usepackage{textcase}
\usepackage{wrapfig}

\usepackage{amstext} 
\usepackage{array}   
\newcolumntype{C}{>{$}c<{$}} 
\newcolumntype{L}{>{$}l<{$}} 

\usepackage{fix-cm}
\usepackage{shuffle}
\DeclareFontFamily{U}{shuffle}{}
\DeclareFontShape{U}{shuffle}{m}{n}{ <-8>shuffle7 <8->shuffle10}{}

\hypersetup{
    colorlinks,
    linkcolor={red!50!black},
    citecolor={blue!50!black},
    urlcolor={blue!80!black}
}

\usepackage{tikz}
\usetikzlibrary{positioning}

\makeatletter
\patchcmd{\@sect}{\uppercase}{\MakeTextUppercase}{}{}
\patchcmd{\@sect}{\uppercase}{\MakeTextUppercase}{}{}
\makeatother

%% file: macros.tex
\newcommand{\tuple}[1]{\left(#1\right)}
\newcommand{\tuplesmall}[1]{(#1)}
\newcommand{\goesto}[1]{\xrightarrow{#1}}

\newcommand{\e}{\varepsilon}

\newcommand{\ignore}[1]{}

\newcommand{\wlg}{w.l.o.g.}
\newcommand{\wrt}{w.r.t.}
\newcommand{\cf}{cf.}
\newcommand{\eg}{e.g.}
\newcommand{\aka}{a.k.a.}
\newcommand{\ie}{i.e.}
\newcommand{\rhs}{r.h.s.}
\newcommand{\lhs}{l.h.s.}
\newcommand{\egs}{EGS}
\newif\ifstartedinmathmode
\renewcommand*{\st}{
  \relax\ifmmode\startedinmathmodetrue\else\startedinmathmodefalse\fi
  \ifstartedinmathmode{\;\cdot\;}\else{s.t.}\fi%
}

\newcommand{\C}{\mathbb C}

\newcommand{\N}{\mathbb N}
\newcommand{\Z}{\mathbb Z}
\newcommand{\Q}{\mathbb Q}

\newcommand{\F}{\mathbb F}

\newcommand{\X}{\mathcal X}
\newcommand{\Y}{\mathcal Y}
\newcommand{\ZZ}{\mathcal Z}
\newcommand{\ZERO}{\mathbf 0}
\newcommand{\ONE}{\mathbf 1}
\newcommand{\FF}{\mathcal F}
\newcommand{\GG}{\mathcal G}

\let\oldcirc\circ
\let\circ\relax
\DeclareMathOperator{\circ}{\oldcirc\,}

\newcommand{\compose}[1]{\,\oldcirc_{#1}\,}

\let\oldshuffle\shuffle
\let\shuffle\relax
\DeclareMathOperator{\shuffle}{\oldshuffle}
\newcommand{\shufflepower}[2]{{#1}^{\oldshuffle#2}}

\let\merge\relax
\DeclareMathOperator{\merge}{\parallel}

\let\colon\relax
\DeclareMathOperator{\colon}{\,:\,}

\DeclareMathOperator{\cauchy}{\ast}
\DeclareMathOperator{\hadamard}{\odot}

\DeclareMathOperator{\binconv}{\ast}

\newcommand{\abs}[1]{\left|#1\right|}
\newcommand{\card}[1]{\left|#1\right|}
\newcommand{\length}[1]{\left|#1\right|}
\newcommand{\norm}[2]{\left| #2 \right|_{#1}}
\newcommand{\onenorm}[1]{\norm 1 {#1}}
\newcommand{\inftynorm}[1]{\norm \infty {#1}}
\newcommand{\height}[1]{\inftynorm{#1}}
\newcommand{\set}[1]{\left\{#1\right\}}
\newcommand{\setof}[2]{\set{#1 \;\middle|\; #2}}
\newcommand{\multiset}[1]{\left\{\!\left\{#1\right\}\!\right\}}

\newcommand{\rad}[1]{\sqrt{#1}}
\newcommand{\poly}[2]{#1[#2]}
\newcommand{\polyof}[3]{\poly {#1} {#2 \mid #3}}
\newcommand{\ncpoly}[2]{#1\langle#2\rangle}
\newcommand{\rational}[2]{#1(#2)}
\newcommand{\powerseries}[2]{#1[\![#2]\!]}

\newcommand{\series}[2]{#1\langle\!\langle#2\rangle\!\rangle}
\newcommand{\ncpowerseries}[2]{\series {#1} {#2}}

\newcommand{\sequences}[2]{#1[\![\N^{#2}]\!]}
\newcommand{\matrices}[3]{#3^{#1 \times #2}}

\newcommand{\ideal}[1]{\langle#1\rangle}
\newcommand{\idealof}[2]{\ideal{#1 \mid #2}}
\renewcommand{\dim}[1]{\mathsf{dim}\,{#1}}
\newcommand{\leadingterm}[1]{\mathsf{lt}\,{#1}}
\newcommand{\sem}[1]{\left\llbracket#1\right\rrbracket}
\newcommand{\support}[1]{\mathsf{supp}(#1)}

\newcommand{\lang}[1]{L\left(#1\right)}
\newcommand{\bigO}[1]{O\left(#1\right)}

\newcommand{\Parikh}[1]{\#({#1})}
\newcommand{\coefficient}[2]{[#1]#2}
\newcommand{\EGS}[1]{\mathsf{EGS}[#1]}
\newcommand{\SETSP}{\mathsf{SET}}
\newcommand{\SET}[1]{\SETSP[#1]}
\newcommand{\SEQSP}{\mathsf{SEQ}}
\newcommand{\SEQ}[1]{\SEQSP[#1]}
\newcommand{\CYCSP}{\mathsf{CYC}}

\newcommand{\CAYLEY}[1]{\mathsf{C}[#1]}

\newcommand{\restrict}[2]{\left.#1\right|_{#2}}
\newcommand{\restrictsmall}[2]{#1|_{#2}}
\newcommand{\equivmod}[3]{{#1} \equiv {#2}\; (\mathsf{mod}\ #3)}
\newcommand{\lneg}{\neg}
\newcommand{\ord}[1]{\mathsf{ord}(#1)}
\newcommand{\stops}[1]{\mathsf{s2p}\left(#1\right)}
\newcommand{\pstos}[1]{\mathsf{p2s}\left(#1\right)}
\newcommand{\erf}{\mathrm{erf}}
\newcommand{\arsinh}{\mathrm{arsinh}}

\newcommand{\artanh}{\mathrm{artanh}}

\newcommand{\sech}{\mathrm{sech}}

\newcommand{\derive}[1]{\delta_{#1}}
\let\oldpartial\partial
\renewcommand{\partial}[1]{\oldpartial_{#1}}
\newcommand{\shift}[1]{\sigma_{#1}}
\newcommand{\deductionrule}[2]{\begin{array}{c}%
{#1} \\ \hline {#2}%
\end{array}}

\newcommand{\CDF}{\textsf{CDF}}

\newcommand{\BPP}{\textsf{BPP}}
\newcommand{\BPA}{\textsf{BPA}}
\newcommand{\WBPP}{\textsf{WBPP}}
\newcommand{\WFA}{\textsf{WFA}}

\newcommand{\NC}{\textsf{NC}}

\newcommand{\PTIME}{\textsf{PTIME}}

\newcommand{\PSPACE}{\textsf{PSPACE}}
\newcommand{\EXPTIME}{\textsf{EXPTIME}}

\newcommand{\EXPSPACE}{\textsf{EXPSPACE}}
\newcommand{\TWOEXPTIME}{2-\EXPTIME}

\newcommand{\TWOEXPSPACE}{2-\EXPSPACE}

\AtEndPreamble{%
  \newtheorem{fact}[theorem]{Fact}
}

\crefname{equation}{}{}
\crefname{definition}{Definition}{Definitions}
\crefname{claim}{Claim}{Claims}
\crefname{section}{Section}{Sections}
\crefname{lemma}{Lemma}{Lemmas}
\crefname{corollary}{Corollary}{Corollaries}
\crefname{theorem}{Theorem}{Theorems}
\crefname{figure}{Figure}{Figures}
\crefname{fact}{Fact}{Facts}
\crefname{example}{Example}{Examples}
\crefname{remark}{Remark}{Remarks}

\crefname{section}{§}{§§}


%% file: 01-introduction.tex
\section{Introduction}

We study the equivalence problem for a class of finitely presented series
originating in weighted automata, process algebra, and combinatorics.
We begin with some background.

\subsection{Motivation and context}

\subparagraph{Weighted automata.}

Classical models of computation arising in the seminal work of Turing from the 1930's~\cite{Turing_1937}
have a Boolean-valued semantics (``is an input accepted?'')
and naturally recognise languages of finite words $L \subseteq \Sigma^*$.
In the 1950's a finite-memory restriction was imposed on Turing machines,
leading to an elegant and robust theory of \emph{finite automata}~\cite{RabinScott:1959},
with fruitful connections with logic~\cite{Buchi:WMSO:1960,Elgot:WMSO:1961,Trakhtenbrot:WMSO:1962} and regular expressions~\cite{Kleene:1956}. 
\emph{Weighted finite automata} over a field $\F$ (\WFA)~\cite{Schutzenberger:IC:1961}
were introduced in the 1960's by Schützenberger
as a generalisation of finite automata
to a quantitative \emph{series} semantics $\Sigma^* \to \F$ (``in how many ways can an input be accepted?'').
%
This has been followed by the development of a rich theory of weighted automata and logics~\cite{HandbookWA}.
While the general theory can be developed over arbitrary semirings,
the methods that we develop in this work are specific to fields,
and in particular for effectiveness we assume the field of rational numbers $\F = \Q$.

The central algorithmic question that we study is the \emph{equivalence problem}:
Given two (finitely presented) series $f, g : \Sigma^* \to \Q$, is it the case that $f = g$?
(In algorithmic group theory this is known as the \emph{word problem}.)
%
%
A mathematical characterisation of equivalence yields a deeper understanding of the interplay between syntax and semantics,
and a decidability result means that this understanding is even encoded as an algorithm.
Equivalence of weighted models generalises \emph{multiplicity equivalence} of their unweighted counterparts
(``do two models accept each input in the same number of ways?''),
in turn generalising language equivalence of \emph{unambiguous models}
(each input is accepted with multiplicity $0$ or $1$).
Since equivalence $f = g$ reduces to \emph{zeroness} $f - g = 0$,
from now on we will focus on the latter problem.

While nonemptiness of \WFA~is undecidable~\cite[Theorem 21]{NasuHonda:IC:1969}
(later reported in Paz' book~\cite[Theorem 6.17]{Paz:1971}),
zeroness is decidable,
even in polynomial time~\cite{Schutzenberger:IC:1961}---%
a fact often rediscovered, \eg,~\cite{StearnsHunt:JoC:1985,Tzeng:SIAMJC:1992}.
This has motivated the search for generalisations of \WFA~with decidable zeroness.
However, many of them are either known to be undecidable (\eg, \emph{weighted Petri nets}~\cite[Theorem 3]{Jancar:TCS:2001}),
or beyond the reach of current techniques (\eg, \emph{weighted one counter automata}, \emph{weighted context-free grammars}, and \emph{weighted Parikh automata}).
One notable exception is \emph{polynomial weighted automata},
although zeroness has very high complexity (Ackermann-complete)~\cite{BenediktDuffSharadWorrell:PolyAut:2017}.
%
%
%
In the restricted case of a unary input alphabet, 
decidability and complexity results can be obtained with algebraic~\cite{BalajiClementeNosanShirmohammadiWorrell:LICS:2023}
and D-finite techniques~\cite{BostanCarayolKoechlinNicaud:ICALP:2020}.

\subparagraph{Process algebra.}

On a parallel line of research, the process algebra community
has developed a variety of formalisms modelling different aspects of concurrency and nondeterminism.
We focus on \emph{basic parallel processes} (\BPP)~\cite{Christensen:PhD:1993},
a subset of the \emph{calculus of communicating systems} without sequential composition~\cite{Milner:CCS:1980}.
\BPP~are also known as \emph{communication-free Petri nets}
(every transition consumes exactly one token)
and \emph{commutative context-free grammars}
(nonterminals in sentential forms are allowed to commute with each other).
While language equality for \BPP~is undecidable~\cite{Huttel:TACS:1994,HuttelKobayashiSuto:IC:2009},
bisimulation equivalence is decidable~\cite{ChristensenHirshfeldMoller:CONCUR:1993}
(even \PSPACE-complete~\cite{Srba:STACS:2002,Jancar:LICS:2003}).
\emph{Multiplicity equivalence},
finer than language equality and incomparable with bisimulation,
does not seem to have been studied for \BPP.

\subparagraph{Combinatorial enumeration and power series.}

We shall make a connection between \BPP, power series, and combinatorial enumeration.
For this purpose, let us recall that the study of multivariate power series in commuting variables has a long tradition
at the border of combinatorics, algebra, and analysis of algorithms~\cite{Stanley:EC:CUP:2011,FlajoletSedgewick:AC:2009}.
We focus on \emph{constructible differentially finite} power series (\CDF)\-
\cite{BergeronReutenauer:EJC:1990,BergeronSattler:TCS:1995},
a class of differentially algebraic power series arising in combinatorial enumeration\-
\cite{LerouxViennot:CE:1986,BergeronFlajoletSalvy:CAAP:1992}.
Their study was initiated in the \emph{univariate} context in~\cite{BergeronReutenauer:EJC:1990},
later extended to multivariate~\cite{BergeronSattler:TCS:1995}.
They generalise rational and algebraic power series,
and are incomparable with D-finite power series~\cite{Stanley:EJC:1980,Lipshitz:D-finite:JA:1989}.
For instance, the exponential generating series $\sum_{n \in \N} n^{n-1} \cdot x^n / n!$
of Cayley trees is \CDF, but it is neither algebraic/D-finite~\cite[Theorem 1]{BostanJimenezPastor:2020},
nor polynomial recursive~\cite[Theorem 5.3]{Cadilhac:Mazowiecki:Paperman:Pilipczuk:Senizergues:ToCS:2021}.

The theory of \emph{combinatorial species}~\cite{Joyal:AM:1981,CombinatorialSpecies:CUP:1998}
is a formalism describing families of finite structures.
It arises as a categorification of power series,
by noticing how primitives used to build structures---%
sum, combinatorial product, composition, differentiation, resolution of implicit equations,---%
are in a one-to-one correspondence with corresponding primitives on series.
Using these primitives, a rich class of \emph{constructible species} can be defined~\cite{PivoteauSalvySoria:JCT:2012}.
For instance the species $\CAYLEY \X$ of Cayley trees (rooted unordered trees)
is constructible since it satisfies $\CAYLEY \X = \X \cdot \SET {\CAYLEY \X}$.
%
%
Two species are \emph{multiplicity equivalent} (\emph{equipotent}~\cite{PivoteauSalvySoria:JCT:2012})
if for every $n \in \N$ they have the same number of structures of size $n$.
Multiplicity equivalence of species has not been studied from an algorithmic point of view.

\subsection{Contributions}


We study \emph{weighted basic parallel processes} over the field of rational numbers (\WBPP),
a weighted extension of \BPP~generalising \WFA.
The following is our main contribution.
\begin{restatable}{theorem}{WBPPmain}
    \label{thm:main}
    The zeroness problem for \WBPP~is in \TWOEXPSPACE.
\end{restatable}
\noindent
This elementary complexity should be contrasted with Ackermann-hardness
of zeroness of polynomial automata~\cite{BenediktDuffSharadWorrell:PolyAut:2017},
another incomparable extension of \WFA.
Since \WBPP~can model the multiplicity semantics of \BPP,
as an application we get the following corollary.
\begin{restatable}{corollary}{BPPmultiplicityEquivalenceAndUnambiguous}
    \label{thm:BPP multiplicity equivalence}
    \label{thm:UBPP language equivalence}
    Multiplicity equivalence of \BPP~and language equivalence of unambiguous \BPP~are decidable in \TWOEXPSPACE.
\end{restatable}
\noindent
On a technical level, \cref{thm:main} is obtained by extending an ideal construction and complexity analysis from~\cite{NovikovYakovenko:1999}
from the case of a single polynomial derivation
to the case of a finite set of not necessarily commuting polynomial derivations.
It is remarkable that such ideal chains have elementary length,
since generic bounds without further structural restrictions are only general recursive~\cite{Seidenberg:TAMS:1974}.
This shows that the \BPP~semantics is adequately captured by differential algebra.
These results are presented in~\cref{sec:WBPP}.
%
%
In \cref{sec:CDF} we observe that \emph{commutative} \WBPP~series coincide with \CDF~power series,
thus establishing a novel connection between automata theory, polynomial differential equations, and combinatorics.
%
This allows us to obtain a zeroness algorithm for \CDF,
which is our second main contribution.

\begin{restatable}{theorem}{CDFzeronessTWOEXPTIME}
    \label{thm:CDF zeroness in TWOEXPTIME}
    The zeroness problem for multivariate \CDF~power series is in \TWOEXPTIME.
\end{restatable}
\noindent
The complexity improvement from \TWOEXPSPACE~to \TWOEXPTIME~is due to commutativity.
In the special \emph{univariate} case,
decidability was observed in~\cite{BergeronReutenauer:EJC:1990} with no complexity analysis,
while~\cite{BergeronSattler:TCS:1995} did not discuss decidability in the multivariate case.
%
%
In~\cref{sec:species} we apply~\cref{thm:CDF zeroness in TWOEXPTIME} to multiplicity equivalence of a class of constructible species.
This follows from the observation that their exponential generating series (\egs) are effectively \CDF,
proved by an inductive argument based on the closure properties of \CDF~series.
For instance, the \egs~of Cayley trees satisfies $C = x \cdot e^C$;
by introducing auxiliary series $D := e^C, E := (1 - C)^{-1}$ and by differentiating
we obtain \CDF~equations $\partial x C = D \cdot E$, $\partial x D = D^2 \cdot E$, $\partial x E = D \cdot E^3$.

\begin{restatable}{theorem}{multiplicityEquivalenceOfConstructibleSpecies}
    \label{thm:multiplicity equivalence of constructible species}
    Multiplicity equivalence of strongly constructible species is decidable.
\end{restatable}

\subsection{Related works}


There have recently been many decidability results for models incomparable with \WBPP,
such as multiplicity equivalence of \emph{boundedly-ambiguous} Petri nets~\cite[Theorem 3]{CzerwinskiHofman:CONCUR:2022};
zeroness for weighted one-counter automata with \emph{deterministic counter updates}~\cite{PrincePenelleSaivasanSreejith:FSTTCS:2023};
zeroness of \emph{P-finite automata}, a model intermediate between \WFA~and polynomial automata
(even in \PTIME~\cite{Buna-MargineanChevalShirmohammadiWorrell:POPL:2024});
and zeroness of \emph{orbit-finite weighted automata} in sets with atoms~\cite{BojanczykKlinMoerman:LICS:2021}.


Regarding power series, there is a rich literature on dynamical systems
satisfying differential equations in the \CDF~format,
that is polynomial ordinary differential equations 
(ODE; \cf~\cite{Platzer:2018} and references therein).
While many algorithms have been proposed for their analysis
(e.g., invariant checking~\cite{PlatzerTan:JACM:2020}),
the complexity of the zeroness problem has not been addressed before.
A decision procedure for zeroness of multivariate \CDF~can be obtained from first principles
as a consequence of Hilbert's \emph{finite basis theorem}~\cite[Theorem 4, §5, Ch.~2]{CoxLittleOShea:Ideals:2015}.
For instance, decidability follows from the algorithm of~\cite{Boreale:IC:2022}
computing pre- and post-conditions for restricted systems of partial differential equations (covering $\CDF$),
and also from the \emph{Rosenfeld--Gr\"obner algorithm}~\cite{BoulierLazardOllivierPetitot:AAECC:2009},
which can be used to test membership in the radical differential ideal generated by the system of \CDF~equations.
In both cases, no complexity-theoretic analysis is provided and only decidability can be deduced.
In the univariate \CDF~case, decidability can also be deduced from~\cite{Boreale:LMCS:2019,Boreale:SCP:2020}.
Univariate \CDF~also arise in the coalgebraic treatment of stream equations with the shuffle product~\cite{BorealeGorla:CONCUR:2021,BorealeCollodiGorla:ACMTCL:2024},
where an equivalence algorithm based on Hilbert's theorem is provided.

The work~\cite{GabrielovVorobjov:2004} studies \emph{Noetherian functions},
which are analytic functions satisfying \CDF~equations.
In fact, Noetherian functions which are analytic around the origin coincide with multivariate \CDF~power series.
The work~\cite{vanderHoevenShackell:JSC:2006} discusses a subclass of Noetherian functions
obtained by iteratively applying certain extensions to the ring of multivariate polynomials
and presents a zeroness algorithm running in doubly exponential time.
\cref{thm:CDF zeroness in TWOEXPTIME} is more general since it applies to all Noetherian power series.

In the context of the realisability problem in control theory,
Fliess has introduced the class of \emph{differentially producible series}~\cite{Fliess:IM:1983}
(\cf~also the exposition of Reutenauer~\cite{Reutenauer:Lie:1986}),
a generalisation of \WBPP~series where the state and transitions are given by arbitrary power series (instead of polynomials).
Such series are characterised by a notion of \emph{finite Lie rank}
and it is shown that differentially producible series of minimal Lie rank exist and are unique.
Such series are not finitely presented and thus algorithmic problems, such as equivalence, cannot even be formulated.

Full proofs can be found in~\cref{app:preliminaries,app:WBPP,app:CDF}.

%% file: 02-preliminaries.tex

\subparagraph*{Preliminaries.}
\label{sec:preliminaries}

Let $\Sigma = \set{a_1, \dots, a_d}$ be a finite alphabet.
We denote by $\Sigma^*$ the set of \emph{finite words} over $\Sigma$,
a monoid under the operation of concatenation,
with neutral element the empty word $\e$.
The \emph{Parikh image} of a word $w \in \Sigma^*$
is $\Parikh w := \tuple{\Parikh w_{a_1}, \dots, \Parikh w_{a_d}} \in \N^d$,
where $\Parikh w_{a_j}$ is the number of occurrences of $a_j$ in $w$.
Let $\Q$ be the field of rational numbers.
Most results in the paper hold for any field,
however for computability considerations we restrict our presentation to $\Q$.
For a tuple of commuting indeterminates $x = \tuple{x_1, \dots, x_k}$,
denote by $\poly \Q x$ the ring of \emph{multivariate polynomials}
($\poly \Q k$ when the name of variables does not matter)
and by $\rational \Q x$ its fraction field of \emph{rational functions}
(that is, ratios of polynomials $p(x)/q(x)$).
The \emph{one norm} $\onenorm z$ of a vector $z = \tuple{z_1, \dots, z_k} \in \Q^k$
is $\abs {z_1} + \cdots + \abs {z_k}$,
and the \emph{infinity norm} is $\inftynorm z = \max_{1 \leq i \leq k} \abs {z_k}$.
Similarly, the \emph{infinity norm} (also called \emph{height})
of a polynomial $p \in \poly \Q k$,
written $\height p$, is the maximal absolute value of any of its coefficients.

A \emph{derivation} of a ring $R$ is a linear function $\delta : R \to R$ satisfying 
\begin{align}
    \tag{Leibniz rule}
    \label{eq:Leibniz rule}
    \delta (a \cdot b) = \delta(a) \cdot b + \delta (b).
\end{align}
%
%
%
A derivation $\delta$ of a polynomial ring $\poly R x$ 
is uniquely defined once we fix $\delta(x) \in \poly R x$.
%
%
\noindent
For instance, $\partial x : \poly R x \to \poly R x$
is the unique derivation $\delta$ of the polynomial ring \st~$\delta(x) = 1$.
Other technical notions will be recalled when necessary.
For a general introduction to algebraic geometry we refer to~\cite{CoxLittleOShea:Ideals:2015}.

%% file: 03-WBPP.tex
\section{Weighted extension of basic parallel processes}
\label{sec:WBPP}

\subsection{Basic parallel processes}

In this section we recall the notion of basic parallel process (\BPP) together with its language semantics.
Let $\set{X_1, X_2, \dots}$ be a countable set of \emph{nonterminals} (process variables)
and let $\Sigma$ be a finite alphabet of \emph{terminals} (actions).
A \emph{\BPP~expression} is generated by the following abstract grammar~(\cf~\cite[Sec.~5]{Esparza:FI:1997}):
%
    $E, F ::= \bot \mid X_i \mid a . E \mid E + F \mid E \merge F$.
%
Intuitively, $\bot$ is a constant representing the \emph{terminated process},
$a.E$ (\emph{action prefix}), is the process that performs action $a$ and becomes $E$,
$E+F$ (\emph{choice}) is the process that behaves like $E$ or $F$,
and $E \merge F$ (\emph{merge}) is the \emph{parallel} execution of $E$ and $F$.
We say that an expression $E$ is \emph{guarded}
if every occurrence of a nonterminal $X_i$ is under the scope of an action prefix.
A \emph{\BPP} consists of a distinguished \emph{starting nonterminal} $X_1$ and rules
\begin{align}
    \label{eq:BPP}
    X_1 \to E_1
        \quad \cdots \quad
            X_k \to E_k,
\end{align}
where the \rhs~expressions $E_1, \dots, E_k$ are guarded
and contain only nonterminals $X_1, \dots, X_k$.
\begin{wrapfigure}{r}{0.5\textwidth}
    \vspace{-3em}
    \begin{center}
        \begin{align*}
            \begin{array}{ccc}
                a.E \goesto a E
                & \deductionrule {E \goesto a E'} {E + F \goesto a E'}
                & \deductionrule {F \goesto a F'} {E + F \goesto a F'} \\
                \deductionrule {E_i \goesto a E'} {X_i \goesto a E'}
                & \deductionrule {E \goesto a E'} {E \merge F \goesto a E' \merge F}
                & \deductionrule {F \goesto a F'} {E \merge F \goesto a E \merge F'}
            \end{array}
        \end{align*}
    \end{center}
    \vspace{-2.5em}
\end{wrapfigure}
A \BPP~induces an infinite labelled transition system 
where states are expressions and the labelled transition relations ${\goesto a}$
are the least family of relations closed under the rules on the side.
The transition relation is extended naturally to words ${\goesto w}$, $w \in \Sigma^*$.
An expression $E$ is \emph{final} 
if there are no $a, E'$ \st~$E \goesto a E'$
(e.g., $\bot \merge \bot$);
it \emph{accepts} a word $w \in \Sigma^*$ if there is a final expression $F$ \st~$E \goesto w F$.
The \emph{language} $\lang E$ recognised by an expression $E$ is the set of words it accepts,
and the language of a \BPP~is $\lang {X_1}$. 

An expression $E$ is in \emph{(full) standard form}
if it is a sum of products $a_1.\alpha_1 + \cdots + a_n.\alpha_n$,
where each $\alpha_i$ is a merge of nonterminals;
a \BPP~\cref{eq:BPP} is in standard form if every $E_1, \dots, E_k$ is in standard form.
The standard form for \BPP~is analogous to the Greibach normal form for context-free grammars~\cite{Greibach:JACM:1965}.
Every \BPP~can be effectively transformed to one in standard form
preserving bisimilarity~\cite[Proposition 2.31]{Christensen:PhD:1993},
and thus the language it recognises.

\begin{example}
    \label{ex:BPP}
    Consider two input symbols $\Sigma = \set{a, b}$
    and two nonterminals $N = \set {S, X}$.
    The following is a \BPP~in standard form:
    $S \to a. X, X \to a. (X \merge X) + b.\bot$.
    An example execution is
    $S \goesto a X \goesto a X \merge X \goesto b \bot \merge X \goesto b \bot \merge \bot$,
    and thus $a^2b^2 \in \lang S$.
    %
\end{example}

\noindent
While language equivalence is undecidable for \BPP~\cite[Sec.~5]{Hirshfeld:CSL:1994},
the finer bisimulation equivalence is decidable~\cite{ChristensenHirshfeldMoller:CONCUR:1993},
and in fact PSPACE-complete~\cite{Srba:STACS:2002,Jancar:LICS:2003}.
These initial results have motivated a rich line of research investigating decidability and complexity
for variants of bisimulation equivalence.
We consider another classical variation on language equivalence,
namely multiplicity equivalence,
and apply it to decide language equivalence of unambiguous~\BPP.
We show in~\cref{thm:UBPP language equivalence} that both problems are decidable and in \TWOEXPSPACE.
This is obtained by considering a more general model, introduced next.

\subsection{Weighted basic parallel processes}

\subparagraph{Preliminaries.}
\label{sec:series}

Let $\Sigma^* \to \Q$  be the set of \emph{(non-commutative) series} with coefficients in $\Q$,
also known as \emph{weighted languages}.
An alternative notation is $\ncpowerseries \Q \Sigma$.
We write a series as $f = \sum_{w \in \Sigma^*} f_w \cdot w$,
where the value of $f$ at $w$ is $f_w \in \Q$.
Thus, $3aba - \frac 5 2 bc$ and $1+a+a^2 + \cdots$ are series.
The set of series carries the structure of a vector space over $\Q$,
with element-wise scalar product $c \cdot f$ ($c \in \Q$) and sum $f + g$.
The \emph{support} of a series $f$
is the subset of its domain $\support f \subseteq \Sigma^*$ where it evaluates to a nonzero value.
\emph{Polynomials} $\ncpoly \Q \Sigma$ are  series with finite support.
%
%
%
The \emph{characteristic series} of a language $L \subseteq \Sigma^*$
is the series that maps words in $L$ to $1$ and all the other words to $0$.
%




For two words $u \in \Sigma^m$ and $v \in \Sigma^n$,
let $u \shuffle v$ be the multiset of all words $w = a_1 \cdots a_{m+n}$
\st~the set of indices $\set{1, \dots, m+n}$ can be partitioned into two subsequences
$i_1 < \cdots < i_m$ and $j_1 < \cdots < j_n$
\st~$u = a_{i_1} \cdots a_{i_m}$ and $v = a_{j_1} \cdots a_{j_m}$.
%
%
The multiset semantics preserves multiplicities,
e.g., $ab \shuffle a = \multiset{aab, aab, aba}$. 
The \emph{shuffle} of two series $f, g$ is the series $f \shuffle g$ defined as
$(f \shuffle g)_w := \sum_{w \in u \shuffle v} f_u \cdot g_v$,
for every $w \in \Sigma^*$,
where the sum is taken with multiplicities.
Shuffle product (called \emph{Hurwitz product} in~\cite{Fliess:1974})
leads to the commutative \emph{ring of shuffle series}
$\tuple{\ncpowerseries \Q \Sigma; +, \shuffle, 0, 1}$,
whose \emph{shuffle identity} $1$ is the series mapping $\e$ to $1$ and all other words to $0$.
A series $f$ has a \emph{shuffle inverse} $g$,
i.e., $f \shuffle g = 1$,
iff $f_\e \neq 0$.
The \emph{$n$-th shuffle power} $\shufflepower f n$ of a series $f$
is inductively defined by $\shufflepower f 0 := 1$ and $\shufflepower f {(n+1)} := f \shuffle \shufflepower f n$.
%
%

Consider the mapping $\derive {} : \Sigma^* \to \ncpowerseries \Q \Sigma \to \ncpowerseries \Q \Sigma$
\st~for every $u \in \Sigma^*$ and $f \in \ncpowerseries \Q \Sigma$,
$\derive u f \in \ncpowerseries \Q \Sigma$ is the series defined as
%
    $(\derive u f)_w = f_{uw}$, for every $w \in \Sigma^*$.
%
We call $\derive u f$ the \emph{$u$-derivative} of $f$
(\aka~\emph{shift} or \emph{left-quotient}).
%
For example, $\derive a (ab + c) = b$.
The derivative operation $\derive u$ is linear, for every $u \in \Sigma^*$.
The one-letter derivatives $\derive a$'s are (noncommuting) derivations of the shuffle ring
since they satisfy~\cref{eq:Leibniz rule},
\begin{align}
    \label{eq:derive shuffle}
    \derive a (f \shuffle g) = \derive a f \shuffle g + f \shuffle \derive a g,
    \quad \text{for all } a \in \Sigma, f, g \in \ncpowerseries \Q \Sigma,
\end{align}

\subparagraph{Syntax and semantics.}

A \emph{weighted basic parallel process} (\WBPP)
is a tuple $P = \tuple{\Sigma, N, S, F, \Delta}$
where $\Sigma$ is a finite input alphabet of \emph{terminal symbols}/\emph{actions},
$N$ is a finite set of \emph{nonterminal symbols}/\emph{processes},
$S \in N$ is the \emph{initial nonterminal},
$F : N \to \Q$ assigns a \emph{final weight} $F X \in \Q$ to each nonterminal $X \in N$,
and $\Delta : \Sigma \times N \to \poly \Q N$
is a \emph{transition function}
mapping a nonterminal $X \in N$ and an input symbol $a \in \Sigma$
to a polynomial $\Delta_a X \in \poly \Q N$.

\begin{example}
    \label{ex:from BPP to WBPP}
    A \BPP~in standard form is readily converted to a \WBPP~with $0, 1$ weights:
    The \BPP~from~\cref{ex:BPP} yields the \WBPP\-
    with output function $FS = FX = 0$ and transitions
    $\Delta_a S = X, \Delta_a X = X^2, \Delta_b S = 0, \Delta_b X = 1$.
    Configurations reachable from $S, X$ are of the form $cX^n$ ($c \in \N$).
    Action ``$a$'' acts as an increment $X^n \goesto a nX^{n+1}$
    and ``$b$'' as a decrement $X^n \goesto b nX^{n-1}$.
    The constant coefficient $c \in \N$ in a reachable configuration $cX^n$
    keeps track of the ``multiplicity'' of reaching this configuration,
    i.e., the number of distinct runs leading to it.
    For instance, $\sem S_{a^2b^2} = 2$ since
        $S \goesto a X \goesto a X^2 \goesto b 2X \goesto b 2$.
    %
    In the underlying \BPP,
    \begin{center}
        \begin{tikzpicture}[node distance = 3ex, thick, every label/.append style={font=\scriptsize}]]%
            \node (1) {$S$};

            \node (2) [right=of 1] {$X$};
            \draw[->] (1) -- node [midway, above] {$a$} (2);

            \node (3) [right=of 2] {$X \merge X$};
            \draw[->] (2) -- node [midway, above] {$a$} (3);

            \node (4) [right=of 3, yshift=3ex] {$\bot \merge X$};
            \draw[->] (3) -- node [midway, above] {$b$} (4);
            \node (5) [right=of 4] {$\bot \merge \bot$};
            \draw[->] (4) -- node [midway, above] {$b$} (5);

            \node (6) [right=of 3, yshift=-3ex] {$X \merge \bot$};
            \draw[->] (3) -- node [midway, below] {$b$} (6);
            \node (7) [right=of 6] {$\bot \merge \bot$};
            \draw[->] (6) -- node [midway, above] {$b$} (7);
        \end{tikzpicture}
    \end{center}
    where the branching upon reading the first symbol ``$b$''
    depends on whether the first or second occurrence of $X$ reads this symbol.
\end{example}


A \emph{configuration} of a \WBPP~is a polynomial $\alpha \in \poly \Q N$.
The transition function extends uniquely to a derivation of the polynomial ring $\poly \Q N$
via linearity and \cref{eq:Leibniz rule}:
\begin{align}
    \nonumber
    &\Delta : \Sigma \times \poly \Q N \to \poly \Q N \\
    \nonumber
    &\Delta_a(c \cdot \alpha)
        = c \cdot \Delta_a(\alpha),
        && \forall a \in \Sigma, c \in \Q, \\
    \nonumber
    &\Delta_a(\alpha + \beta)
        = \Delta_a (\alpha) + \Delta_a(\beta),
        && \forall a \in \Sigma, \alpha, \beta \in \poly \Q N, \\
    \label{eq:WBPP Leibniz}
    &\Delta_a (\alpha \cdot \beta)
        = \Delta_a(\alpha) \cdot \beta + \alpha \cdot \Delta_a(\beta),
        && \forall a \in \Sigma, \alpha, \beta \in \poly \Q N.
\end{align}
For example, from configuration $X \cdot Y$ we can read $a$ and go to
$\Delta_a(X\cdot Y) = \Delta_a(X) \cdot Y + X \cdot \Delta_a(Y)$;
this is models the fact that either $X$ reads $a$ and $Y$ is unchanged, or vice versa.
%
%
The transition function is then extended homomorphically to words:
\begin{align}
    \nonumber
    &\Delta : \Sigma^* \times \poly \Q N \to \poly \Q N \\
    \label{eq:Delta ext}
    &\Delta_\e \alpha := \alpha,\ 
    \Delta_{a \cdot w} \alpha := \Delta_w(\Delta_a \alpha),
    \quad \forall (a \cdot w) \in \Sigma^*, \alpha \in \poly \Q N.
\end{align}
Sometimes we write $\alpha \goesto w \beta$ when $\beta = \Delta_w(\alpha)$.
For instance, from configuration $\alpha$
we can read $ab \in \Sigma^*$ visiting configurations
$\alpha \goesto a \Delta_a(\alpha) \goesto b \Delta_b (\Delta_a(\alpha))$.
The order of reading symbols matters:
For the transition function $\Delta_a(X) = 0$, $\Delta_b(X) = Y$, and $\Delta_a(Y) = \Delta_b(Y) = 1$,
we have $X \goesto a 0 \goesto b 0$
but $X \goesto b Y \goesto a 1$.
The \emph{semantics} of a \WBPP~is the mapping
\begin{align}
    \nonumber
    &\sem \_ : \poly \Q N \to \series \Q \Sigma \\
    &\sem \alpha_w := F (\Delta_w \alpha),
        \quad \forall \alpha \in \poly \Q N, w \in \Sigma^*.
\end{align}
Here $F$ is extended homomorphically from nonterminals to configurations:
$F(\alpha + \beta) = F(\alpha) + F(\beta)$ and $F(\alpha \cdot \beta) = F(\alpha) \cdot F(\beta).$
We say that configuration $\alpha$ \emph{recognises} the series $\sem \alpha$.
The series recognised by a \WBPP~is the series recognised by its initial nonterminal.
A \emph{\WBPP~series} is a series which is recognised by some \WBPP.



\begin{example}
    \label{ex:WBPP nonregular support}
    We show a \WBPP~series which is not a \WFA~series.
    In particular, its support is nonregular support since \WFA~supports include the regular languages.
    Consider the \WBPP~from~\cref{ex:from BPP to WBPP}.
    The language $L := \support {\sem S} \cap a^* b^*$
    is the set of words of the form $a^n b^n$, which is not regular,
    and thus $\support {\sem S}$ is not regular either.
    Moreover, $\sem S$ is not a \WFA~series:
    \begin{inparaenum}[1)]
        \item the set $M$ of words of the form $a^m b^n$ with $m \neq n$ is a \WFA~support,
        \item if a language and its complement are \WFA~supports,
        then they are regular by a result of Restivo and Reutenauer~\cite[Theorem 3.1]{RestivoReutenauer:JCSS:1984}, and
        \item since $M$ is not regular, it follows that its complement is not a \WFA~support,
        and thus $L = (\Sigma^* \setminus M) \cap a^* b^*$ is not a \WFA~support either.
    \end{inparaenum}
\end{example}

\subsection{Basic properties}

We present some basic properties of the semantics of \WBPP.
First of all, applying the derivative $\derive w$ to the semantics
corresponds to applying $\Delta_w$ to the configuration.
\begin{restatable}[Exchange]{lemma}{WBPPexchangeProperty}
    \label{lem:exchange}
    For every $\alpha \in \poly \Q N$ and $w \in \Sigma^*$,
    $\derive w \sem \alpha = \sem {\Delta_w \alpha}$.
\end{restatable}
\noindent
As a consequence, the semantics is a homomorphism from configurations to series.
\begin{restatable}[Homomorphism]{lemma}{WBPPhomomorphismProperty}
    \label{lem:WBPP semantics is a homomorphism}
    The semantics function $\sem \_$ is a homomorphism
    from the polynomial to the shuffle series ring:
    \begin{align*}
        &\sem \_ : \tuple{\poly \Q N; +, \cdot} \to \tuple{\series \Q N; +, \shuffle} \\
        &\sem{c \cdot \alpha} = c \cdot \sem \alpha,\quad
            \sem{\alpha + \beta} = \sem \alpha + \sem \beta,\quad
                \sem{\alpha \cdot \beta} = \sem \alpha \shuffle \sem \beta.
    \end{align*}
\end{restatable}
%
%
%
%
%
\noindent
\Cref{lem:exchange,lem:WBPP semantics is a homomorphism} illustrate the interplay between the syntax and semantics of \WBPP,
and they can be applied to obtain some basic closure properties for the class of \WBPP~series.
\begin{restatable}[Closure properties]{lemma}{WBPPclosureProperties}
    \label{lem:WBPP basic closure properties}
    Let $f, g \in \series \Q \Sigma$ be \WBPP~series.
    The following series are also \WBPP:
    $c \cdot f$,
    $f + g$,
    $f \shuffle g$,
    $\derive a f$,
    the shuffle inverse of $f$ (when defined).
\end{restatable}
%
%
\noindent
\WBPP~series generalise the rational series
(i.e., recognised by finite weighted automata~\cite{BerstelReutenauer:CUP:2010}),
which in fact correspond to \WBPP~with a linear transition relation.

\begin{example}
    The shuffle of two \WBPP~series with context-free support
    can yield a \WBPP~series with non-context-free support.
    Consider the \WBPP~from~\cref{ex:from BPP to WBPP} over $\Sigma = \set {a, b}$.
    %
    Make a copy of this \WBPP~over a disjoint alphabet $\Gamma = \set {c, d}$ with nonterminals $\set{T, Y}$.
    Now consider the shuffle $f := \sem S \shuffle \sem T \in \series \Q {\Sigma \cup \Gamma}$.
    It is \WBPP~recognisable by \cref{lem:WBPP basic closure properties}.
    (For instance we can add a new initial nonterminal $U$ with rules
    %
    $\Delta_a U = X \cdot T$, $\Delta_c U = S \cdot Y$, and $\Delta_b U = \Delta_d U = 0$.)
    $\support f$ is not context free,
    since intersecting it with the regular language $a^* c^* b^* d^*$
    yields $\setof{a^m c^n b^m d^n}{m, n \in \N}$,
    which is not context-free by the pumping lemma for context-free languages~\cite[Theorem 7.18]{HopcroftMotwaniUllman:2000}
    (\cf~\cite[Problem 101]{AutomataBook:CUP:2023}).
\end{example}

\subsection{Differential algebra of shuffle-finite series}
\label{sec:shuffle-finite series}

Differential algebra allows us to provide an elegant characterisation of \WBPP~series.
An \emph{algebra} (over $\Q$)
is a vector space equipped with a bilinear product.
Shuffle series are a commutative algebra, called \emph{shuffle series algebra}.
A subset of $\series \Q \Sigma$ is a \emph{subalgebra}
if it contains $\Q$ and is closed under scalar product, addition, and shuffle product.
It is \emph{differential} if it is closed under derivations $\derive a$ ($a \in \Sigma$).
By \cref{lem:WBPP basic closure properties},
\WBPP~series are a differential subalgebra.
Let $\poly \Q {f^{(1)}, \dots, f^{(k)}} \subseteq \series \Q \Sigma$
be the smallest subalgebra containing $f^{(1)}, \dots, f^{(k)} \in \series \Q \Sigma$.
Algebras of this form are called \emph{finitely generated}.
%
A series 
is \emph{shuffle finite}
if it belongs to a finitely generated differential subalgebra of shuffle series.
\begin{restatable}{theorem}{WBPPequalsShuffleFinite}
    \label{thm:WBPP equals shuffle finite}
    A series is shuffle finite iff it is \WBPP.
\end{restatable}
\noindent
The characterisation above provides an insight into the algebraic structure of \WBPP~series.
%
%
Other classes of series can be characterised in a similar style.
For instance, a series is accepted by a \WFA\-
iff it belongs to a finitely generated differential vector space over $\Q$
\cite[Proposition 5.1]{BerstelReutenauer:CUP:2010};
by a weighted context-free grammar iff
it belongs to a $\derive a$-closed, finitely generated subalgebra of the algebra of series with (noncommutative) \emph{Cauchy product}
($(f \cauchy g)_w := \sum_{w = u \cdot v} f_u \cdot f_v$);
and by a polynomial automaton~\cite{BenediktDuffSharadWorrell:PolyAut:2017} iff its reversal
($f^R_{a_1\dots a_n} := f_{a_n \cdots a_1}$)
belongs to a $\derive a$-closed, finitely generated subalgebra of the algebra of series with \emph{Hadamard product}
($(f \hadamard g)_w := f_w \cdot g_w$).
Considering other products yields novel classes of series, too.
For instance, the \emph{infiltration product}~\cite{BasoldHansenPinRutten:MSCS:2017}
yields the class of series that belong to a $\derive a$-closed, finitely generated subalgebra of the algebra of series with infiltration product.




%% file: 03-zeroness.tex
\subsection{Equivalence and zeroness problems}

The \emph{\WBPP~equivalence problem} takes in input two \WBPP~$P, Q$
and amounts to determine whether $\sem P = \sem Q$.
In the special case where $\sem Q = 0$, we have an instance of the \emph{zeroness problem}.
Since \WBPP~series form an effective vector space, equivalence reduces to zeroness,
and thus we concentrate on the latter.

\subparagraph{Evaluation and word-zeroness problems.}

We first discuss a simpler problem,
which will be a building block in our zeroness algorithm.
The \emph{evaluation problem} takes in input a \WBPP~with initial configuration $\alpha$~and a word $w \in \Sigma^*$,
and it amounts to compute $\sem \alpha _w$.
The \emph{word-zeroness} problem takes the same input,
and it amounts to decide whether $\sem \alpha_w = 0$.
\begin{theorem}
    \label{thm:WBPP word evaluation and zeroness problems}
    The evaluation and word-zeroness problems for \WBPP~are in \PSPACE.
\end{theorem}
The proof follows from the following three ingredients:
The construction of an algebraic circuit of exponential size computing the polynomial $\Delta_w \alpha$ (\cref{lem:multi-step circuit}),
the fact that this polynomial has polynomial degree (\cref{lem:small degree property}),
and the fact that circuits computing multivariate polynomials of polynomial degree can be evaluated in \NC~\cite[Theorem 2.4.5]{Mittmann:PhD:2013}.

\begin{restatable}{lemma}{WBPPmultistepCircuit}
    \label{lem:multi-step circuit}
    Fix a word $w \in \Sigma$ and an initial configuration $\alpha \in \poly \Q N$ of a \WBPP,
    where $\alpha, \Delta_a X_i \in \poly \Q N$ are the outputs of an algebraic circuit of size $n$.
    We can construct an algebraic circuit computing $\Delta_w \alpha$ of size $\leq 4^{\length w} \cdot n$.
    The construction can be done in space polynomial in $\length w$ and logarithmic in $n$.
\end{restatable}

\begin{restatable}{lemma}{SmallDegreeProperty}
    \label{lem:small degree property}
    Let $D \in \N$ be the maximum of the degree of the transition relation $\Delta$
    and the initial configuration $\alpha$.
    The configuration $\Delta_w \alpha \in \poly \Q N$
    reached by reading a word $w \in \Sigma^n$ of length $n$
    has total degree $\bigO {n \cdot D}$.
\end{restatable}

\subparagraph{Decidability of the zeroness problem.}

Fix a \WBPP~and a configuration $\alpha \in \poly \Q N$.
Suppose we want to decide whether $\sem \alpha$ is zero.
An algorithm for this problem follows from first principles.
Recall that an \emph{ideal} $I \subseteq \poly \Q N$ is a subset
closed under addition, and multiplication by arbitrary polynomials~\cite[§4, Ch.~1]{CoxLittleOShea:Ideals:2015}.
Let $\ideal S$ be the smallest ideal including $S \subseteq \poly \Q N$.
Intuitively, this is the set of ``logical consequences'' of the vanishing of polynomials in $S$.
Build a chain of polynomial ideals
\begin{align}
    \label{eq:WBPP ideal chain}
    I_0 \subseteq I_1 \subseteq \cdots \subseteq \poly \Q N,
    \ \text{with } I_n := \idealof {\Delta_w \alpha} {w \in \Sigma^{\leq n}}, n \in \N.
\end{align}
Intuitively, $I_n$ is the set of polynomials that vanish
as a consequence of the vanishing of $\Delta_w \alpha$ for all words $w$ of length $\leq n$.
The chain above has some important structural properties,
essentially relying on the fact that the $\Delta_a$'s are derivations of the polynomial ring.
\begin{restatable}{lemma}{LemmaDeltaOfIdeal}
    \label{lem:Delta of ideal}
    %
    \begin{inparaenum}[1.]
        \item $\Delta_a I_n \subseteq I_{n+1}$.
        \item $I_{n+1} = I_n + \ideal{\bigcup_{a \in \Sigma} \Delta_a I_n}$.
        \item $I_n = I_{n+1}$ implies $I_n = I_{n+1} = I_{n+2} = \cdots$.
    \end{inparaenum}
\end{restatable}
\noindent
By Hilbert's finite basis theorem~\cite[Theorem 4, §5, Ch.~2]{CoxLittleOShea:Ideals:2015},
there is $M \in \N$ \st~$I_M = I_{M+1} = \cdots$.
By~\cref{lem:Delta of ideal}~(3) and decidability of ideal inclusion~\cite{Mayr:STACS:1989},
$M$ can be computed.
This suffices to decide \WBPP~zeroness.
Indeed, let $\Delta_{w_1} \alpha, \dots, \Delta_{w_m} \alpha$ be the generators of $I_M$.
For every input word $w \in \Sigma^*$
there are $\beta_1, \dots, \beta_m \in \poly \Q N$ \st\-
$\Delta_w \alpha = \beta_1 \cdot \Delta_{w_1} \alpha + \cdots \beta_m \cdot \Delta_{w_m} \alpha$.
By applying the output function $F$ on both sides, we have
$\sem \alpha_w = F (\Delta_w \alpha) =  F \beta_1 \cdot \sem \alpha_{w_1} + \cdots + F \beta_m \cdot \sem \alpha_{w_m}$.
It follows that if $\sem \alpha_w = 0$ for all words of length $\leq M$, then $\sem \alpha = 0$.
One can thus enumerate all words $w$ of length $\leq M$
and check $\sem \alpha_w = 0$ with~\cref{thm:WBPP word evaluation and zeroness problems}.
%
%
So far we only know that $M$ is computable.
In the next section we show that in fact $M$ is an elementary function of the input \WBPP.

\subparagraph{Elementary upper bound for the zeroness problem.}

We present an elementary upper bound on the length of the chain of polynomial ideals~\cref{eq:WBPP ideal chain}.
This is obtained by generalising the case of a single derivation from Novikov and Yakovenko~\cite[Theorem 4]{NovikovYakovenko:1999}
to the situation of several, not necessarily commuting derivations $\Delta_a$, $a \in \Sigma$.
The two main ingredients in the proof of~\cite[Theorem 4]{NovikovYakovenko:1999} are
1) a structural property of the chain~\cref{eq:WBPP ideal chain} called \emph{convexity}, and
2) a degree bound on the generators of the $n$-th ideal $I_n$
(which we have already established in \cref{lem:small degree property}).
For two sets $I, J \subseteq \poly \Q N$ consider the \emph{colon set}
$I : J := \setof {f \in \poly \Q N} {\forall g \in J, f \cdot g \in I}$
\cite[Def.~5, §4, Ch.~4]{CoxLittleOShea:Ideals:2015}.
%
If $I, J$ are ideals of $\poly \Q N$ then $I : J$ is also an ideal.
%
%
%
An ideal chain $I_0 \subseteq I_1 \subseteq \cdots$ is \emph{convex}
if the colon ideals $I_n : I_{n+1}$ form themselves a chain $I_0 : I_1 \subseteq I_1 : I_2 \subseteq \cdots$.
Chain of ideals obtained by iterated application of a single derivation are convex by~\cite[Lemma 7]{NovikovYakovenko:1999}.
We extend this observation to a finite set of derivations.
\begin{restatable}[\protect{generalisation of~\cite[Lemma 7]{NovikovYakovenko:1999}}]{lemma}{convexIdealChain}
    \label{lem:WBPP ideal chain is convex}
    The ideal chain~\cref{eq:WBPP ideal chain} is convex.
\end{restatable}
\begin{proof}
    We extend the argument from~\cite{NovikovYakovenko:1999} to the case of many derivations.
    Assume $f \in I_{n-1} : I_n$ and let $h \in I_{n+1}$ be arbitrary.
    We have to show $f \cdot h \in I_n$.
    \begin{claim*}
        $f \cdot \Delta_a g \in I_n$,
        for all $a \in \Sigma$ and $g \in I_n$.
    \end{claim*}
    \begin{proof}[Proof of the claim.]
        Since $\Delta_a$ is a derivation~\cref{eq:Delta ext},
        $\Delta_a (f \cdot g) = \Delta_a f \cdot g + f \cdot \Delta_a g$,
        and by solving for $f \cdot \Delta_a g$ we can write
        %
            $f \cdot \Delta_a g = \underbrace {\Delta_a (\overbrace{f \cdot g}^{\text{(a) } I_{n-1}})}_{\text{(b) }I_n} - \underbrace{\Delta_a f \cdot g}_{\text{(c) } I_n}$.
        %
        Condition (a) follows from the definition of colon ideal,
        (b) from point (1) of~\cref{lem:Delta of ideal},
        and (c) from $I_n$ being an ideal.
    \end{proof}
    Since $h \in I_{n+1}$, by point (2) of~\cref{lem:Delta of ideal},
    we can write $h = h_0 + h_1$ with $h_0 \in I_n$
    and $h_1 \in \ideal {\bigcup_{a \in \Sigma} \Delta_a I_n}$.
    In particular, $h_1 = \sum_i p_i \cdot \Delta_{a_i} g_i$ with $g_i \in I_n$,
    By the claim, $f \cdot h_1 = \sum_i p_i \cdot f \cdot \Delta_{a_i} g_i \in I_n$.
    Consequently, $f \cdot h = f \cdot h_0 + f \cdot h_1 \in I_n$ as well.
\end{proof}
\noindent
Thanks to~\cref{lem:WBPP ideal chain is convex}
we can generalise the whole proof of~\cite[Theorem 4]{NovikovYakovenko:1999},
eventually arriving at the following elementary bound.
The \emph{order} of a \WBPP~is the number of nonterminals
and its \emph{degree} is the maximal degree of the polynomials $\Delta_a X$ ($a \in \Sigma, X \in N$).
\begin{restatable}{theorem}{NovikovYakovenko}
    \label{thm:chain length bound}
    %
    Consider a \WBPP~of order $\leq k$ and degree $\leq D$.
    The length of the ideal chain~\cref{eq:WBPP ideal chain} is at most $D^{k^{\bigO{k^2}}}$.
\end{restatable}
\noindent
%
The elementary bound above may be of independent interest.
Already in the case of a single derivation, it is not known whether the bound from~\cite{NovikovYakovenko:1999} is tight,
albeit it is expected not to be so.
%
We provide a proof sketch of~\cref{thm:chain length bound}
in order to illustrate the main notions from algebraic geometry which are required.
The full proof is presented in~\cref{sec:ideal chain complexity}.
\begin{proof}[Proof sketch.]
    We recall some basic facts from algebraic geometry.
    %
    %
    The \emph{radical} $\rad I$ of an ideal $I$
    is the set of elements $r$ \st~$r^m \in I$ for some $m \in \N$;
    note that $\rad I$ is itself an ideal.
    An ideal $I$ is \emph{primary} if $p \cdot q \in I$ and $p \not\in I$ implies $q \in \rad I$.
    %
    %
    A \emph{primary decomposition} of an ideal $I$
    is a collection of primary ideals $\set{Q_1, \dots, Q_s}$,
    called \emph{primary components},
    \st~$I = Q_1 \cap \cdots \cap Q_s$.
    The \emph{dimension} $\dim I$ of a polynomial ideal $I \subseteq \poly \Q k$
    is the dimension of its associated variety
    $V(I) = \setof {x \in \C^k} {\forall p \in I. p(x) = 0}$.
    Since the operation of taking the variety of an ideal is inclusion-reversing,
    ideal inclusion is dimension-reversing:
    $I \subseteq J$ implies $\dim I \geq \dim J$.
    Consider a convex chain of polynomial ideals as in \cref{eq:WBPP ideal chain}.
    By convexity, the colon ideals also form a chain
    $I_0 \colon I_1 \subseteq I_1 \colon I_2 \subseteq \cdots \subseteq \poly \Q k$.
    The colon dimensions are at most $k$ and non-increasing,
    $k \geq \dim {(I_0 \colon I_1)} \geq \dim {(I_1 \colon I_2)} \geq \cdots$.
    Divide the original ideal chain \cref{eq:WBPP ideal chain} into segments,
    where in the $i$-th segment the colon dimension is a constant $m_i$:
    \begin{align}
        \label{eq:i-th segment}
        \underbrace {I_0 \subseteq \cdots \subseteq I_{n_0 - 1}}_{\dim {(I_n \colon I_{n+1})} = m_0} \subseteq 
        \underbrace {I_{n_0} \subseteq \cdots \subseteq I_{n_1 - 1}}_{\dim {(I_n \colon I_{n+1})} = m_1} \subseteq
        \cdots \subseteq
        \underbrace {I_{n_i} \subseteq \cdots \subseteq I_{n_{i+1} - 1}}_{\dim {(I_n \colon I_{n+1})} = m_i} \subseteq \cdots.
    \end{align}
    Since the colon dimension can strictly decrease at most $k$ times,
    there are at most $k$ segments.
    In the following claim we show that
    the length of a convex ideal chain with equidimensional colon ideal chain
    can be bounded by the number of primary components of the initial ideal.
    \begin{restatable}[\protect{\cite[Lemmas 8+9]{NovikovYakovenko:1999}}]{claim}{claimA}
        \label{claim:NY:Lemma 8+9}
        Consider a strictly ascending convex chain of ideals
        $I_0 \subsetneq I_1 \subsetneq \cdots \subsetneq I_\ell$ of length $\ell$
        where the colon ratios have the same dimension
        $m := \dim {(I_0 \colon I_1)} = \cdots = \dim {(I_{\ell - 1} \colon I_\ell)}$.
        Then $\ell$ is at most the number of primary components 
        of any primary ideal decomposition of the initial ideal $I_0$ (counted with multiplicities%
        \footnote{We refer to \cite[Sec.~4.1]{NovikovYakovenko:1999}
        for the notion of \emph{multiplicity} of a primary component.}%
        ).
    \end{restatable}
    \noindent
    We apply~\cref{claim:NY:Lemma 8+9} to the $i$-th segment~\cref{eq:i-th segment}
    and obtain that its length $\ell_i := n_{i+1} - n_i$ is at most the number of primary components
    in any primary ideal decomposition of its starting ideal $I_{n_i}$.
    We now use a result from effective commutative algebra
    showing that we can compute primary ideal decompositions of size bounded by the degree of the generators.
    \begin{restatable}[\protect{variant of \cite[Corollary 2]{NovikovYakovenko:1999}}]{claim}{claimB}
        \label{lem:NY variant:Corollary 4}
        An ideal $I \subseteq \poly \C k$ generated by polynomials of degree $\leq D$
        admits a primary ideal decomposition of size $D^{k^{\bigO k}}$ (counted with multiplicities).
    \end{restatable}
    \noindent
    By \cref{lem:NY variant:Corollary 4}, $I_{n_i}$ admits some primary decomposition of size $d_i^{k^{\bigO k}}$,
    where $d_i$ is the maximal degree of the generators of $I_{n_i}$.
    By~\cref{lem:small degree property}, $d_i$ is at most $\bigO {D \cdot n_i}$.
    All in all, the $i$-th segment has length 
    $\ell_i = n_{i+1} - n_i \leq (D \cdot n_i)^{k^{\bigO k}}$.
    We have $n_i \leq \bigO {f_i}$ where $f_i$ satisfies $f_{i+1} \leq a \cdot f_i^b$
    with $a = D^b$ and $b = k^{\bigO k}$.
    Thus $f_k \leq a \cdot a^b \cdots a^{b^{k-1}} \leq a^{b^{\bigO k}}$,
    yielding the required upper bound on the length of the ideal chain
    $n_k \leq D^{k^{\bigO {k^2}}}$.
\end{proof}
Thanks to the bound from~\cref{thm:chain length bound},
we obtain the main contribution of the paper, which was announced in the introduction.%
\WBPPmain*
\begin{proof}
    The bound on the length of the ideal chain~\cref{eq:WBPP ideal chain} from~\cref{thm:chain length bound}
    implies that if the \WBPP~is not zero,
    then there exists a witnessing input word of length at most doubly exponential.
    We can guess this word and verify its correctness in \TWOEXPSPACE~by \cref{thm:WBPP word evaluation and zeroness problems}.
    This is a nondeterministic algorithm,
    but by courtesy of Savitch's theorem~\cite{Savitch:JCSS:1970} we obtain a \emph{bona fide} deterministic \TWOEXPSPACE~algorithm.
\end{proof}

\subparagraph{Application to \BPP.}

The \emph{multiplicity semantics} of a \BPP~is its series semantics as an $\N$-\WBPP.
%
Intuitively, one counts all possible ways in which an input is accepted by the model.
The \emph{\BPP~multiplicity equivalence problem} takes as input two \BPP~$P, Q$
and returns ``yes'' iff $P, Q$ have the same multiplicity semantics.
Decidability of \BPP~multiplicity equivalence readily follows from~\cref{thm:main}.
%
%
We say that a \BPP~is \emph{unambiguous} if its multiplicity semantics is $\set{0, 1}$-valued.
While \BPP~language equivalence is undecidable~\cite[Sec.~5]{Hirshfeld:CSL:1994},
we obtain decidability for unambiguous \BPP.
We have thus proved~\cref{thm:BPP multiplicity equivalence}.
This generalises decidability for deterministic~\BPP,
which follows from decidability of bisimulation equivalence~\cite{ChristensenHirshfeldMoller:CONCUR:1993}.
%
%
Language equivalence of unambiguous context-free grammars, the sequential counterpart of \BPP\-
(sometimes called \BPA~in process algebra),
is a long-standing open problem, as well as the more general multiplicity equivalence problem
(\cf~\cite{ForejtJancarKieferWorrell:IC:2014,Clemente:EPTCS:2020,BalajiClementeNosanShirmohammadiWorrell:LICS:2023}).

%% file: 04-CDF.tex

\section{Constructible differentially finite power series}
\label{sec:CDF}

In this section we study a class of multivariate power series in commuting variables
called \emph{constructible differentially finite} (\CDF)~\cite{BergeronReutenauer:EJC:1990,BergeronSattler:TCS:1995}.
We show that \CDF~power series arise naturally as the commutative variant of \WBPP~series from~\cref{sec:WBPP}.
Stated differently, the novel \WBPP~can be seen as the noncommutative variant of \CDF,
showing a connection between the theory of weighted automata and differential equations.
As a consequence, by specialising to the commutative context the \TWOEXPSPACE~\WBPP~zeroness procedure,
we obtain an algorithm to decide zeroness for \CDF~power series in \TWOEXPTIME.
This is the main result of the section,
which was announced in the introduction (\cref{thm:CDF zeroness in TWOEXPTIME}).

On the way, we recall and further develop the theory of \CDF~power series.
In particular, we provide a novel closure under regular support restrictions (\cref{lem:CDF closure under regular cardinality restrictions}).
In~\cref{sec:species} we illustrate a connection between \CDF~power series and combinatorics,
by showing that the generating series of a class of constructible species of structures are \CDF,
which will broaden the applicability of the \CDF~zeroness algorithm to multiplicity equivalence of species.

\subparagraph{Preliminaries.}
In the rest of the section, we consider commuting variables
$x = \tuple{x_1, \dots, x_d}$, $y = \tuple{y_1, \dots, y_k}$.
We denote by $\powerseries \Q x$ the set of multivariate power series in $x$,
endowed with the structure of a commutative ring $\tuple{\powerseries \Q x; +, \cdot, 0, 1}$
with pointwise addition and (Cauchy) product.
%
The partial derivatives $\partial {x_j}$'s satisfy~\cref{eq:Leibniz rule},
and thus form a family of commuting derivations of this ring.
To keep notations compact, we use vector notation:
For a tuple of naturals $n = \tuple{n_1, \dots, n_d} \in \N^d$,
define $n! := n_1! \cdots n_d!$,
$x^n := x_1^{n_1} \cdots x_d^{n_d}$,
and $\partial x^n := \partial {x_1}^{n_1} \cdots \partial {x_d}^{n_d}$.
We write a power series as $f = \sum_{n \in \N^d} f_n \cdot \frac {x^n} {n!} \in \powerseries \Q x$,
and define the \emph{(exponential) coefficient extraction} operation $\coefficient {x^n} f := f_n$, for every $n \in \N^d$.
This is designed in order to have the following simple commuting rule with partial derivative:
\begin{align}
    \label{eq:coefficient extraction and derivation}
    \coefficient {x^m} (\partial x^n f) = \coefficient {x^{m + n}} f,
    \quad \text{for all } m, n \in \N^d.
\end{align}
\noindent
Coefficient extraction is linear,
and \emph{constant term} extraction $\coefficient {x^0}$ is even a homomorphism 
since $\coefficient {x^0} (f \cdot g) = \coefficient {x^0} f \cdot \coefficient {x^0} g$.
%
%
The \emph{Jacobian matrix} of a tuple of power series
$f = \tuplesmall{f^{(1)}, \dots, f^{(k)}} \in \powerseries \Q x^k$
is the matrix $\partial x f \in \matrices k d {\powerseries \Q x}$
where entry $(i, j)$ is $\partial {x_j} f^{(i)}$.
Consider commuting variables $y = \tuple{y_1, \dots, y_k}$.
For a set of indices $I \subseteq \set{1, \dots, k}$,
by $y_I$ we denote the tuple of variables $y_i$ \st~$i \in I$ and
by $y_{\setminus I}$ we denote the tuple of variables $y_i$ \st~$i \not\in I$.
A power series $f \in \powerseries \Q y$ is \emph{locally polynomial \wrt~$y_I$}
if $f \in \powerseries {\poly \Q {y_I}} {y_{\setminus I}}$
($f$ is a power series in $y_{\setminus I}$ with coefficients polynomial in $y_I$),
and that it is \emph{polynomial \wrt~$y_I$}
if $f \in \poly {\powerseries \Q {y_{\setminus I}}} {y_I}$
($f$ is a polynomial in $y_I$ with coefficients which are power series in $y_{\setminus I}$).
For instance $\frac 1 {1 - y_1 \cdot y_2} = 1 + y_1 y_2 + (y_1 y_2)^2 + \cdots$ is not polynomial,
but it is locally polynomial in $y_{\set 1}$ (and $y_{\set 2}$).
A power series $f \in \powerseries \Q {x, y}$
and a tuple $g = \tuplesmall{g^{(1)}, \dots, g^{(k)}} \in \powerseries \Q x^k$
are \emph{$y$-composable} if $f$ is locally polynomial \wrt~$y_I$,
where $I$ is the set of indices $i$ \st~$g^{(i)}(0) \neq 0$;
\emph{strongly $y$-composable} is obtained by replacing ``locally polynomial'' with ``polynomial''.
As a corner case often arising in practice, $f, g$ are always strongly $y$-composable when $g(0) = 0$.
When $f, g$ are $y$-composable, their \emph{composition}
$f \compose y g \in \powerseries \Q x$
obtained by replacing $y_i$ in $f$ with $g^{(i)}$, for every $1 \leq i \leq k$, exists.
Composition extends component-wise to vectors and matrices.
%
%
%
%
%

\subsection{Multivariate \CDF~power series}

    A power series $f^{(1)} \in \powerseries \Q x$ is \CDF~\cite{BergeronReutenauer:EJC:1990,BergeronSattler:TCS:1995}
    if it is the first component of a solution $f = \tuplesmall{f^{(1)}, \dots, f^{(k)}} \in \powerseries \Q x^k$
    of a system of polynomial partial differential equations
    \begin{align}
        \label{eq:multivariate CDF - matrix form}
        \partial x f = P \compose y f,
            \quad\text{where } P \in \matrices k d {\poly \Q {x, y}}.
    \end{align}
%
We call $k$ the \emph{order} of the system and $d$ its \emph{dimension};
in the univariate case $d = 1$, \cref{eq:multivariate CDF - matrix form} is a system of ordinary differential equations.
The matrix $P$ is called the \emph{kernel} of the system.
The \emph{degree} the system is the maximum degree of polynomials in the kernel,
and so it is its \emph{height}.
When the kernel does not contain $x$ the system is called \emph{autonomous}, otherwise \emph{non-autonomous}.
There is no loss of expressive power in considering only autonomous systems.
Many analytic functions give rise to univariate \CDF~power series,
such as polynomials, the exponential series $f := e^x = 1 + x + x^2/2! + \cdots$ (since $\partial x f = f$),
the trigonometric series $\sin x$, $\cos x$, $\sec x := 1 / \cos x$, $\arcsin, \arccos, \arctan$
their hyperbolic variants $\sinh, \cosh, \tanh$, $\sech = 1 / \cosh$, $\arsinh$, $\artanh$,
the non-elementary error function $\erf(x) := \int_0^x e^{-t^2} dt$
(since $\partial x \erf = e^{-x^2}$ and $\partial x (e^{-x^2} )= -2x \cdot e^{-x^2}$).
Multivariate \CDF~power series include polynomials, rational power series,
constructible algebraic series (in the sense of \cite[Sec.~2]{Fliess:1974};
\cite[Theorem 4]{BergeronReutenauer:EJC:1990},\cite[Corollary 13]{BergeronSattler:TCS:1995}),
and a large class of D-finite series (\cite[Lemma 6]{BergeronSattler:TCS:1995}; but not all of them).
Moreover, we demonstrate in~\cref{thm:strongly constructible species are CDF}
that the generating series of strongly constructible species are \CDF.
We recall some basic closure properties for the class of \CDF~power series.
\begin{restatable}[\protect{Closure properties; \cite[Theorem 2]{BergeronReutenauer:EJC:1990}, \cite[Theorem 11]{BergeronSattler:TCS:1995}}]{lemma}{basicClosurePropertiesCDF}
    \label{lem:CDF basic closure properties}
    \begin{inparaenum}[(1)]
        \item If $f, g \in \powerseries \Q x$ are \CDF,
        then are also \CDF:
        $c \cdot f$ for $c \in \Q$,
        $f + g$,
        $f \cdot g$,
        $\partial {x_j} f$ for $1 \leq j \leq d$,
        $1/f$ (when defined).
        \item If $\partial {x_1} f, \dots, \partial {x_d} f$ are \CDF,
        then so is $f$.
        \item Closure under strong composition:
        If $f \in \powerseries \Q {x, y}, g \in \powerseries \Q x^k$
        are strongly $y$-composable and \CDF, then $f \compose y g$ is \CDF.
    \end{inparaenum}
\end{restatable}
\begin{remark}
    In the univariate case $d = 1$,~\cite[Theorem 2]{BergeronReutenauer:EJC:1990}
    proves closure under composition under the stronger assumption $g(0) = 0$.
    In the multivariate case,~\cite[Theorem 11]{BergeronSattler:TCS:1995} claims without proof closure under composition (when defined).
    We leave it open whether \CDF~power series are closed under composition. 
\end{remark}
\noindent
Of the many pleasant closure properties above, especially composition is remarkable,
since this does not hold for other important classes of power series,
such as the algebraic and the D-finite power series.
For instance, $e^x$ and $e^x - 1$ are D-finite, but $e^{e^x - 1}$ is not~\cite[Problem 7.8]{KauerPaule:Tetrahedron:2011}.
On the other hand, \CDF~power series are not closed under \emph{Hadamard product},
already in the univariate case~\cite[Sec.~4]{BergeronReutenauer:EJC:1990}.
(The \emph{Hadamard product} of $f = \sum_{n \in \N^d} f_n \cdot x^n, g = \sum_{n \in \N^d} g_n \cdot x^n \in \powerseries \Q x$
is $f \hadamard g = \sum_{n \in \N^d} (f_n \cdot g_n) \cdot x^n$.)
Another paramount closure property regards resolution of systems of power series equations.
A system of equations of the constructible form $y = f$ with $f \in \powerseries \Q {x, y}^k$ is \emph{well posed}
if $f(0, 0) = 0$ and the Jacobian matrix 
evaluated at the origin $\partial y f(0, 0)$ is nilpotent.
A \emph{canonical solution} is a series $g \in \powerseries \Q x^k$ solving the system for $y := g(x)$ \st~$g(0) = 0$.
The following is a slight generalisation of \cite[Corollary 13]{BergeronSattler:TCS:1995}.
\begin{restatable}[Constructible power series theorem]{lemma}{constructiblePowerSeriesTheorem}
    \label{lem:CDF constructible nilpotent}
    A well-posed system of equations
    %
        $y = f(x, y)$
    %
    has a unique canonical solution $y := g(x)$.
    Moreover, if $f$ is \CDF, then $g$ is \CDF.
\end{restatable}
\noindent
For example, the unique canonical solution of the well-posed equation $y = f := x \cdot e^y$ is \CDF.

\subsection{Support restrictions}
\label{sec:CDF support restrictions}

We discuss a novel closure property for \CDF~power series,
which will be useful later in the context of combinatorial enumeration (\cref{sec:species}).
The \emph{restriction} of $f \in \powerseries \Q x$
by a \emph{support constraint} $S \subseteq \N^d$
is the series $\restrict f S \in \powerseries \Q x$
which agrees with $f$ on the coefficient of $x^n$ for every $n \in S$,
and is zero otherwise.
%
%
We introduce a small constraint language in order to express a class of support constraints.
The set of \emph{constraint expressions} of dimension $d \in \N$ is generated by the following abstract grammar,
\begin{align}
    \varphi, \psi \ ::=\ z_j = n \mid \equivmod {z_j} n m \mid \varphi \lor \psi \mid \varphi \land \psi \mid \lneg \varphi, 
\end{align}
where $1 \leq j \leq d$ and $m, n \in \N$ with $m \geq 1$.
Expressions $z_j \leq n$ and $z_j \geq n$ can be derived.
The semantics of a constraint expressions $\varphi$ of dimension $d$, written $\sem \varphi \subseteq \N^d$,
is defined by structural induction in the expected way.
For instance, the semantics of $z_1 \geq 2 \land \equivmod {z_2} 1 2$
is the set of pairs $\tuple {a, b} \in \N^2$ where $a \geq 2$ and $b$ is odd.
%
%
%
%
Call a set $S \subseteq \N^d$ \emph{regular} if it is denoted by a constraint expression.
%

\begin{restatable}
{lemma}{CDFclosureUnderRegularRestrictions}
    \label{lem:CDF closure under regular cardinality restrictions}
    \CDF~power series are closed under regular support restrictions.
\end{restatable}

\noindent
For instance, since $e^x$ is \CDF~also
$\sinh x = \restrict {e^x}{\sem{\equivmod {z_1} 1 2}}$ is \CDF.
\CDF~are not closed under more general \emph{semilinear support restrictions}.
E.g., restricting to the semilinear set $\setof{\tuple{m, \dots, m} \in \N^d}{m \in \N}$
amounts to taking the \emph{diagonal},
which in turn can be used to express the Hadamard product of power series~\cite[remark (2) on pg.~377]{Lipshitz:JA:1988},
and \CDF~are closed under none of these operations. 

\subsection{\CDF~=~Commutative \WBPP~series}

We demonstrate that \CDF~power series correspond to \WBPP~series satisfying a commutativity condition.
In particular, they coincide in the univariate case $x = \tuple{x_1}$ and $\Sigma = \set {a_1}$.
A series $f \in \series \Q \Sigma$ over a finite alphabet $\Sigma = \set{a_1, \dots, a_d}$
is \emph{commutative} if $f_u = f_v$ whenever $\Parikh u = \Parikh v$;
in this case we associate to it a power series $\stops f \in \powerseries \Q x$ in commuting variables $x = \tuple{x_1, \dots, x_d}$
by $\stops f := \sum_{n \in \N^d} f_n \cdot \frac {x^n} {n!}$
where $f_n := f_w$ for any $w \in \Sigma^*$ \st~$\Parikh w = n$.
Conversely, to any power series $f \in \powerseries \Q x$ we associate a commutative series $\pstos f \in \series \Q \Sigma$
by $\pstos f := \sum_{w \in \Sigma^*} \coefficient {x^{\Parikh w}} f \cdot w$.
These two mappings are mutual inverses
and by the following lemma we can identify $\CDF$~power series with commutative \WBPP~series,
thus providing a bridge between the theory of weighted automata and differential equations.
\begin{restatable}{lemma}{commutativeShuffleFinite}
    \label{lem:commutativeShuffleFinite}
    If $f \in \series \Q \Sigma$ is a commutative \WBPP~series,
    then $\stops f \in \powerseries \Q x$ is a \CDF~power series.
    Conversely, if $f \in \powerseries \Q x$ is a \CDF~power series,
    then $\pstos f \in \series \Q \Sigma$ is a commutative \WBPP~series.
\end{restatable}

\subsection{Zeroness of \CDF~power series}
\label{sec:CDF zeroness}

\subparagraph{Coefficient computation.}
\label{sec:CDF computing coefficients}

We provide an algorithm to compute \CDF~power series coefficients.
While a \PSPACE~algorithm follows from~\cref{thm:WBPP word evaluation and zeroness problems},
we are interested here in the precise complexity \wrt~degree, height, and order.
This will allow us obtain the improved \TWOEXPTIME~complexity for zeroness (\cref{thm:CDF zeroness in TWOEXPTIME}).

\begin{restatable}{lemma}{CDFtableLemma}
    \label{lem:CDF table}
    Given a tuple of $d$-variate \CDF~power series $f \in \powerseries \Z x^k$
    satisfying an integer system of \CDF~equations~\cref{eq:multivariate CDF - matrix form}
    of degree $D$, order $k$, height $H$,
    and a bound $N$, we can compute all coefficients $\coefficient {x^n} f \in \Z^k$
    with total degree $\onenorm n \leq N$ in deterministic time
    $\leq (N + d \cdot D + k)^{\bigO{d \cdot D + k}} \cdot (\log H)^{\bigO 1}$.
\end{restatable}
\noindent
The lemma is proved by a dynamic programming algorithm storing all required coefficients in a table,
which is feasible since numerators and denominators are not too big.
This rough estimation shows that the complexity is exponential in $d, D, k$ and \emph{polynomial} in $N$.
%

\subparagraph{Zeroness.}
\label{subsec:CDF zeroness}

The \emph{zeroness problem} for \CDF~power series takes as input
a polynomial $p \in \poly \Q y$
and a system of equations~\cref{eq:multivariate CDF - matrix form}
with an initial condition $c \in \Q^k$ extending to a (unique) power series series solution $f$ \st~$f(0) = c$,
and asks whether $p \compose y f = 0$.
\begin{remark}
    This is a \emph{promise problem}: We do not decide solvability in power series.
    In our application in \cref{sec:species}
    this is not an issue since power series solutions exist by construction.
    In the univariate case $d = 1$ the promise is always satisfied.
    We leave it as future work to investigate the problem of solvability in power series of \CDF~equations.
\end{remark}
The following lemma gives short nonzeroness witnesses. 
It follows immediately from the \WBPP~ideal construction~\cref{eq:WBPP ideal chain}.
Together with \Cref{lem:CDF table} it yields the announced~\cref{thm:CDF zeroness in TWOEXPTIME}.
\begin{restatable}{lemma}{CDFnonzeroWitnessBound}
    \label{lem:CDF nonzero witness degree bound}
    Consider a \CDF~$f \in \powerseries \Q x^k$
    and $p \in \poly \Q y$, both of degree $\leq D$.
    The power series $g := p \compose y f$ is zero iff
    $\coefficient {x^n} g = 0$ for all monomials $x^n$
    of total degree $\onenorm n \leq D^{k^{\bigO {k^2}}}$.
\end{restatable}

%% file: 05-species.tex
\section{Constructible species of structures}
\label{sec:species}

The purpose of this section is to show how a rich combinatorial framework for building classes of finite structures (called \emph{species})
gives rise in a principled way to a large class of \CDF~power series.
The main result of this section is that multiplicity equivalence is decidable for a large class of species
(\cref{thm:multiplicity equivalence of constructible species}).
\emph{Combinatorial species of structures}~\cite{Joyal:AM:1981}
are a formalisation of combinatorics based on category theory,
designed in such a way as to expose a bridge between combinatorial operations on species
and corresponding algebraic operations on power series.
Formally, a \emph{$d$-sorted species} is a $d$-ary endofunctor $\FF$ in the category of finite sets and bijections.
In particular, $\FF$ defines a mapping from $d$-tuples of finite sets $U = \tuple{U_1, \dots, U_d}$ to a finite set $\FF[U]$,
satisfying certain naturality conditions which ensure that $\FF$ is independent of the names of the elements of $U$.
In particular, the cardinality of the output $\card {\FF[U_1, \dots, U_d]}$
depends only on the cardinality of the inputs $\card {U_1}, \dots, \card{U_d}$,
which allows one to associate to $\FF$ the \emph{exponential generating series} (\egs)~$\EGS \FF := \sum_{n \in \N^d} \FF_n \cdot \frac {x^n} {n!}$,
where $\FF_{n_1, \dots, n_d} := \card {\FF[U_1, \dots, U_d]}$
for some (equivalently, all) finite sets of cardinalities $\card {U_1} = n_1, \dots, \card {U_d} = n_d$.
%
%
We refer to~\cite[Sec.~1]{PivoteauSalvySoria:JCT:2012} for an introduction to species
tailored towards combinatorial enumeration (\cf~also the book~\cite{CombinatorialSpecies:CUP:1998}).
Below we present the main ingredients relevant for our purposes by means of examples.

Species can be built from basic species by applying species operations and solving species equations.
Examples of \emph{basic species} are the \emph{zero} species $\ZERO$ with \egs~$0$,
the \emph{one} species $\ONE$ with \egs~$1$,
the \emph{singleton} species $\X_j$ of sort $j$ with \egs~$x_j$,
the \emph{sets} species $\SETSP$ with \egs~$e^x = 1 + x + x^2 / 2! + \cdots$
(since there is only one set of size $n$ for each $n$),
and the \emph{cycles} species $\CYCSP$ with \egs~$-\log(1 - x)$.
New species can be obtained by the operations of \emph{sum} (disjoint union) $\FF + \GG$,
\emph{combinatorial product} $\FF \cdot \GG$ (generalising the Cauchy product for words),
\emph{derivative} $\partial {X_j} \FF$~(\cf~\cite[Sec.~1.2 and 1.4]{PivoteauSalvySoria:JCT:2012} for formal definitions),
and \emph{cardinality restriction} $\restrict \FF S$ (for a \emph{cardinality constraint} $S \subseteq \N^d$).
Regarding the latter, $\restrict \FF S$ equals $\FF$ on inputs $\tuple{U_1, \dots, U_d}$
satisfying $\tuple{\card {U_1}, \dots, \card {U_d}} \in S$,
and is $\emptyset$ otherwise;
we use the notation $\FF_{\sim n}$ for the constraint
$\card {U_1} + \cdots + \card {U_d} \sim n$,
for $\sim$ a comparison operator such as $=$ or $\geq$.

Another important operation is that of composition of species~\cite[Sec.~1.5]{PivoteauSalvySoria:JCT:2012}.
Consider sorts $\X = \tuple{\X_1, \dots, \X_d}$ and $\Y = \tuple{\Y_1, \dots, \Y_k}$.
Let $\FF$ be a $\tuple{\X, \Y}$-sorted species and let $\GG = \tuple{\GG_1, \dots, \GG_k}$ be a $k$-tuple of $\X$-sorted species.
For a set of indices $I \subseteq \set{1, \dots, k}$, we write $\Y_I$ for the tuple of those $\Y_i$'s \st~$i \in I$.
We say that $\FF$ is \emph{polynomial} \wrt~$\Y_I$ if $\EGS \FF$ is polynomial \wrt~$y_I$,
and similarly for \emph{locally polynomial}.
We say that $\FF, \GG$ are \emph{$\Y$-composable} if $\FF$ is locally polynomial \wrt~$\Y_I$,
where $I$ is the set of indices $i$ \st~$\GG_i[\emptyset] \neq \emptyset$. 
The notion of \emph{strongly $\Y$-composable} is obtained by replacing ``locally polynomial'' with ``polynomial''.
For two $\Y$-composable species $\FF, \GG$ their \emph{composition} $\FF \compose \Y \GG$
is a well-defined $\X$-sorted species.
Informally, it is obtained by replacing each $\Y_i$ in $\FF$ by $\GG_i$.

We will not need the formal definitions of these operations,
but we will use the fact that each of these has a corresponding operation on power series~\cite[Ch.~1]{CombinatorialSpecies:CUP:1998}:
$\EGS{\FF + \GG} = \EGS\FF + \EGS\GG$,
$\EGS{\FF \cdot \GG} = \EGS\FF \cdot \EGS\GG$,
$\EGS{\partial {X_j} \FF} = \partial {x_j} \EGS\FF$,
$\EGS{\FF \compose \Y \GG} = \EGS\FF \compose y \EGS\GG$,
and $\EGS{\restrict \FF S} = \restrict {\EGS\FF} S$.
For instance, $\SET \X_{\geq 1}$ is the species of nonempty sets, with \egs~$e^x - 1$;
$\SET \X \cdot \SET \X$ is the species of \emph{subsets} with \egs~$e^x \cdot e^x = \sum_{n\in\N} 2^n \cdot x^n/n!$
since subsets correspond to partitions of a set into two parts and there are $2^n$ ways to do this for a set of size $n$;
$\X \cdot \X$ is the species of pairs with \egs~$2! \cdot x^2 / 2!$
since there are two ways to organise a set of size $2$ into a pair;
$\SEQ \X = 1 + \X + \X \cdot \X + \cdots$ is the species of lists
with \egs~$(1-x)^{-1} = 1 + x + x^2 + \cdots$ since there are $n!$ ways to organise a set of size $n$ into a tuple of $n$ elements;
$\SET \Y \compose \Y \SET \X_{\geq 1}$ is the species of set partitions with \egs~$e^{e^x - 1}$
since a set partition is a collection of nonempty sets which are pairwise disjoint and whose union is the whole set.


%

Finally, species can be defined as unique solutions of systems of species equations.
E.g., the species of sequences $\SEQ \X$ is the unique species satisfying $\Y = \ONE + \X \cdot \Y$
since a nonempty sequence decomposes uniquely into a first element together with the sequence of the remaining elements;
\emph{binary trees} is the unique species solution of $\Y = \ONE + \X \cdot \Y^2$;
\emph{ordered trees} is the unique species solution of $\Y = \ONE + \X \cdot \SEQ \Y$;
\emph{Cayley trees} (rooted unordered trees)
is the unique species satisfying $\Y = \X \cdot \SET \Y$
since a Cayley tree uniquely decomposes into a root together with a set of Cayley subtrees.
%
\begin{wrapfigure}{r}{0.62\textwidth}
    \vspace{-1em}
    \!\!\!\!
    \!\!\!\!\begin{align}
        \label{eq:SP graphs}
        \left\{\begin{array}{rclc}
        \Y_1 &=& \X + \Y_2 + \Y_3,
            &\text{(sp graphs)} \\
        \Y_2 &=& \SEQ {\X + \Y_3}_{\geq 2},
            &\text{(series graphs)} \\
        \Y_3 &=& \SET {\X + \Y_2}_{\geq 2}.
            &\text{(parallel graphs)}
        \end{array}\right.
    \end{align}
    \vspace{-2em}
\end{wrapfigure}
For a more elaborate example, the species of \emph{series-parallel graphs} is the unique solution for $\Y_1$
of the system in~\cref{eq:SP graphs} \cite[Sec.~0]{PivoteauSalvySoria:JCT:2012}.
Joyal's \emph{implicit species theorem}~\cite{Joyal:AM:1981}
(\cf~\cite[Theorem 2.1]{PivoteauSalvySoria:JCT:2012}, \cite[Theorem 2 of Sec.~3.2]{CombinatorialSpecies:CUP:1998}), which we now recall,
provides conditions guaranteeing existence and uniqueness of solutions to species equations.
%
%
Let a system of species equations $\Y = \FF(\X, \Y)$ (with $\FF$ a $k$-tuple of species) be \emph{well posed}
if $\FF(\ZERO, \ZERO) = \ZERO$ and the Jacobian matrix $\partial \Y \FF$ (defined as for power series~\cite[Sec.~1.6]{PivoteauSalvySoria:JCT:2012})
is nilpotent at $(\ZERO, \ZERO)$.
A \emph{canonical solution} is a solution $\Y := \GG(\X)$ \st~$\GG(\ZERO) = \ZERO$.
\begin{theorem}[Implicit species theorem~\cite{Joyal:AM:1981}]
    \label{thm:implicit species theorem}
    A well-posed system of species equations
    %
        $\Y = \FF(\X, \Y)$
    %
    admits a unique canonical solution $\Y := \GG(\X)$.
\end{theorem}
\noindent
The implicit species theorem is a direct analogue of the implicit function theorem for power series.
Furthermore, if $\Y = \FF(\X, \Y)$ is a well-posed system of species equations
then $y = \EGS \FF(x, y) $ is a well-posed system of power series equations;
moreover the \egs~of the canonical species solution of the former is the canonical power series solution of the latter.

\noindent
We now have enough ingredients to define a large class of combinatorial species.
%
    \emph{Strongly constructible species} are the smallest class of species
    \begin{inparaenum}[(1)]
        \item containing the basic species $\ZERO, \ONE, \X_j~\text{($j \in \N$)}, \SETSP, \CYCSP$;
        \item closed under sum, product, strong composition, regular cardinality restrictions; and
        \item closed under canonical resolution of well-posed systems $\Y = \FF(\X, \Y)$ with $\FF$ a tuple of strongly constructible species.
    \end{inparaenum}
%
Note that the equation $\Y = \ONE + \X \cdot \Y$ for sequences is not well posed,
nonetheless sequences are strongly constructible:
Nonempty sequences $\SEQ \X_{\geq 1}$ are the unique canonical solution of the well-posed species equation $\ZZ = \X + \X \cdot \ZZ$
and $\SEQ \X = \ONE + \SEQ \X_{\geq 1}$.
Similar manipulations show that all the examples mentioned are strongly constructible.
%
%
%
\begin{remark}
    The class of strongly constructible species is incomparable with the class from~\cite[Definition 7.1]{PivoteauSalvySoria:JCT:2012}.
    On the one hand, \cite{PivoteauSalvySoria:JCT:2012} considers as cardinality restrictions only finite unions of intervals,
    while we allow general regular restrictions, \eg~periodic constraints such as ``even size'';
    moreover, constraints in~\cite{PivoteauSalvySoria:JCT:2012} are applied only to basic species,
    while we allow arbitrary strongly constructible species.
    On the other hand, we consider well-posed systems,
    while~\cite{PivoteauSalvySoria:JCT:2012} considers more general \emph{well-founded systems}.
    Finally, we consider strong composition, while \cite{PivoteauSalvySoria:JCT:2012} considers composition.
    %
\end{remark}
%
%
%
Since \CDF~power series include the the basic species \egs\-
$0$, $1$, $x_j$ ($j \in \N$), $(1-x)^{-1}$, $e^x$, and $-\log(1-x)$,
from the \CDF~closure properties\-
\cref{lem:CDF basic closure properties,lem:CDF constructible nilpotent,lem:CDF closure under regular cardinality restrictions}
and the discussion above, we have:
\begin{theorem}
    \label{thm:strongly constructible species are CDF}
    The \egs~of a strongly constructible species is effectively \CDF.
\end{theorem}
\begin{remark}
    \emph{Constructible species} are obtained by considering composition instead of strong composition.
    We conjecture that even the \egs~of constructible species are \CDF,
    which would follow by generalising~\cref{lem:CDF basic closure properties}(3) 
    from ``strongly composable'' to ``composable''.
\end{remark}
For instance, the well-posed species equation $\Y = \X \cdot \SET \Y$ for Cayley trees
translates to the well-posed power series equation $y = x \cdot e^y$ for its \egs.
\begin{wrapfigure}{r}{0.46\textwidth}
    \vspace{-2em}
    \begin{align}
        \label{eq:SP graphs - power series equations}
        \left\{
        \begin{array}{l}
            y_1 = x + y_2 + y_3, \\
            y_2 = \frac 1 {1 - (x + y_3)} - 1 - (x + y_3), \\
            y_3 = e^{x + y_2} - 1 - (x + y_2),
        \end{array}
        \right.
    \end{align}
    \vspace{-1.5em}
\end{wrapfigure}

\noindent
The well-posed species equations for series-parallel graphs~\cref{eq:SP graphs}
translate to well-posed power series equations for their \egs~in~\cref{eq:SP graphs - power series equations}.
We conclude this section by deciding multiplicity equivalence of species.
Two $d$-sorted species $\FF, \GG$ are \emph{multiplicity equivalent}
(\emph{equipotent}~\cite{PivoteauSalvySoria:JCT:2012})
if $\FF_n = \GG_n$ for every $n \in \N^d$.
Decidability of multiplicity equivalence of strongly constructible species,
announced in \cref{thm:multiplicity equivalence of constructible species},
follows from \cref{thm:strongly constructible species are CDF,thm:CDF zeroness in TWOEXPTIME}.

%% file: 06-conclusions.tex

\section{Conclusions}
\label{sec:conclusions}

We have presented two related computation models, \WBPP~series and \CDF~power series.
We have provided decision procedures of elementary complexity for their zeroness problems (\cref{thm:CDF zeroness in TWOEXPTIME,thm:main}),
which are based on a novel analysis on the length of chains of polynomial ideals obtained by iterating a finite set of possibly noncommuting derivations
(\cref{thm:chain length bound}).
On the way, we have developed the theory of \WBPP~and \CDF,
showing in particular that the latter arises as the commutative variant of the former.
Finally, we have applied \WBPP~to the multiplicity equivalence of \BPP\-
(\cref{thm:BPP multiplicity equivalence}),
and \CDF~to the multiplicity equivalence of constructible species
(\cref{thm:multiplicity equivalence of constructible species}).
Many directions are left for further work.
Some were already mentioned in the previous sections.
We highlight here some more.

\subparagraph{Invariant ideal.}

Fix a \WBPP~(or \CDF). Consider the \emph{invariant ideal}
of all configurations evaluating to zero
$Z := \setof {\alpha \in \poly \Q N} {\sem \alpha = 0}$.
%
%
Zeroness is just membership in $Z$.
Since $Z$ is a polynomial ideal, it has a finite basis.
The most pressing open problem is whether we can compute one such finite basis,
perhaps leveraging on differential algebra~\cite{Kolchin:1973}.
$Z$ is computable in the special case of~\WFA\-
\cite{HrushovskiOuakninePoulyWorrell:LICS:2018,HrushovskiOuakninePouly:Worrell:JACM:2023},
however for polynomial automata it is not~\cite{MullnerMoosbruggerKovacs:arXiv:2023}.



\subparagraph{Regular support restrictions.}

\BPP~languages are not closed under intersection with regular languages~\cite[proof of Proposition 3.11]{Christensen:PhD:1993},
and thus it is not clear for instance whether we can decide \BPP~multiplicity equivalence within a given regular language.
%
%
We do not know whether \WBPP~series are closed under regular support restriction,
and thus also zeroness of \WBPP~series within a regular language is an open problem.

\subparagraph{\WBPP~with edge multiplicities.}

One can consider a slightly more expressive \BPP~model
where one transition can remove more than one token from the same place~\cite{MayrWeihmann:RP:2013}.
It is conceivable that zeroness stays decidable,
however a new complexity analysis is required
since the corresponding ideal chains may fail to be convex.



%% file: A02-preliminaries.tex

\section{Additional preliminaries}
\label{app:preliminaries}
\subsection{Complexity bounds}

In this section we develop some basic complexity estimations which will be useful when applied to our \WBPP~and \CDF~zeroness algorithms.
We begin by a rough estimation of the complexity of evaluating polynomials.
\begin{lemma}[Polynomials evaluation bound]
    \label{lem:polynomials evaluation bound}
    Let $p \in \poly \Q k$ be a multivariate polynomial
    of total degree $D := \deg p$ and height $H := \height p$.
    The evaluation of $p$ at some point $x \in \Q^k$ bounded by $B := \inftynorm x$
    produces a value $p(x)$ bounded by
    \begin{align*}
        \inftynorm {p(x)} \leq \binom {D + k} k \cdot H \cdot B^D.
    \end{align*}
\end{lemma}
\begin{proof}
    Recall that the number of multivariate monomials in $k$ variables
    of total degree $\leq D$ is $\binom {D + k} k$.
    For $x \in \Q^k$ we have
    \begin{align*}
        \inftynorm {p(x)}
            &= \inftynorm {\sum_{\onenorm n \leq D} p_n \cdot x^n}
            \leq \binom {D + k} k \cdot H \cdot B^D. \qedhere
    \end{align*}
\end{proof}

The next lemma shows that degrees of polynomials obtained by iteratively applying a derivation of the polynomial ring
(\emph{Lie derivative})
grow at most linearly in the number of iterations.
The \emph{degree} and \emph{height} of a derivation $L : \poly \Q k \to \poly \Q k$,
denoted $\deg L$, resp., $\height L$,
are the maximal degree, resp., heights of $L x_1, \dots, L x_k \in \poly \Q k$.

\begin{restatable}{lemma}{SmallDegreePropertyDerivation}
    \label{lem:iterated derivation - complexity bounds}
    Consider a derivation $L : \poly \Q k \to \poly \Q k$ and a polynomial $\alpha \in \poly \Q k$.
    %
    For every $n \in \N$,
    \begin{enumerate}
        \item $L^n\alpha$ has degree $\deg {L^n\alpha} \leq \max(\deg \alpha + n \cdot (\deg L - 1), 0)$, 
        \item $L^n\alpha$ has height
        \begin{align*}
            \height {L^n\alpha}
            &\leq (\deg \alpha + (n-1) (\deg L-1) +k)^{n(k+1)} \cdot \height \alpha \cdot \height L^n.
        \end{align*}
    \end{enumerate}    
\end{restatable}
%
%
\noindent
This should be contrasted with iteratively applying a \emph{homomorphism} of the polynomial ring,
where the degree can grow exponentially: For instance, let $h(x) := x^2$ and extend it homomorphically to all univariate polynomials.
We have $h^n(x) = ((x^2)^2 \cdots)^2 = x^{2^n}$ for every $n \in \N$.
This relative tameness of derivations will be a contributing factor in the elementary bound on convex ideal chains from~\cref{thm:chain length bound}.
\begin{proof}
    Write $D := \deg L$, $E := \deg \alpha$, $H := \height L$, and $I := \height \alpha$.
    %
    We first prove the claim for $n = 1$.
    We begin from the degree bound (1).
    We proceed by induction on the structure of $\alpha$.
    In one base case, $\alpha = 1$ of degree $E = 0$, and we have $\Delta_a 1 = 0$
    which has degree $0 \leq \max(D - 1, 0)$.
    In the other base case, $\alpha = X_i$ of degree $E = 1$,
    and we have $\Delta_a X_i$ by assumption of degree $\leq D \leq 1 + D - 1$.
    For the first inductive case,
    write $\alpha = \beta + \gamma$ of degree $E$.
    Thus $\beta, \gamma$ are of degree $\leq E$ and we have
    \begin{align*}
        \Delta_a \alpha &=
            \underbrace {\Delta_a \beta}_{\leq \max(E + D - 1, 0)} \ +\ \underbrace {\Delta_a \gamma}_{\leq \max(E + D - 1, 0)},
    \end{align*}
    from which we conclude that $\Delta_a \alpha$ is also of degree $\leq \max(E + D - 1, 0)$.
    For the second inductive case,
    write $\alpha = \beta \cdot \gamma$ of total degree $E = F + G$,
    where $\beta$ is of total degree $F \geq 1$ and $\gamma$ of total degree $G \geq 1$.
    We have
    \begin{align*}
        \Delta_a \alpha &=
            \underbrace {\Delta_a \beta}_{\leq \max(F + D - 1, 0) = F + D - 1} \cdot \underbrace \gamma_{G}\ +\ \underbrace \beta_{F} \cdot \underbrace {\Delta_a \gamma}_{\leq \max(G + D - 1, 0) = G + D - 1}.
    \end{align*}
    It follows that the degree of $\Delta_a \alpha$ is at most
    \begin{align*}
        \max((F + D - 1) + G, F + (G + D - 1)) =  E + D - 1 \geq 0,
    \end{align*}
    as required.

    We now come to the height bound (2).
    Let $\height L \leq H$ and $\deg \alpha \leq E$.
    We first address the case of a single monic monomial.
    Let $y = \tuple{y_1, \dots, y_k}$ be a $k$-tuple of variables.
    \begin{claim*}
        If $y^m$ is a monic monomial ($m \in \N^k$) of total degree $\onenorm m \leq E$,
        then $\height {L y^m} \leq E \cdot H$.
    \end{claim*}
    \begin{proof}
        This is readily proved by induction.
        The base cases $E = 0$ and $E = 1$ are clear.
        For larger degrees, we have
        \begin{align*}
            \height {L(y_i \cdot y^m)}
                &= \height {L y_i \cdot y^m + y_i \cdot L y^m} \leq \\
                &\leq \height {L y_i \cdot y^m} + \height {y_i \cdot L y^m} \leq \\
                &\leq \height {L y_i} + \height {L y^m} \leq \\
                &\leq H + E \cdot H = (E+1) \cdot H. \qedhere
        \end{align*}
    \end{proof}
    
    We now consider the case of a single application $n = 1$ of $L$.
    Consider a polynomial $\alpha = \sum_{m \in S} c_m \cdot y^m$ of total degree $\deg \alpha$,
    where $S \subseteq \N^k$ is its support.
    For every $m \in S$ we have $\onenorm m \leq \deg \alpha$.
    Recall also that the number of $k$-variate monomials of total degree $D$
    is ${D + k \choose k}$, and that this quantity is monotonic in $D$.
    We can thus write
    \begin{align}
        \nonumber
        \height {L \alpha}
        &\leq \sum_{m \in S} \abs {c_m} \cdot \height {L y^m}
        \leq \binom {\deg \alpha + k} k \cdot \height \alpha \cdot \deg \alpha \cdot \height L \leq \\
        \label{eq:height bound n = 1}
        &\leq (\deg \alpha + k)^{k+1} \cdot \height \alpha \cdot \height L.
    \end{align}

    We now come to the general case $n \in \N$.
    %
    %
    Let $\deg \alpha \leq E$, $\deg L \leq D$, $\height \alpha \leq I$, and $\height L \leq H$.
    We prove the following height bound,
    for all $n \in \N$ and $\alpha \in \poly \Q k$ satisfying the prescribed bounds:
    \begin{align*}
        \height {L^n \alpha} \leq (E + (n-1) (D-1) +k)^{n(k+1)} \cdot I \cdot H^n.
    \end{align*}
    The base case $n = 0$ is clear since $L^0 \alpha = \alpha$
    and the right-hand side is just $I$.
    Inductively, we have
    \begin{align*}
        &\height {L^{n+1} \alpha} = \height {L (L^n \alpha)} \leq
            \quad \text{(by~\cref{eq:height bound n = 1})} \\
        &\leq (\deg {L^n \alpha} + k)^{k+1} \cdot \height {L^n \alpha} \cdot H \leq
            \quad \text{(by~(1))} \\
        &\leq (E + n (D - 1) + k)^{k+1} \height {L^n \alpha} \cdot H \leq
            \quad \text{(by~ind.)} \\
        &\leq (E + n (D - 1) + k)^{k+1} ((E + (n-1) (D-1) +k)^{n(k+1)} \cdot I \cdot H^n) H \leq \\
        &\leq (E + n (D - 1) + k)^{(n+1)(k+1)} \cdot I \cdot H^{n+1}. \qedhere
    \end{align*}
\end{proof}

The following bounds on iterated application of derivations and their evaluation
will be useful in complexity considerations.
It is an immediate consequence of \cref{lem:iterated derivation - complexity bounds,lem:polynomials evaluation bound}.

\begin{restatable}{lemma}{iteratedDerivationBounds}
    \label{lem:iterated derivation bounds}
    Consider a derivation $L : \poly \Q k \to \poly \Q k$ and a polynomial $\alpha \in \poly \Q k$
    both of degree bounded by $D$ and height bounded by $H$.
    For every $n \in \N$ and tuple of rationals $x \in \Q^k$
    of height also bounded by $\inftynorm x \leq H$,
    \begin{align}
        \label{eq:simplified iterated derivation degree bound}
        &\deg {L^n\alpha}
            \leq \bigO {n \cdot D}, \\
        \label{eq:simplified iterated derivation height bound}
        &\height {L^n\alpha}
            \leq (n \cdot D + k)^{\bigO{n \cdot k}} \cdot H^{\bigO n}, \text{ and} \\
        \label{eq:bound on evaluationing iterated derivations}
        &\abs {(L^n \alpha)(x)} \leq (n \cdot D + k)^{\bigO{n \cdot k}} \cdot H^{\bigO {n \cdot D}}.
    \end{align}
\end{restatable}

\begin{proof}
    The first two bounds are just instantiations of~\cref{lem:iterated derivation - complexity bounds}.
    For the last one~\cref{eq:bound on evaluationing iterated derivations}, write $B:=\inftynorm x \leq H$.
    We have,
    \begin{align*}
        &\abs {(L^n \alpha)(x)} \leq
            \quad \textrm{(by~\cref{lem:polynomials evaluation bound})} \\
        &\leq (\deg (L^n \alpha)+k)^k \cdot \height {L^n \alpha} \cdot B^{\deg (L^n \alpha)}
            \quad \textrm{(by~\cref{eq:simplified iterated derivation degree bound}, \cref{eq:simplified iterated derivation height bound})}\\
        &\leq (\bigO{n \cdot D}+k)^k \cdot (\bigO{n \cdot D} + k)^{\bigO{n \cdot k}} \cdot H^{\bigO n} \cdot B^{\bigO{n \cdot D}} \\
        &\leq (nD + k)^{\bigO{n \cdot k}} \cdot H^{\bigO {n \cdot D}}. \qedhere
    \end{align*}
\end{proof}

\subsection{Topology}

Series $\series \Q \Sigma$ carry a natural topology, which we now recall.
The \emph{order} of a series $f \in \series \Q \Sigma$, denoted $\ord f$,
is the minimal length of a word in its support if $f \neq 0$, and $\infty$ otherwise.
%
%
For instance, $\ord {3ab-2abc} = 2$.
This gives rise to the non-Archimedean absolute value $\abs f := 2^{-\ord f}$,
with the convention that $\abs 0 = 2^{-\infty} = 0$.
In turn, this yields the ultrametric $d(f, g) := \abs {f - g}$.
Spelling out the definition, $d(f, g) \leq 2^{-n}$ iff $f, g$ agree on all words of length $\leq n$.
The operations of scalar multiplication, sum, shuffle product, and derivation are all continuous in the topology induced by the ultrametric $d$.
Intuitively, this means that, e.g., $(f \shuffle g)_w$ depends only on finitely many values $f_u, g_v$.

\subsection{Derivations and $\sigma$-derivations}

We recall a generalisation of derivations.
Consider a (not necessarily commutative, but unital) ring $R$ with an endomorphism $\sigma : R \to R$.
A \emph{$\sigma$-derivation} of $R$ is a linear function $\delta : R \to R$ satisfying 
\begin{align}
    \tag{general Leibniz rule}
    \label{eq:general Leibniz rule}
    \delta (a \cdot b) = \delta(a) \cdot b + \sigma(a) \cdot \delta (b).
\end{align}
When $\sigma$ is the identity endomorphism, we recover the notion of derivation from~\cref{eq:Leibniz rule}.
A \emph{$\sigma$-differential ring} is a ring equipped with finitely many $\sigma$-derivations
$\tuple {R; \sigma, \delta_1, \dots, \delta_k}$ (not necessarily commuting).
Similarly, one can define $\sigma$-differential algebras.
%

For a subring $S \subseteq R$ and $a_1, \dots, a_k \in R$,
let $\ncpoly S {a_1, \dots, a_k}$ be the smallest subring of $R$
containing $S \cup \set{a_1, \dots, a_k}$.
We call $a_1, \dots, a_k$ the \emph{generators} of $\ncpoly S {a_1, \dots, a_k}$ over $S$.
We call such subrings \emph{finitely generated} over $S$.
Finitely generated algebras are defined similarly.
The next property says that, in order to close a finitely generated subring
under the application a $\sigma$-derivation $\delta$,
it suffices to add to the generators their image under $\sigma, \delta$.
\begin{fact}
    \label{fact:closure lemma}
    Consider a $\sigma$-differential ring $\tuple{R; \sigma, \delta}$
    with subring $S \subseteq R$.
    For every $a_1, \dots, a_k \in R$,
    \begin{align*}
        \delta (\ncpoly S {a_1, \dots, a_k})
            \subseteq
                \ncpoly S {a_1, \dots, a_k, \sigma(a_1), \dots, \sigma(a_k), \delta(a_1), \dots, \delta(a_k)}.
    \end{align*}
\end{fact}

Let $x$ be an indeterminate.
The following property says that endomorphisms $\sigma$'s and $\sigma$-derivations $\delta$'s
of (not necessarily commutative) polynomial rings $\ncpoly S x$
are uniquely defined once we fix how they act on $S$ and $x$.
For $\sigma$ this is an immediate consequence of the fact that it is an endomorphism,
and for $\sigma$-derivations this is an immediate consequence of linearity and~\cref{eq:general Leibniz rule}.
%
\begin{fact}
    \label{lem:unique extension}
    Fix a $\sigma$-differential ring $\tuple{S; \sigma, \delta}$
    and consider the noncommutative polynomial ring $R = \ncpoly S x$.
    \begin{enumerate}
        \item For every element $a \in R$,
        there exists a unique way to extend $\sigma$
        to an endomorphism of $R$ \st~$\sigma(x) = a$.
        \item For every endomorphism $\sigma$ of $R$ and element $b \in R$,
        there exists a unique way to extend $\delta$
        to a $\sigma$-derivation of $R$ \st~$\delta(x) = b$.
    \end{enumerate}
\end{fact}

%% file: A03-WBPP.tex

\section{Additional material for~\NoCaseChange{\cref{sec:WBPP}}}
\label{app:WBPP}

In this section we present additional material and full proofs regarding \WBPP~and the series they recognise.

\WBPPexchangeProperty*
\begin{proof}
    First of all, we observe
    \begin{align}
        \derive a \sem \alpha = \sem {\Delta_a \alpha}.
    \end{align}
    Indeed, for every $w \in \Sigma^*$ we have
    \begin{align*}
        (\derive a \sem \alpha)_w
        &= \sem \alpha_{a \cdot w}
        = F\, \Delta_{a \cdot w}\, \alpha
        = F \, \Delta_w\, \Delta_a\, \alpha
        = \sem {\Delta_a \alpha}_w.
    \end{align*}
    Using the observation above, we now prove the lemma.
    We proceed by induction on words.
    The base case $w = \e$ is clear,
    since $\derive \e \sem \alpha = \sem \alpha = \sem {\Delta_\e \alpha}$.
    For the inductive step, we have
    \begin{align*}
        \derive {a \cdot w} \sem \alpha
        &= \derive w \derive a \sem \alpha 
        = \derive w \sem {\Delta_a \alpha}
        = \sem {\Delta_w \Delta_a \alpha}
        = \sem {\Delta_{a \cdot w} \alpha}. \quad \qedhere
    \end{align*}
\end{proof}

\WBPPhomomorphismProperty*
\begin{proof}
    It follows directly from its definition that the semantic function is linear:
    \begin{align*}
        \sem {k \cdot \alpha} = k \cdot \sem \alpha \text{ (with $k \in \Q$) }
        \quad\text{and}\quad
            \sem {\alpha + \beta} = \sem \alpha + \sem \beta.
    \end{align*}
    We now show that it is multiplicative
    \begin{align*}
        \sem {\alpha \cdot \beta} = \sem \alpha \shuffle \sem \beta.
    \end{align*}
    We proceed by induction on the length of words.
    In the base case $w = \e$, we have
    $\sem {\alpha \cdot \beta}_\e = F(\alpha \cdot \beta) = F \alpha \cdot F \beta$ and
    $(\sem \alpha \shuffle \sem \beta)_\e = \sem \alpha_\e \cdot \sem \beta_\e = F \alpha \cdot F \beta$.
    In the inductive step, we have
    \begin{align*}
        \sem {\alpha \cdot \beta}_{a \cdot w}
        &= F\, \Delta_{a \cdot w}\, (\alpha \cdot \beta) = \\
        &= F\, \Delta_w\, \Delta_a\, (\alpha \cdot \beta) = 
	        && \text{(def.~of $\Delta$)}\\
        &= F\, \Delta_w\, (\Delta_a \alpha \cdot \beta + \alpha \cdot \Delta_a \beta) = 
	        && \text{(by~\cref{eq:Leibniz rule})} \\
        &= F\, \Delta_w\, (\Delta_a \alpha \cdot \beta) + F\, \Delta_w\, (\alpha \cdot \Delta_a \beta) = 
	        && \text{(by~linearity)} \\
        &= \sem{\Delta_a \alpha \cdot \beta}_w + \sem {\alpha \cdot \Delta_a \beta}_w =
        		&& \text{(by def.~of the semantics)} \\
        &= (\sem {\Delta_a \alpha} \shuffle \sem \beta)_w + (\sem \alpha \shuffle \sem {\Delta_a \beta})_w = 
            && \text{(by ind.)} \\
        &= (\sem {\Delta_a \alpha} \shuffle \sem \beta + \sem \alpha \shuffle \sem {\Delta_a \beta})_w = 
            && \text{(by~linearity)} \\
        &= (\derive a \sem \alpha \shuffle \sem \beta + \sem \alpha \shuffle \derive a \sem \beta)_w =
            && \text{(by~\cref{lem:exchange})} \\ 
        &= (\derive a (\sem \alpha \shuffle \sem \beta))_w =
				&& \text{(by~\cref{eq:Leibniz rule})} \\
        &= (\sem \alpha \shuffle \sem \beta)_{a \cdot w}.
        		&& \text{(by def.~of $\derive a$)}
        \quad \qedhere
    \end{align*}
\end{proof}

\WBPPclosureProperties*
\begin{proof}
    Let $f = \sem S$, $g = \sem T$ be two \WBPP~series.
    The correctness of the following constructions follows from~\cref{lem:exchange,lem:WBPP semantics is a homomorphism}.
    Consider a fresh initial nonterminal $U$.
    \begin{enumerate}
        \item Scalar product by a constant $c \in \Q$:
        Define $\Delta_a U = c \cdot \Delta_a S$ and $F U = c \cdot F S$.
        Correctness $\sem U = c \cdot \sem S$ holds since
        \begin{align*}
            \sem U_\e &= F U = c \cdot S U = c \cdot \sem S_\e, \text{ and} \\
            \derive a \sem U &= \sem {\Delta_a U} = \sem {c \cdot \Delta_a S} = c \cdot \sem {\Delta_a S} = c \cdot \derive a \sem S, \forall a \in \Sigma.
        \end{align*}

        \item Sum: Define $\Delta_a U = \Delta_a S + \Delta_a T$ and $F U = F S + F T$.
        Correctness $\sem U = \sem S + \sem T$ follows from
        \begin{align*}
            \sem U_\e &= F U = F S + F T = \sem S_\e + \sem T_\e, \text{ and} \\
            \derive a \sem U
                &= \sem {\Delta_a U}
                = \sem {\Delta_a S + \Delta_a T}
                = \sem {\Delta_a S} + \sem {\Delta_a T} = \\
                &= \derive a \sem S + \derive a \sem T, \forall a \in \Sigma.
        \end{align*}

        \item Shuffle product: Define $\Delta_a U = \Delta_a S \cdot T + S \cdot \Delta_a T$
        and $F U = F S \cdot F T$.
        Correctness $\sem U = \sem S \shuffle \sem T$ follows from
        \begin{align*}
            \sem U_\e &= F U = F S \cdot F T = \sem S_\e \cdot \sem T_\e, \text{ and} \\
            \derive a \sem U
                &= \sem {\Delta_a U}
                = \sem {\Delta_a S \cdot T + S \cdot \Delta_a T} = \\
                &= \sem {\Delta_a S} \shuffle \sem T + \sem S \shuffle \sem {\Delta_a T} = \\
                &= \derive a \sem S \shuffle \sem T + \sem S \shuffle \derive a \sem T = \\
                &= \derive a (\sem S \shuffle \sem T), \forall a \in \Sigma.
        \end{align*}

        \item Derivation: Define $\Delta_b U = \Delta_b \Delta_a S$ and $F U = \sem S_a$.
        Correctness $\sem U = \derive a \sem S$ follows from
        \begin{align*}
            \sem U_\e &= F U = \sem S_a = (\derive a \sem S)_\e, \text{ and} \\
            \derive b \sem U
                &= \sem {\Delta_b U} = \sem {\Delta_b \Delta_a S} = \derive b (\derive a \sem S).
        \end{align*}

        \item Shuffle inverse of $f = \sem S$.
        Assume $f_\e \neq 0$, so that its shuffle inverse $g$ \st~$f \shuffle g = 1$ exists.
        By the basic properties of derivations, we have
        \begin{align*}
            \derive a g = \derive a f \shuffle \shufflepower g 2.
        \end{align*}
        This leads us to define $\Delta_a U = \Delta_a S \cdot U^2$ and $F U = \frac 1 {F S}$.
        Correctness follows as in the previous cases. \qedhere
    \end{enumerate}
\end{proof}


We notice that \BPP~languages
coincide with supports of \emph{$\N$-\WBPP} series,
i.e., those series recognised by \WBPP~where all coefficients are natural numbers:
$\Delta_a X \in \poly \N N, F X \in \N$ for every $a \in \Sigma, X \in N$.
\begin{restatable}{lemma}{BPPlanguages}
    \label{lem:BPP languages characterisation}
    The class of \BPP~languages equals the class of supports of $\N$-\WBPP~series.
\end{restatable}
\begin{proof}
    The left-to-right inclusion is obtained by generalising~\cref{ex:from BPP to WBPP}.
    Consider a \BPP~in standard form with productions~\cref{eq:BPP}.
    Write $E_i = a_1 . \alpha_{i, 1} + \cdots + a_n . \alpha_{i, n}$
    where $\alpha_{i, j}$ is a merge of nonterminals.
    We interpret \BPP~expressions $\alpha$ not containing action prefixes $a. \_$
    as polynomials $[\alpha] \in \poly \N N$ (in fact, coefficients are $0, 1$):
    \begin{align*}
        [\bot] &= 1, \\
        [X_i] &= X_i, \\
        [\alpha \merge \beta] &= [\alpha] \cdot [\beta], \\
        [\alpha + \beta] &= [\alpha] + [\beta].
    \end{align*}        
    For instance, $[X \merge Y \merge Y + \bot \merge Y] = XY^2 + 1 \cdot Y$.
    We then define $\Delta_a(X_i) := [E_{i, a}]$,
    where $E_{i, a}$ is the sum of the $\alpha_{i, j}$'s \st~$a_j = a$.
    The output function is $FX_i = 0$, for every nonterminal $X_i$.
    %
    %
    We claim that $\lang {X_i} = \support {\sem {X_i}}$.
    
    Since the \BPP~is in standard form, all reachable states $X_1 \goesto w E$
    are merge of nonterminals $E = X_{i_1} \merge \cdots \merge X_{i_n}$,
    which are thus mapped to monomials $[E]$.
    We assume that all nonterminals $X_i$ are \emph{productive}
    in the sense that there is an action $a \in \Sigma$ \st~$X_i \goesto a \_$.
    Consequently, reachable final states are precisely those of the form $\alpha = \bot \merge \cdots \merge \bot$,
    and thus are mapped to $[\alpha] = 1$ by the translation.
    For a merge of nonterminals $E = X_{i_1} \merge \cdots \merge X_{i_n}$
    and a polynomial $\alpha \in \poly \N N$, write
    \begin{align*}
        E \in \alpha
            \quad\text{if}\quad
                \alpha - [E] \in \poly \N N.
    \end{align*}
    One then proves that ``$\in$'' is a simulation of the \BPP~transition system
    by the (deterministic) \WBPP~transition system:
    If $E \in \alpha$ then
    \begin{enumerate}
        \item if $E$ is final, then $\e \in \support {\sem \alpha}$ (that is, $F \alpha \neq 0$), and
        \item for all $a \in \Sigma$ and $E \goesto a E'$
        we have $\alpha \goesto a \alpha'$ (where $\alpha'$ is uniquely determined by $\alpha, a$)
        and $E' \in \alpha'$.
    \end{enumerate}
    Since $X_i \in [X_i]$,
    induction on words shows that $X_i \goesto w E$ implies $[X_i] \goesto w \alpha$ with $E \in [\alpha]$.
    In particular, if $E$ is final, then $F \alpha \neq 0$ and thus $\sem {X_i}_w \neq 0$, that is $w \in \support {\sem {X_i}}$.
    This shows the inclusion $\lang {X_i} \subseteq \support {\sem {X_i}}$.
    
    For the other inclusion $\support {\sem {X_i}} \subseteq \lang {X_i}$,
    one shows by induction on words $w \in \Sigma^*$ that
    whenever there is a \WBPP~execution $X_i \goesto w \Delta_w X_i$ and $E \in \Delta_w X_i$,
    there is a \BPP~run $X_i \goesto w E$.
    The inclusion follows since $w \in \support {\sem {X_i}}$ means $\sem {X_i}_w = F \Delta_w X_i \neq 0$,
    and thus there is a \WBPP~execution $X_i \goesto w \Delta_w X_i$.
    But since $\Delta_w X_i$ is a polynomial with coefficients in $\N$ and a nonzero constant term,
    $\Delta_w X_i - 1$ is in $\poly \N N$
    and thus there is a final state $E \in \Delta_w X_i$ \st~$X_i \goesto w E$.
    This means $w \in \lang {X_i}$, as required.

    We now prove the right-to-left inclusion.
    Consider a $\N$-\WBPP~$P$ with transition function $\Delta_a X_i \in \poly \N N$.
    The supports $\support {\sem X} = \tuple{\support {\sem {X_1}}, \dots, \support {\sem {X_1}}}$
    are not changed if we replace the nonzero coefficients in $\Delta_a X_i$ by $1$.
    Similarly $F X_i$ can be assumed to be in $\set{0, 1}$.

    \begin{claim*}
        The following assumption is without loss of generality:
        $F X_i \neq 0$ implies $\Delta_a X_i = 0$, for all $X_i \in N$ and $a \in \Sigma$.
    \end{claim*}
    \begin{proof}[Proof of the claim]
        This is proved by splitting $X_i = Y_i + Z_i$ as the sum of two fresh variables $Y_i, Z_i$
        and separating their roles in the transition structure \wrt~output function:
        \begin{align*}
            \Delta_a Y_i := \Delta_a X_i,\ \Delta_a Z_i := 0,\ F Y_i := 0,\ \text{and } F Z_i := F X_i.
        \end{align*}
        Let $Y + Z = \tuple{Y_1 + Z_1, \dots, Y_k + Z_k}$.
        One can then prove by structural induction on configurations
        \begin{align*}
            (\Delta_a \alpha) \compose N (Y+Z)
                &= \Delta_a (\alpha \compose N (Y+Z)),
                    \text{ for all } a \in \Sigma, \alpha \in \poly \Q N.
        \end{align*}
        (The base case $\Delta_a X_i = \Delta_a (Y_i + Z_i)$ holds by definition.)
        From this and induction on words, one then proves
        \begin{align*}
            (\Delta_w \alpha) \compose N (Y+Z)
                &= \Delta_w (\alpha \compose N (Y+Z)),
                \text{ for all } w \in \Sigma^*, \alpha \in \poly \Q N.
        \end{align*}
        Finally, this yields $\sem \alpha = \sem {\alpha \compose N (Y + Z)}$ for every $\alpha \in \poly \Q N$.
    \end{proof}

    We have obtained just a different representation of a \BPP~$Q$ in standard form,
    since for instance $X^2 Y + Y$ corresponds to the \BPP~expression $X \merge X \merge Y + Y$.
    By the claim, $FX_i = 1$ implies that $X_i$ is final in the \BPP,
    and thus can be replaced by $\bot$ without changing any language recognised by \BPP~nonterminals.
    If we now convert the \BPP~$Q$ to an $\N$-\WBPP~$P'$ using the construction from the first part of the proof,
    then in fact we come back to $P' = P$.
    This shows that the language of $Q$ is the same as the support of $P$.
\end{proof}

%% file: A04-shuffle-finite.tex

\subsection{Shuffle-finite series}

In this section we prove that shuffle-finite series coincide with \WBPP~series.
First of all, we rephrase the definition of shuffle-finite series into a more syntactic format.
This will be our working definition for shuffle-finite series from now on.

\begin{restatable}{lemma}{shufflefinite}
    \label{lem:shuffle-finite working characterisation}
    A series $f \in \series \Q \Sigma$ is shuffle finite iff
    there exist series $f = f^{(1)}, \dots, f^{(k)} \in \series \Q \Sigma$,
    and commutative polynomials $p_a^{(i)} \in \poly \Q k$ for every $a \in \Sigma$ and $1 \leq i \leq k$,
    \st, for every $a \in \Sigma$, they satisfy the following \emph{system of shuffle-finite equations}:
    \begin{align}
        \label{eq:shuffle-finite equations}
        \left\{
            \begin{array}{rcl}
                \derive a f^{(1)} &=& p^{(1)}_a \compose {} \tuplesmall{f^{(1)}, \dots, f^{(k)}}, \\
                &\vdots& \\
                \derive a f^{(k)} &=& p^{(k)}_a \compose {} \tuplesmall{f^{(1)}, \dots, f^{(k)}}.
            \end{array}
        \right.
    \end{align}
\end{restatable}
\noindent
In the statement of the lemma,
$p^{(i)}_a \compose {} \tuplesmall{f^{(1)}, \dots, f^{(k)}} \in \series \Q \Sigma$
is the series that is obtained from the polynomial $p^{(i)}_a$
by replacing its $h$-th variable with the series $f^{(h)}$,
for every $1 \leq h \leq k$.
When doing so, the product operation in $p^{(i)}_a$
becomes the shuffle operation on series.
For instance,
\begin{align*}
    (y_1 \cdot y_2 - y_3) \compose {} \tuplesmall{f^{(1)}, f^{(2)}, f^{(3)}} = f^{(1)} \shuffle f^{(2)} - f^{(3)}.
\end{align*}
%
\begin{proof}
    For the ``if'' direction, let $f = f^{(1)}, \dots, f^{(k)} \in \series \Q \Sigma$ satisfy \cref{eq:shuffle-finite equations}.
    This witnesses that the commutative ring $$S := \poly \Q {f^{(1)}, \dots, f^{(k)}}$$ generated by these series
    contains $f = f^{(1)} \in S$ and it is closed under derivations.
    The latter claim follows since
    \begin{inparaenum}[1)]
        \item the derivatives of the generators $\derive a f^{(i)}$
        are in $S$ by \cref{eq:shuffle-finite equations}, and thus
        \item the derivative of every element of $S$ is in $S$
        by Leibniz rule~\cref{eq:derive shuffle} and induction on the complexity of expressions;
        \cf~\cref{fact:closure lemma}.
    \end{inparaenum}

    For the ``only if'' direction, assume $f$ is shuffle-finite.
    There are series $f^{(2)}, \dots, f^{(k)}$ generating a differential ring
    \begin{align*}
        S := \poly \Q {f^{(2)}, \dots, f^{(k)}}
    \end{align*}
    containing $f = f^{(1)} \in S$.
    Since $\derive a f^{(i)} \in S$ there is a polynomial $p^{(i)}_a$
    \st~$\derive a f^{(i)} = p^{(i)}_a \circ \tuple{p^{(1)}_a, \dots, p^{(k)}_a}$.
\end{proof}

We are now ready to prove the announced characterisation of \WBPP~series.
\WBPPequalsShuffleFinite*
\begin{proof}
    For the ``if'' direction, consider a \WBPP~$P = \tuple{\Sigma, N, S, F, \Delta}$ with nonterminals $N = \set{X_1, \dots, X_k}$.
    By~\cref{lem:WBPP semantics is a homomorphism},
    $\sem \_$ is a homomorphism from configurations to shuffle series.
    Consider the tuple of series $\sem X = \tuple{\sem {X_1}, \dots, \sem{X_k}}$.
    Since composition $\_ \circ \sem X$ is also such a homomorphism
    and moreover it agrees with $\sem \_$ on the generators,
    $X_i \circ \sem X = \sem {X_i}$, they are in fact the same function:
    \begin{align*}
        \sem \alpha = \alpha \circ \sem X, \quad\text{for all } \alpha \in \poly \Q N.
    \end{align*}
    In particular,
    \begin{align}
        \label{eq:if direction}
        \derive a \sem {X_i} = \sem {\Delta_a X_i} = \Delta_a X_i \circ \sem X,
    \end{align}
    showing that $\sem X$ satisfies a system of shuffle-finite equations~\cref{eq:shuffle-finite equations}.
    
    For the ``only if'' direction,
    let $f \in \series \Q \Sigma$ be a shuffle-finite series.
    By~\cref{lem:shuffle-finite working characterisation}
    there are series $f = f^{(1)}, \dots, f^{(k)}$ and a system~\cref{eq:shuffle-finite equations} of shuffle-finite equations
    defined by polynomials $p_a^{(i)}$'s \st
    \begin{align}
        \label{eq:s-finite}
        \derive a f^{(i)} = p_a^{(i)} \circ \tuple {f^{(1)}, \dots, f^{(k)}},
            \ \forall a \in \Sigma, i \in \set{1, \dots, k}.
    \end{align}
    We build a \WBPP~$P = \tuple{\Sigma, N, X_1, F, \Delta}$
    whose semantics is the solution of~\cref{eq:s-finite}.
    The set of nonterminals is $N = \set{X_1, \dots, X_k}$,
    the initial nonterminal is $X_1$,
    the final weight function $F$ maps $X_i$ to $f^{(i)}_\e$,
    and the transition function satisfies
    \begin{align*}
        \Delta_a(X_i) := p_a^{(i)}(X_1, \dots, X_k).
    \end{align*}
    For convenience, we assume that the polynomials $p_a^{(i)}$'s are in $\poly \Q N$.
    We now prove $f^{(i)} = \sem {X_i}$ for all $i \in \set{1, \dots, k}$.
    Clearly the two series $f^{(i)}, \sem {X_i}$ agree on the empty word $w = \e$,
    since $\sem {X_i}_\e = F\, X_i = f^{(i)}_\e$ by definition of $F$.
    Notice that system~\cref{eq:shuffle-finite equations} has a unique solution,
    once the initial condition $\tuple{f^{(1)}_\e, \dots, f^{(k)}_\e}$ is fixed.
    Consequently, it suffices to show that $\sem {X_i}$ is a solution of this system.
    This is established as in~\cref{eq:if direction}.
\end{proof}

%% file: A04-zeroness.tex

\subsection{Zeroness}

In this section we present supplementary material concerning decidability and complexity
of the zeroness problem for \WBPP.

\subsubsection{Word-zeroness}

We begin by expanding on the complexity of the word-zeroness problem for \WBPP.
First of all, we establish that \WBPP~configurations obtained by reading words of length $n$
have total degree linear in $n$.

\SmallDegreeProperty*
\begin{proof}
    This follows immediately from~\cref{lem:iterated derivation bounds}
    since $\Delta_w$ is the composition of $n$ derivations
    each of which has degree $\leq D$.
\end{proof}

We now show how reachable \WBPP~configurations can be compactly represented with algebraic circuits.
Let us recall that an \emph{algebraic circuit} (or simply, \emph{circuit}) $C$ over $\poly \Q {x_1, \dots, x_m}$
is a directed acyclic graph of a special kind.
Its nodes are commonly called \emph{gates},
its incoming edges are called \emph{input edges} and its outgoing edges are called \emph{output edges}.
Gates with zero output edges are called \emph{output gates}.
Gates have either zero input edges (\emph{nullary gates} or \emph{input gates})
or two input edges (\emph{internal gates}).
Input and internal gates are further refined into the following kinds:
\begin{enumerate}
    \item \emph{input constant gates}, labelled with rational numbers $\Q$;
    \item \emph{input variable gates}, bijectively labelled with $x_1, \dots, x_m$;
    \item \emph{internal addition gates} labelled with ``$+$''; and
    \item \emph{internal multiplication gates} labelled with ``$\cdot$''.
\end{enumerate}
An gate $G$ of an circuit as above
encodes in a natural way a polynomial $\sem G \in \poly \Q {x_1, \dots, x_m}$.
The advantage over an explicit encoding (sum of product of monomials)
is that circuits can be more concise in terms of number of monomials and degree.
For instance, a circuit of size $n$ can encode a polynomial of degree $2^n$ by iterative squaring.

If we have a circuit encoding the \WBPP~transition relation
and a starting configuration $\alpha$,
then we can extend it to a circuit computing $\Delta_a \alpha$
with only a linear blow-up in complexity.
This is the content of the next lemma,
with a slightly more general statement in order to make it inductive.
\begin{lemma}
    \label{lem:one-step circuit}
    Let $C$ be an circuit of size $n$
    containing gates $G_{a, i}$'s computing polynomials
    \begin{align*}
        \sem{G_{a, i}} = \Delta_a X_i \in \poly \Q N,
        \quad\text{for all $a \in \Sigma$, $1 \leq i \leq k$}.
    \end{align*}
    For every input symbol $a \in \Sigma$,
    we can extend $C$ to an circuit $C_a$ of size $\leq 4 \cdot n$
    \st~for every gate $G$ of $C$ the circuit $C_a$ contains a gate $G_a$ computing
    \begin{align*}
        \sem {G_a} = \Delta_a \sem G.
    \end{align*}
    The construction can be done in space logarithmic in $n$.
\end{lemma}
\begin{proof}
    We proceed by a case analysis on gate $G$ of $C$,
    where we assume that all gates of smaller height (= maximum distance to an input gate) have already been added.
    If $G$ is a constant gate, then we add the constant gate $G_a = 0$.
    If $G = X_i$ is a variable gate,
    then we just set $G_a := G_{a, i}$ (which is already in $C$ by assumption), without adding new gates.
    If $G = G^{(0)} + G^{(1)}$ is an addition gate
    then we add a gate $G_a$ defined as
    \begin{align*}
        G_a := G^{(0)}_a + G^{(1)}_a.
    \end{align*}
    (Since $G^{(0)}, G^{(1)}$ have smaller heights than $G$ in $C$,
    the gates $G^{(0)}_a, G^{(1)}_a$ have already been added to $C_a$.)
    If $G = G^{(0)} \cdot G^{(1)}$ is a multiplication gate,
    then we add the following three new gates:
    \begin{align*}
        D^{(0)} := G^{(0)}_a \cdot G^{(1)}, \ 
        D^{(1)} := G^{(0)} \cdot G^{(1)}_a, \text{ and }
        G_a := D^{(0)} + D^{(1)}.
    \end{align*}
    (Again, $G^{(0)}_a, G^{(1)}_a$ were already in $C_a$.)
    In each case we added $\leq 3$ gates,
    and thus the new circuit has size $\leq 4 \cdot n$.
    The logarithmic space complexity for performing this construction
    follows from the fact that the construction of a new gate
    depends only on a constant number of original or already constructed gates.
\end{proof}

In the next lemma we iteratively apply~\cref{lem:one-step circuit}
to achieve closure under $\Delta_w$ when reading an arbitrary word $w \in \Sigma^*$.
The complexity is exponential in the length of $w$.
\WBPPmultistepCircuit*
\begin{proof}
    Let $C$ be an circuit of size $n$
    containing gates $G_{a, i}$'s computing polynomials $\Delta_a X_i \in \poly \Q N$ for all $1 \leq i \leq k$,
    and a gate $D$ computing $\alpha$.
    An inductive application of~\cref{lem:one-step circuit}
    immediately shows the following claim.
    \begin{claim*}
        For every word $w \in \Sigma^*$
        we can construct a circuit $C_w$ of size $\leq 4^{\length w} \cdot n$
        \st~for every gate $G$ of $C$ there is a gate $G_w$ of $C_w$ computing
        \begin{align*}
            \sem {G_w} = \Delta_w \sem G.
        \end{align*}
    \end{claim*}
    By the claim, $C_w$ contains in particular a gate $D_w$ computing $\Delta_w \alpha$.
    Step $i+1$ of the construction can be done in space logarithmic
    in the size $\leq 4^i \cdot n$ of the circuit obtained in the previous step,
    and thus altogether polynomial in $\length w$ and logarithmic in $n$.
\end{proof}

\subsubsection{Decidability}

In this section we provide more details on the decidability of the \WBPP~zeroness problem.
This amounts to prove some important structural properties of the ideal chain~\cref{eq:WBPP ideal chain}.

\LemmaDeltaOfIdeal*
\begin{proof}
    For the first point, let $p \in \Delta_a I_n$.
    Then $p = \Delta_a q$ for some $q \in I_n$ of the form
    \begin{align*}
        q = \beta_1 \cdot \Delta_{w_1} \alpha + \cdots + \beta_m \cdot \Delta_{w_m} \alpha,
        \quad\text{with } w_1, \dots, w_m \in \Sigma^{\leq n}.
    \end{align*}
    Applying $\Delta_a$ to both sides of the equation above
    and using the fact that by~\cref{eq:Delta ext} $\Delta_a(\alpha \cdot \beta)$
    is a polynomial function of $\alpha$, $\beta$, $\Delta_a \alpha$, and $\Delta_a \beta$,
    we easily deduce $\Delta_a q \in I_{n+1}$.

    Let us now consider the second point.
    The ``$\supseteq$'' inclusion follows immediately from the definitions and the previous point.
    For the reverse ``$\subseteq$'' inclusion, let $p \in I_{n+1}$.
    We can write
    \begin{align*}
        p = \beta_1 \cdot \Delta_{w_1} \alpha + \cdots + \beta_m \cdot \Delta_{w_m} \alpha,
        \quad\text{with } w_1, \dots, w_m \in \Sigma^{\leq n + 1}.
    \end{align*}
    It is now a matter of partitioning the $w_i$'s into two classes:
    In the first class we put those words in $\Sigma^{\leq n}$
    and in the second one those in $\Sigma^{n+1}$.
    We re-index the $w_i$'s in order to have the words from the first class appear first.
    Let $k$ be the size of the first class.
    We can write
    \begin{align*}
        p = \underbrace{\beta_1 \cdot \Delta_{w_1} \alpha + \cdots + \beta_k \cdot \Delta_{w_k} \alpha}_{\in I_n\ \text{(a)}}
        + \underbrace{\beta_{k+1} \cdot \Delta_{w_{k+1}} \alpha + \cdots + \beta_m \cdot \Delta_{w_m} \alpha}_{\in \ideal{\bigcup_{a \in \Sigma} \Delta_a I_n}\ \text{(b)}}.
    \end{align*}
    Point (a) follows directly from the definition of $I_n$.
    For point (b), notice that $\Delta_u \alpha$ for $u$ of length $n+1$ can be written as
    $\Delta_{v \cdot a} \alpha = \Delta_a \Delta_v \alpha$ for $v$ of length $n$,
    and thus $\Delta_u \alpha \in \Delta_a I_n$.

    The third point is an immediate consequence of the second point.
\end{proof}

\subsubsection{Complexity}
\label{sec:ideal chain complexity}

In this section we provide the details necessary to establish
the elementary upper bound on the complexity of the \WBPP~zeroness problem, claimed in~\cref{thm:main}.
In fact, all we need to do is to prove the ideal chain length bound~\cref{thm:chain length bound},
which is the main technical contribution of the paper.
In turn, it suffices to prove~\cref{claim:NY:Lemma 8+9}.
This is done by generalising the argument of Novikov and Yakovenko~\cite{NovikovYakovenko:1999}
from the case of one derivation to the case of a finite number of derivations, possibly noncommuting.
%
%
We recall two easy properties of colon ideals:
\begin{enumerate}
    \item If $I$ is an ideal, then $I \subseteq I : J$.
    \item The colon operator is monotonic in its first argument and anti-monotonic in its second argument:
    $I \subseteq I'$ and $J' \subseteq J$ implies $I \colon J \subseteq I' \colon J'$.
\end{enumerate}
We adapt the proof from~\cite[Sec.~4]{NovikovYakovenko:1999} from the setting of a single derivation
to the more general situation of multiple, possible noncommuting, derivations $\Delta_a$, $a \in \Sigma$.
%
%
%
We need to recall some definitions.
An ideal $I$ is \emph{prime} if $p \cdot q \in I$ and $p \not\in I$ implies $q \in I$.
The \emph{radical} $\rad I$ of an ideal $I$
is the set of elements $r$ \st~$r^m \in I$ for some $m \in \N$.
Note that $\rad I$ is itself an ideal.
A \emph{radical ideal} is an ideal $I$ \st~$I = \rad I$;
in other words, $p^m \in I$ implies $p \in I$.
%
%
An ideal $I$ is \emph{primary} if $p \cdot q \in I$ and $p \not\in I$ implies $q \in \rad I$.
The radical $\rad I$ of a primary ideal $I$ is prime
and it is called the \emph{associated prime} of $I$;
in the same situation, we also say that $I$ is \emph{$\rad I$-primary}.
%
The colon operation preserves primary ideals.
\begin{claim}
    If $I, J$ are ideals with $I$ primary, then $I \colon J$ is also primary.
\end{claim}
\begin{proof}[Proof of the claim]
    Let $p \cdot q \in I \colon J$ and $p \not\in I \colon J$.
    There is $r \in J$ \st~$p \cdot r \not\in I$.
    But $(p \cdot q) \cdot r = (p \cdot r) \cdot q \in I$ and since $I$ is primary,
    there is $m \in \N$ \st~$q^m \in I \subseteq I \colon J$.
\end{proof}

The \emph{dimension} $\dim I$ of a polynomial ideal $I \subseteq \poly \Q k$
is the dimension of its associated variety
$V(I) = \setof {x \in \C^k} {\forall p \in I. p(x) = 0}$.
Since the operation of taking the variety of an ideal is inclusion-reversing,
ideal inclusion is dimension-reversing:
$I \subseteq J$ implies $\dim I \geq \dim J$.
By Hilbert's Nullstellensatz~\cite[Chapter 4, §1, Theorem 2]{CoxLittleOShea:Ideals:2015},
an ideal $I$ and its radical $\rad I$
have the same associated variety,
and thus the dimension of an ideal $I$ its the same as that of its radical,
$\dim I = \dim {\rad I}$.

A \emph{primary decomposition} of an ideal $I$
is a collection of primary ideals $\set{Q_1, \dots, Q_s}$
\st~$I = Q_1 \cap \cdots \cap Q_s$.
Such a decomposition is \emph{irredundant} if
\begin{enumerate}
    \item no $Q_i$ can be omitted without changing the intersection, and
    \item the associated primes $\rad {Q_i}$'s are pairwise distinct.
\end{enumerate}
By the Lasker-Noether theorem~\cite[\protect{Chapter 4, §8, Theorem 7}]{CoxLittleOShea:Ideals:2015}
every polynomial ideal $I$ has an irredundant primary decomposition.
Let the \emph{leading term} of an irredundant primary decomposition of $I$
be the intersection of all primary components of maximal dimension.
While irredundant primary decompositions of an ideal are not unique in general,
the leading terms of two irredundant primary decompositions of the same ideal $I$ are in fact equal.
Thus the leading term, denoted $\leadingterm I$, depends only on $I$.

Chains of equidimensional ideals induce chains of the respective leading terms.
\begin{claim}[\protect{variant of \cite[Lemma 4]{NovikovYakovenko:1999}}]
    \label{lem:NY adapted:Lemma 4}
    If two polynomial ideals $I \subseteq J$ have the same dimension $\dim I = \dim J$,
    then $\leadingterm I \subseteq \leadingterm J$.
\end{claim}
\begin{proof}[Proof of the claim]
    Write $\leadingterm I = P_1 \cap \cdots \cap P_s$ where the $P_i$'s are the primary ideals of an irredundant decomposition of $I$ of maximal dimension,
    and similarly $\leadingterm J = Q_1 \cap \cdots \cap Q_t$.
    Under the same assumptions, \cite[Lemma 4]{NovikovYakovenko:1999} shows that
    each component $Q_i$ of $\leadingterm J$ contains some component $P_j$ of $\leadingterm I$,
    from which $\leadingterm I \subseteq \leadingterm J$ follows immediately.
\end{proof}

As a first stepping stone, we can already derive a bound on the length of chains of ideals
whose primary decomposition only involves associated primes from a fixed set.
We refer to \cite[Sec.~4.1]{NovikovYakovenko:1999} and the discussion therein
for the notion of \emph{multiplicity} of a primary component.
\begin{claim}[\protect{\cite[Lemma 6]{NovikovYakovenko:1999}}]
    \label{lem:NY:Lemma 6}
    Consider a finite collection of pairwise distinct prime ideals $\mathcal P$ of the same dimension
    and a strictly ascending ideal chain
    \begin{align*}
        J_0 \subsetneq J_1 \subsetneq \cdots \subsetneq J_\ell
    \end{align*}
    where each $J_i$'s has an irredundant primary ideal decomposition with associated primes from $\mathcal P$.
    (Due to the dimensionality requirement such a decomposition is unique.)
    Then, the length of the chain $\ell$ is at most the number of primary components of $J_0$
    (counted with multiplicities).
\end{claim}

The results recalled so far did not use convexity of the chain~\cref{eq:WBPP ideal chain}.
In the next claim we are generalising the analysis of~\cite{NovikovYakovenko:1999} to multiple derivations.
With~\cref{lem:NY adapted:Lemma 4,lem:NY:Lemma 6}
we obtain our second stepping stone,
namely a chain length bound
when the ideal chain and the colon chain are equidimensional.
\begin{claim}[\protect{\NoCaseChange{generalising} \protect{\cite[Lemma 8]{NovikovYakovenko:1999}}}]
    \label{claim:NY:Lemma 8}
    Consider a strictly ascending chain of convex ideals of the form~\cref{eq:WBPP ideal chain} of length $\ell$
    where the colon ratios have the same dimension as the starting ideal $I_0$:
    \begin{align*}
        \dim {I_0} = \dim {(I_0 \colon I_1)} = \dim {(I_1 \colon I_2)} = \cdots = \dim {(I_{\ell - 1} \colon I_\ell)}.
    \end{align*}
    (In particular, also the ideals $I_n$'s have the same dimension,
    \begin{align*}
        \dim {I_0} = \dim {I_1} = \cdots = \dim {I_\ell}.)
    \end{align*}
    Then $\ell$ is at most the number of primary components of the leading term $\leadingterm {I_0}$
    of the starting ideal (counted with multiplicities).
\end{claim}
\begin{proof}[Proof of the claim.]
    Let $m := \dim {I_0}$ be the common dimension of the ideals and colon ideals.
    Since the ideals $I_n$ have the same dimension $m$,
    by~\cref{lem:NY adapted:Lemma 4} the leading terms chain is ascending,
    $\leadingterm {I_0} \subseteq \leadingterm {I_1} \subseteq \cdots \subseteq \leadingterm {I_\ell}$.
    
    In fact the latter chain is strictly ascending.
    If $\leadingterm {I_n} = \leadingterm {I_{n+1}}$ for some $0 \leq n \leq \ell - 1$,
    then we have
    \begin{align}
        \label{eq:successor in leading term}
        \Delta_w \alpha \in I_{n+1} \subseteq \leadingterm {I_{n+1}} = \leadingterm {I_n},
        \quad \text{for all } w \in \Sigma^{n+1},
    \end{align}
    and thus $\Delta_w \alpha$ belongs to all primary components of $I_n$ of maximal dimension $m$.
    Write $I_n = \leadingterm {I_n} \cap R_n$, where $R_n$ is the intersection of primary components of dimension $< m$
    and define
    \begin{align}
        \label{eq:def of Delta n}
        \Delta_n \alpha := \bigcup_{w \in \Sigma^n} \Delta_w \alpha
    \end{align}
    for the set of polynomials which can be obtained from $\alpha$
    by words of length exactly $n$.
    By~\cref{eq:successor in leading term} we have $\Delta_{n+1} \alpha \subseteq \leadingterm {I_n}$ and thus
    \begin{align}
        \label{eq:trivial colon ideal}
        \leadingterm {I_n} \colon \ideal {\Delta_{n+1} \alpha} = \ideal 1.
    \end{align}
    We can now write
    \begin{align*}
        I_n \colon I_{n+1}
            &= I_n \colon \ideal {\Delta_{n+1} \alpha} = \\
            &= (\leadingterm {I_n} \cap R_n) \colon \ideal {\Delta_{n+1} \alpha} = \\
            &= \left(\leadingterm {I_n} \colon \ideal {\Delta_{n+1} \alpha} \right) \cap \left (R_n \colon \ideal {\Delta_{n+1} \alpha} \right) = && \text{(by~\cref{eq:trivial colon ideal})} \\
            &= \ideal 1 \cap \left (R_n \colon \ideal {\Delta_{n+1} \alpha} \right)
            \supseteq R_n.
    \end{align*}
    Passing to dimensions this means $m = \dim {(I_n \colon I_{n+1})} \leq \dim {R_n} < m$, which is a contradiction.
    We have thus established that the leading terms chain is strictly ascending:
    $\leadingterm {I_0} \subsetneq \leadingterm {I_1} \subsetneq \cdots \subsetneq \leadingterm {I_\ell}$.

    Consider now the affine variety $X_n := V(\leadingterm {I_n}) \subseteq \C^k$
    corresponding to the leading term ideal $\leadingterm {I_n}$.
    We have a non-increasing chain of varieties $X_0 \supseteq X_1 \supseteq \cdots \supseteq X_\ell$,
    where the next variety $X_{n+1}$ is obtained from the previous one $X_n$
    by intersecting with $V(\Delta_{n+1} \alpha)$,
    \begin{align*}
        X_{n+1} = X_n \cap V(\Delta_{n+1} \alpha).
    \end{align*}
    Each maximal-dimensional irreducible component $C$ of $X_n$
    gives rise to an affine subvariety $C \cap V(\Delta_{n+1} \alpha)$ of $X_{n+1}$.
    There are two options:
    Either all polynomials from $\Delta_{n+1} \alpha$ vanish on $C$ 
    and in this case $C = C \cap V(\Delta_{n+1} \alpha)$ is also a maximal-dimensional irreducible component of $X_{n+1}$,
    or otherwise $C$ gives rise to an affine (not necessarily irreducible) subvariety $C \cap V(\Delta_{n+1} \alpha) \subsetneq C$ of $X_{n+1}$
    of strictly smaller dimension $< m$.
    It follows that the irreducible maximal-dimensional components of $X_{n+1}$ include those of $X_n$.

    Recall that the ideals of the irreducible components of maximal dimension of $X_n$
    are precisely the associated primes of $I_n$ of maximal dimension.
    It follows that the set of associated primes of $\leadingterm {I_{n+1}}$ (which are all of the same dimension $m$)
    includes those of $\leadingterm {I_n}$.
    In particular, the associated primes of $\leadingterm {I_n}$ include those of $\leadingterm {I_0}$,
    and are all of the same dimension $m$,
    call them $\mathcal P$.

    Summarising, we have proved that the leading term chain is strictly increasing
    \begin{align*}
        \leadingterm {I_0} \subsetneq \leadingterm {I_1} \subsetneq  \cdots \subsetneq \leadingterm {I_\ell}
    \end{align*}
    and the associated primes all come from a fixed set $\mathcal P$ of prime ideals,
    all of the same dimension $m$.
    We can thus apply~\cref{lem:NY:Lemma 6} to this chain and bound its length $\ell$
    by the number of primary components of $\leadingterm {I_0}$ (counted with multiplicities).
\end{proof}

We would like to relax the requirement from~\cref{claim:NY:Lemma 8}
that the chain of colon ideals has the same dimension as the first element $I_0$ of the ideal chain.
To this end, we show that the length of ideal chains as in~\cref{claim:NY:Lemma 8}
in fact provide an upper bound on the length of ideal chains without the additional requirement.
\begin{claim}[generalising \protect{\cite[Lemma 9]{NovikovYakovenko:1999}}]
    \label{claim:NY:Lemma 9}
    Assume that the colon ideals $I_n \colon I_{n+1}$
    along a strictly ascending chain~\cref{eq:WBPP ideal chain} of length $\ell$
    have all the same dimension $m$.
    Consider any primary decomposition of the initial ideal $I_0$ and write it as
    \begin{align*}
        I_0 = I_0' \cap S
    \end{align*}
    where $I_0'$ is the intersection of primary components of dimension $\leq m$
    and $S$ is the intersection of primary components of dimension $\geq m + 1$.
    Consider the new ideal chain started at $I_0'$,
    \begin{align}
        \label{eq:NY:Lemma 9:chain}
        I_0' \subseteq I_1' \subseteq \cdots \subseteq I_\ell',
            \quad\text{where }
                I_{n+1}' := I_n' + \ideal {\Delta_{n+1} \alpha}.
    \end{align}
    ($\Delta_n$ is defined in~\cref{eq:def of Delta n}.)
    Then, for all $0 \leq n < \ell$,
    \begin{enumerate}[(1)]
        \item $\Delta_{n+1}\alpha \subseteq S$,
        \item $I_n = I_n' \cap S$, and
        \item $I_n \colon I_{n+1} = I_n' \colon I_{n+1}'$.
    \end{enumerate}
\end{claim}
\noindent
Note that from condition (2) it follows that termination of the $I_n'$ chain
implies the same for the $I_n$ chain.
\begin{proof}[Proof of the claim.]
    Properties (1) and (2) are proved by simultaneous induction on $n$.
    Property (3) will then follow at once.

    \noindent
    (1)~
    By way of contradiction,
    assume $p := \Delta_w \alpha \in \Delta_{n+1} \alpha \subseteq I_{n+1}$
    is not in $S$ for some word $w \in \Sigma^{n+1}$
    of minimal length $n+1$.
    (In particular, (1) holds for words of shorter length $n$ and thus below we can use (2) at $n$.)
    Then $p$ does not belong to some primary component $Q$ of $S$ of dimension $\geq m+1$.
    We argue that $\dim {(Q \colon \ideal p)} = \dim Q$.
    In fact, we prove the stronger fact
    \begin{align*}
        \rad {Q \colon \ideal p} = \rad Q.
    \end{align*}
    The ``$\supseteq$'' inclusion follows from the fact that the colon and radical operations are monotone.
    For the ``$\subseteq$'' inclusion, let $q \in \rad {Q \colon \ideal p}$.
    There is $m \in \N$ \st~$q^m \in Q \colon \ideal p$.
    By definition of colon ideal, $q^m \cdot p \in Q$.
    Since $Q$ is primary and $p \not\in Q$,
    we have $q^{m \cdot h} \in Q$ for some $h \in \N$,
    which means $q \in \rad Q$.
    
    By (2), we have $I_n \subseteq S \subseteq Q$.
    By monotonicity of colon ideals,
    $I_n : I_{n+1} \subseteq I_n : \ideal p$ (since $\ideal p \subseteq I_{n+1}$)
    and $I_n : \ideal p \subseteq Q : \ideal p$ (since $I_n \subseteq Q$).
    Since the latter ideal $\dim {(Q \colon \ideal p)}$ has dimension $\geq m + 1$,
    the same holds for the colon ideal $I_n : I_{n+1}$, which is a contradiction.
    We conclude $p \in S$, as required.
    
    \noindent
    (2)~
    We proceed by induction.
    The base case $n = 0$ holds by construction.
    For the inductive step,
    \begin{align*}
        I_{n+1}' \cap S
            &= (I_n' + \ideal {\Delta_{n+1} \alpha}) \cap S = \\
            &= I_n' \cap S + \ideal {\Delta_{n+1} \alpha} \cap S = 
                && \text{(by (1))} \\
            &= I_n' \cap S + \ideal {\Delta_{n+1} \alpha} =
                && \text{(by (2))} \\
            &= I_n + \ideal {\Delta_{n+1} \alpha} = I_{n+1}.
    \end{align*}

    \noindent
    (3)~We have
    \begin{align*}
        I_n \colon I_{n+1}
            &= I_n \colon \ideal {\Delta_{n+1} \alpha} =
                && \text{(by (2))} \\
            &= (I_n' \cap S) \colon \ideal {\Delta_{n+1} \alpha} =
                && \text{(by (1))} \\
            &= I_n' \colon \ideal {\Delta_{n+1} \alpha} = I_n' \colon I_{n+1}'.
            \qedhere
    \end{align*}
\end{proof}

Thanks to \cref{claim:NY:Lemma 9},
we get rid of the assumption from \cref{claim:NY:Lemma 8} that the chain of colon ideals
has the same dimension as the initial ideal,
yielding our third, and final, stepping stone.
%
\claimA*
\begin{proof}
    Consider an arbitrary primary component decomposition of $I_0$
    and as in the statement of \cref{claim:NY:Lemma 9} write it as $I_0 = I_0' \cap S$,
    where $I_0'$ is the intersection of primary components of dimension $\leq m$
    and $S$ is the intersection of primary components of dimension $\geq m+1$.
    Let $\ell'$ be the length of the ideal chain~\cref{eq:NY:Lemma 9:chain}.
    By construction, this chain starts at an ideal $I_0'$ of dimension $m$,
    and by \cref{claim:NY:Lemma 9}, point (3), this dimension is the same for the corresponding chain of colon ideals $I_n' : I_{n+1}'$.
    Clearly $\ell \leq \ell'$: If $I_n' = I_{n+1}'$, then by \cref{claim:NY:Lemma 9}, point (2), $I_n = I_{n+1}$.
    Moreover, \cref{claim:NY:Lemma 8} applies to the new ideal chain~\cref{eq:NY:Lemma 9:chain},
    since the dimension of colon ratios of the new chain is $m$,
    the same dimension as the new starting ideal $I_0'$.
    Thus $\ell'$ is at most the number of primary components of $I_0'$
    of maximal dimension $= m$ (counted with multiplicities).
    In turn, this number is at most the number of \emph{all} primary components
    of the given decomposition of the starting ideal $I_0$ (counted with multiplicities).
\end{proof}
%
%
Having established~\cref{claim:NY:Lemma 8+9},
the proof of~\cref{thm:chain length bound} is now completed.

%% file: A05-CDF.tex
\section{Additional material for \NoCaseChange{\cref{sec:CDF}}}
\label{app:CDF}

In this section we provide additional details on the development of \CDF~power series.

We recall a characterisation of \CDF~in the same spirit as for \WBPP~series in~\cref{sec:shuffle-finite series}.
This characterisation will be put at work in the proofs of~\cref{lem:CDF basic closure properties,lem:CDF closure under regular cardinality restrictions,lem:CDF strong composition}.
%
%
\begin{lemma}[\protect{\cite[Proposition 10]{BergeronSattler:TCS:1995}}]
    \label{lem:CDF characterisation}
    A power series is \CDF~if, and only if, it belongs to a finitely generated differential subalgebra of~$\powerseries \Q x$.
\end{lemma}

\begin{remark}[Simple \CDF]
    Let $f \in \powerseries \Q x$ be \emph{simple \CDF}~\cite[Sec.~5, eq.~(12)]{BergeronReutenauer:EJC:1990}
    if the algebra $\polyof \Q {\partial x^n f} {n \in \N^d}$ generated by all derivatives of $f$ is finitely generated.
    Since this algebra is closed under derivation by definition,
    if $f$ is simple \CDF~then it is \CDF.
    In~\cite[Sec.~5]{BergeronReutenauer:EJC:1990} it is shown that
    $e^{x^2}$ is a \CDF~power series which is not simple.
\end{remark}

\subsection{Basic closure properties}

In order to be self-contained, we provide full proofs
of the basic \CDF~closure properties.
This is also a nice playground for the characterisation of \CDF~from~\cref{lem:CDF characterisation}.

\basicClosurePropertiesCDF*
\begin{proof}
    A quick argument for these basic closure properties
    can be obtained by virtue of~\cref{lem:CDF characterisation}.
    Thanks to \cref{fact:closure lemma},
    in order to show that a finitely generated algebra is closed under derivatives
    it suffices to show that this is the case for its generators,
    which will be used several times below.
    
    Let $f, g$ be \CDF. By the characterisation given by~\cref{lem:CDF characterisation}
    there are generators
    $f^{(1)}, \dots, f^{(k)}, g^{(1)}, \dots, g^{(m)} \in \powerseries \Q x$ \st
    \begin{align*}
        f \in F := \poly \Q {f^{(1)}, \dots, f^{(k)}}
            \quad \text{and} \quad
                g \in G := \poly \Q {g^{(1)}, \dots, g^{(m)}},
    \end{align*}
    and $F, G$ are closed under partial derivatives $\partial {x_j}$'s.
    Closure under scalar product is immediate since $c \cdot f \in F$ by definition.
    Regarding closure under sum and product it is enough to take the union of the generators:
    \begin{align*}
        f + g, f \cdot g \in R := \poly \Q {f^{(1)}, \dots, f^{(k)}, g^{(1)}, \dots, g^{(m)}}.
    \end{align*}
    The new ring $R$ is clearly closed for derivatives $\partial {x_j}$.
    Closure under derivation holds by definition since $\partial {x_j} f \in F$ already.

    We show closure under multiplicative inverse (when it exists).
    Let $g := f^{-1}$.
    By definition $f \cdot g = 1$ and thus by the product rule
    $\partial {x_j} (f \cdot g) = \partial {x_j} f \cdot g + f \cdot \partial {x_j} g = 0$.
    We solve for $\partial {x_j} g$ and write
    \begin{align}
        \label{eq:equation for inverse g}
        \partial {x_j} g = - f^{-1} \cdot \partial {x_j} f \cdot g = - \partial {x_j} f \cdot g^2.
    \end{align}
    This suggests to add $g$ to the set of generators
    and consider the finitely generated algebra
    \begin{align*}
        g \in R := \poly \Q {f^{(1)}, \dots, f^{(k)}, g}.
    \end{align*}
    The ring $R$ is closed under derivatives $\partial {x_j}$'s by~\cref{eq:equation for inverse g}.

    Finally, assume $\partial {x_1} f, \dots, \partial {x_d} f$ are \CDF.
    By taking the union of their generators,
    we can assume that they all lie in $\poly \Q {g^{(1)}, \dots, g^{(m)}}$,
    closed under derivatives.
    We construct the finitely generated algebra
    \begin{align*}
        \poly \Q {g^{(1)}, \dots, g^{(m)}, f}
    \end{align*}
    which clearly contains $f$ and is closed under derivatives.
    
    Closure under strong composition requires some additional developments and it will be proved in~\cref{app:CDF composition} (\cref{lem:CDF strong composition}).
\end{proof}

\subsection{Composition of \CDF~power series}
\label{app:CDF composition}

In analogy for series, the \emph{order} $\ord f$ of a power series $f \in \powerseries \Q d$
is the minimal total degree of a monomial in its support.
This endows the set of power series with a natural topology.
Common operations on power series are continuous for this topology
(addition, multiplication, partial derivative, composition $\_ \compose y g$).

A family of power series $\set{f_0, f_1, \dots} \subseteq \powerseries \Q x$ is \emph{summable} if the following limit exists,
\begin{align*}
    \sum_{n = 0}^\infty f_n := \lim_{n \to \infty} \sum_{i = 0}^n f_i \in \powerseries \Q x.
\end{align*}
Spelling out the definition, this means that for every $n \in \N^d$
there are finitely many power series $f_i$ with nonzero coefficient $\coefficient {x^n} f_i$.
This also gives formal meaning to the infinite sum representation of power series
$f = \sum_{n \in \N^d} f_n \cdot \frac {x^n} {n!}$
since the family of monomials $f_n \cdot \frac {x^n} {n!}$
is summable and its sum is $f$.

A linear and continuous function on power series $F : \powerseries \Q x \to \powerseries \Q x$
commutes with sums of families of summable functions, provided it preserves summability:
\begin{align*}
    F \sum_{n = 0}^\infty f_n
    = F \lim_{n \to \infty} \sum_{i = 0}^n f_i
    = \lim_{n \to \infty} F \sum_{i = 0}^n f_i
    = \lim_{n \to \infty} \sum_{i = 0}^n F f_i
    = \sum_{n = 0}^\infty F f_n.
\end{align*}
The operations of scalar product $c \cdot \_$ ($c \in \Q$)
and partial derivation $\partial {x_j}$ preserve summability.
Moreover, if $F, G : \powerseries \Q x \to \powerseries \Q x$ preserve summability,
then $F + G$ (defined as $(F + G)(f) := Ff + Gg$)
and $F \cdot G$ (defined as $(F \cdot G)(f) := Ff \cdot Gf$)
also preserve summability.

\begin{lemma}
    \label{lem:composable closure}
    If $f \in \powerseries \Q {x, y}, g \in \powerseries \Q x^k$ are $y$-composable,
    then the same is true for $\partial {x_j} f, g$ and $\partial {y_i} f, g$.
\end{lemma}
\begin{proof}
    Let $I$ be the set of indices $i$ \st~$g^{(i)}(0) \neq 0$.
    By definition of $y$-composability, $f$ is locally polynomial \wrt~$y_I$,
    that is $f \in \powerseries {\poly \Q {y_I}}{x, y_{\setminus I}}$.
    Then it is clear that $\partial {x_j} f, \partial {y_i} f$ belong to the same ring,
    and thus they are $y$-composable with $g$.
\end{proof}

The chain rule is a fundamental property connecting derivation and composition.
It is the main ingredient in proving closure under composition for \CDF~series.
\begin{restatable}[Chain rule for power series]{lemma}{powerSeriesChainRule}
    \label{lem:chain rule for power series}
    For tuples of commuting variables $x = \tuple{x_1, \dots, x_d}$, $y = \tuple{y_1, \dots, y_k}$,
    and $y$-composable power series $f \in \powerseries \Q {x, y}$, $g \in \powerseries \Q x^k$
    we have, for every $1 \leq j \leq d$,
    \begin{align}
        \label{eq:chain rule for power series}
        \partial {x_j} (f \compose y g)
        &= \partial {x_j} f \compose y g + \sum_{i = 1}^k \partial {x_j} g^{(i)} \cdot (\partial {y_i} f \compose y g),
    \end{align}
\end{restatable}
%
\noindent
This is a general result on power series;
in particular, there is no assumption on whether the involved power series are \CDF.
%
\begin{proof}
    Fix an index $j$ and a tuple of power series $g$ as in the statement of the lemma.
    First note that the \rhs~of~\cref{eq:chain rule for power series} is well-defined
    since $\partial {x_j} f, g$ and $\partial {y_i} f, g$ are $y$-composable by~\cref{lem:composable closure}.
    
    Rewrite~\cref{eq:chain rule for power series} as
    \begin{align}
        \label{eq:chain rule FG}
        F(f) = G(f),
    \end{align}
    where we have used the abbreviations
    \begin{align*}
        F(f) &:= \partial {x_j} (f \compose y g), \text{and} \\
        G(f) &:= \partial {x_j} f \compose y g + \sum_{i = 1}^k \partial {x_j} g^{(i)} \cdot (\partial {y_i} f \compose y g).
    \end{align*}
    Since partial differentiation $\partial {x_j}$ and composition $\_ \compose y g$
    are linear and continuous, the same is true for $F$ and $G$.
    This implies that \cref{eq:chain rule FG} is inductive
    over scalar multiplication, summable families of power series, and products.
    We show this in the following claims.
    
    \begin{claim*}
        If \cref{eq:chain rule FG} holds for $f$,
        then it holds for $c \cdot f$, for every $c \in \Q$,
    \end{claim*}
    \begin{proof}[Proof of the claim]
        \begin{align*}
            F(c \cdot f) = c \cdot F(f) = c \cdot G(f) = G(c \cdot f).
            \quad \qedhere
        \end{align*}
    \end{proof}

    %
    \begin{claim*}
        Assume $\set{f_0, f_1, \dots} \subseteq \powerseries \Q x$ is a summable family
        \st~1) $\set{F(f_0), F(f_1), \dots}$ and $\set{G(f_0), G(f_1), \dots}$ are also summable
        and 2) for every $f_i$ the chain rule~\cref{eq:chain rule FG} holds.
        Then, it also holds for $\sum_{n = 0}^\infty f_n$.
    \end{claim*}
    \begin{proof}[Proof of the claim]
        \begin{align*}
            F(\sum_{n=0}^\infty f_n) = \sum_{n=0}^\infty F(f_n) = \sum_{n=0}^\infty G(f_n) = G(\sum_{n=0}^\infty f_n). 
            \quad \qedhere
        \end{align*}
    \end{proof}

    \begin{claim*}
        If \cref{eq:chain rule FG} holds for $f_1, f_2$,
        then it holds for $f_1 \cdot f_2$.
    \end{claim*}
    \begin{proof}[Proof of the claim]
        Both $F$ and $G$ satisfy a property akin to~\cref{eq:Leibniz rule}:
        \begin{align}
            \label{eq:almost derivation F}
            F(f_1 \cdot f_2) &= F(f_1) \cdot (f_2 \compose y g) + (f_1 \compose y g) \cdot F(f_2), \\
            \label{eq:almost derivation G}
            G(f_1 \cdot f_2) &= G(f_1) \cdot (f_2 \compose y g) + (f_1 \compose y g) \cdot G(f_2).
        \end{align}
        Indeed, regarding $F$ we have
        \begin{align*}
            F(f_1 \cdot f_2)
            &= \partial {x_j} ((f_1 \cdot f_2) \compose y g) = \\
            &= \partial {x_j} ((f_1 \compose y g) \cdot (f_2 \compose y g)) = \\
            &= \partial {x_j} (f_1 \compose y g) \cdot (f_2 \compose y g) + (f_1 \compose y g) \cdot \partial {x_j} (f_2 \compose y g) = \\
            &= F(f_1) \cdot (f_2 \compose y g) + (f_1 \compose y g) \cdot F(f_2),
        \end{align*}
        Regarding $G$, write $G_0(f) := \partial {x_j} f \compose y g$
        and $G_i(f) := \partial {y_i} f \compose y g$ for $1 \leq i \leq k$.
        We notice that~\cref{eq:almost derivation G} holds for $G_0, \dots, G_k$:
        \begin{align*}
            G_0(f_1 \cdot f_2)
            &= \partial {x_j} (f_1 \cdot f_2) \compose y g \\
            &= (\partial {x_j} f_1 \cdot f_2 + f_1 \cdot \partial {x_j} f_2) \compose y g = \\
            &= (\partial {x_j} f_1 \compose y g) \cdot (f_2 \compose y g) + (f_1 \compose y g) \cdot (\partial {x_j} f_2 \compose y g) = \\
            &=  G_0 (f_1) \cdot (f_2 \compose y g) + (f_1 \compose y g) \cdot G_0(f_2).
        \end{align*}
        The proof for $G_1, \dots, G_k$ is analogous.
        Since $G$ is a linear combination of $G_0, \dots, G_k$,
        property~\cref{eq:almost derivation G} follows for $G$ as well.
        
        Having established the product rules \cref{eq:almost derivation F} and~\cref{eq:almost derivation G},
        we show that~\cref{eq:chain rule FG} is inductive over products of power series,
        completing the proof of the claim.
        Indeed, assume that~\cref{eq:chain rule FG} holds for $f_1, f_2$.
        We can then write
        \begin{align*}
            F(f_1 \cdot f_2)
            &= F(f_1) \cdot (f_2 \compose y g) + (f_1 \compose y g) \cdot F(f_2) = \\
            &= G(f_1) \cdot (f_2 \compose y g) + (f_1 \compose y g) \cdot G(f_2) = \\
            &= G(f_1 \cdot f_2).
            \quad \qedhere
        \end{align*}
    \end{proof}
    In the last claim we establish the chain rule for single variables.%
    \begin{claim*}
        \cref{eq:chain rule FG} holds for single variables $f = x_\ell$ and $f = y_\ell$.
    \end{claim*}
    \begin{proof}[Proof of the claim]
        If $f = x_\ell$ then we have $\partial {y_i} f = 0$, and thus
        \begin{align*}
            F(f) &= \partial {x_j} (x_\ell \compose y g) = \partial {x_j} x_\ell, \text{and} \\
            G(f) &= \partial {x_j} x_\ell \compose y g = \partial {x_j} x_\ell.
        \end{align*}
        If $f = y_\ell$ then $f \compose y g = g^{(\ell)}$
        and $\partial {x_j} f = \partial {y_i} f = 0$ for all $i \neq \ell$,
        implying that the \rhs~of~\cref{eq:chain rule FG}~is just $\partial {x_j} g^{(\ell)}$, as required.
    \end{proof}

    We now prove~\cref{eq:chain rule FG} by bringing all claims together.
    Write the power series $f$ as the sum
    \begin{align*}
        f = \sum_{m\in\N^d, n \in \N^k} f_{m, n} \cdot x^m y^n
    \end{align*}
    of the summable family of monomials $\setof{f_{m, n} \cdot x^m y^n}{m \in \N^d, n \in \N^k}$.
    Since $f, \partial {x_j} f, \partial {y_i} f$ are all $y$-composable with $g$, the families
    \begin{align*}
        &\setof{F(f_{m, n} \cdot x^m y^n)}{m \in \N^d, n \in \N^k}
        \text{ and } \\
        &\setof{G(f_{m, n} \cdot x^m y^n)}{m \in \N^d, n \in \N^k}
    \end{align*}
    are also summable.
    Since $F(h) = G(h)$ for $h$ a single variable $h = x_\ell$ or $h = y_\ell$ by the last claim,
    we have $F(x^m y^n) = G(x^m y^n)$ by multiple applications of the claim for products,
    but then we also have $F(f_{m, n} \cdot x^m y^n) = G(f_{m, n} \cdot x^m y^n)$ by the claim for scalar multiplication,
    and finally
    \begin{align*}
        F(f)
        &= F (\sum_{m\in\N^d, n \in \N^k} f_{m, n} \cdot x^m y^n) = \\
        &= \sum_{m\in\N^d, n \in \N^k} F(f_{m, n} \cdot x^m y^n) = \\
        &= \sum_{m\in\N^d, n \in \N^k} G(f_{m, n} \cdot x^m y^n) = \\
        &= G (\sum_{m\in\N^d, n \in \N^k} f_{m, n} \cdot x^m y^n) = G(f)
    \end{align*}
    by the claim for summable families.
\end{proof}

The following lemma shows that if a \CDF~series is polynomial in a set of variables,
then the generators can be chosen with the same property.
This will be used in the proof of closure under strong composition
for \CDF~power series (\cref{lem:CDF strong composition}).
\begin{lemma}
    \label{lem:cutting lemma}
    Let $f$ belong to a finitely generated differential subalgebra of $\powerseries \Q x$.
    %
    If $f$ is a polynomial in $Z \subseteq \set{x_1, \dots, x_d}$,
    then the generators can be chosen to be polynomial in $Z$.
\end{lemma}
\noindent
We leave it open whether the lemma holds by replacing ``polynomial'' with ``locally polynomial''.
Were this the case, we could prove the generalisation of \cref{lem:CDF strong composition} stating that \CDF~power series are closed under composition.
\begin{proof}
    We show the lemma for $Z = \set {x_1}$.
    Since the construction preserves polynomiality of the generators,
    iterated application proves the general statement.
    
    Assume that $f$ is a polynomial in $x_1$ of degree $E \in \N$
    and that it belongs to the finitely generated differential algebra
    \begin{align*}
        R := \poly \Q {f^{(1)}, \dots, f^{(m)}} \subseteq \powerseries \Q x.
    \end{align*}
    For a power series $g \in \powerseries \Q x$ and $e \in \N$, write
    \begin{align*}
        g_e := \coefficient {x_1^e} g \cdot \frac {x_1^e} {e!},
        \quad \coefficient {x_1^e} g \in \powerseries \Q {x_2, \dots, x_d},
    \end{align*}
    so that any power series $g$ can be written as $g = g_0 + g_1 + \cdots$,
    where the family $\set{g_0, g_1, \dots}$ is summable.
    For instance, since $f$ is polynomial in $x_1$ of degree $E$ we can write
    \begin{align}
        \label{eq:representation for f}
        f = f_0 + f_1 + \cdots + f_E.
    \end{align}
    The operation $(\_)_e$ is linear, continuous, and preserves summability,
    for every $e \in \N$.
    We will use the following sum and product rules.
    \begin{claim*}
        For every power series $g, h \in \powerseries \Q x$ and $e \in \N$,
        \begin{align}
            \label{eq:coefficient extraction convolution identity}
            (g + h)_e = g_e + h_e
                \quad\text{and}\quad
                    (g \cdot h)_e = \sum_{a + b = e} g_a \cdot h_b.
        \end{align}
    \end{claim*}
    \begin{proof}[Proof of the claim.]
        The sum rule follows immediately by linearity.
        For the product rule, we have
        \begin{align*}
            (g \cdot h)_e
            &= \left(\sum_{a \in \N} g_a \cdot \sum_{b \in \N} h_b\right)_e
            = \left(\sum_{a, b \in \N} g_a \cdot h_b\right)_e
            = \sum_{a, b \in \N} \left(g_a \cdot h_b\right)_e = \\
            &= \sum_{a, b \in \N} \left(\coefficient {x_1^a} g \cdot \frac {x_1^a} {a!} \cdot \coefficient {x_1^b} h \cdot \frac {x_1^b} {b!}\right)_e = \\
            &= \sum_{a, b \in \N} \coefficient {x_1^e} \left(\coefficient {x_1^a} g \cdot \frac {x_1^a} {a!} \cdot \coefficient {x_1^b} h \cdot \frac {x_1^b} {b!}\right) \cdot \frac {x_1^e} {e!} = \\
            &= \sum_{a, b \in \N} \coefficient {x_1^e} \left(\coefficient {x_1^a} g \cdot \coefficient {x_1^b} h \cdot \frac {x_1^{a+b}} {a! b!}\right) \cdot \frac {x_1^e} {e!} = \\
            &= \sum_{a + b = e} \coefficient {x_1^a} g \cdot \coefficient {x_1^b} h \cdot \frac {x_1^{a+b}} {a!b!} = \\
            &= \sum_{a + b = e} \coefficient {x_1^a} g \cdot \frac {x_1^a} {a!} \cdot \coefficient {x_1^b} h \cdot \frac {x_1^b} {b!}
            = \sum_{a + b = e} g_a \cdot h_b. \qedhere
        \end{align*}
    \end{proof}

    %
    Consider the finitely generated algebra
    \begin{align*}
        \widetilde R := \polyof \Q {f^{(\ell)}_e} {0 \leq e \leq E, 1 \leq \ell \leq m}, 
    \end{align*}
    Two observations are in order.
    First, the $f^{(\ell)}_e$'s are polynomial in $x_1$ by construction.
    Second, if $f^{(\ell)}$ is polynomial in $x_j$, for any $1 \leq j \leq d$,
    then the same is true for $f^{(\ell)}_e$, guaranteeing repeatability of the argument.
    The proof is concluded by the following two claims.
    \begin{claim*}
        $f \in \widetilde R$.
    \end{claim*}
    \begin{proof}[Proof of the claim.]
        Since $f \in R$, the power series $f$ can be written as a polynomial function
        $f = p(f^{(1)}, \dots, f^{(m)})$ of the original generators $f^{(\ell)}$'s,
        for some $p \in \poly \Q m$.
        By the sum and product rules~\cref{eq:coefficient extraction convolution identity}, 
        the term containing $x_1^e$ in $f$
        depends only on the terms containing $x_1^{\leq e}$ in $f^{(1)}, \dots, f^{(m)}$.
        Thus, we can write
        \begin{align*}
            f_e = q\left(\tuple{f^{(\ell)}_a}_{0 \leq a \leq e, 1 \leq \ell \leq m}\right)
        \end{align*}
        for some polynomial $q \in \poly \Q {m \cdot (e+1)}$.
        Therefore $f_e \in \widetilde R$,
        and we conclude $f \in \widetilde R$
        since $f$ admits the representation~\cref{eq:representation for f}.
    \end{proof}
    \begin{claim*}
        The ring $\widetilde R$ is closed under partial derivatives $\partial {x_j}$'s.
    \end{claim*}
    \begin{proof}[Proof of the claim.]
        It suffices to show $\partial {x_j} f^{(\ell)}_e \in \widetilde R$.
        %
        If $e = 0$ and $j = 1$,
        we have $\partial {x_j} f^{(\ell)}_e = 0$ and we are done.
        Now assume $e \geq 1$ or $j \neq 1$.
        We have the following commuting rule
        \begin{align}
            \label{eq:a commuting rule}
            \partial {x_j} f^{(\ell)}_e =
            \left\{\begin{array}{ll}
                (\partial {x_j} f^{(\ell)})_{e-1}
                    &\text{ if } j = 1, \\
                (\partial {x_j} f^{(\ell)})_e
                    &\text{ otherwise.}
            \end{array}\right.
        \end{align}
        Indeed, for $j = 1$ (thus $e \geq 1$ by assumption) we have
        \begin{align*}
            \partial {x_1} f^{(\ell)}_e
            &= \partial {x_1} (\coefficient {x_1^e} f^{(\ell)} \cdot \frac {x_1^e} {e!})
            = \coefficient {x_1^e} f^{(\ell)} \cdot \partial {x_1} (\frac {x_1^e} {e!}) = \\
            &= \coefficient {x_1^e} f^{(\ell)} \cdot \frac {x_1^{e-1}} {(e-1)!}
            = \coefficient {x_1^{e-1}} (\partial {x_1} f^{(\ell)}) \cdot \frac {x_1^{e-1}} {(e-1)!} = \\
            &= (\partial {x_1} f^{(\ell)})_{e-1},
        \end{align*}
        and for $j \neq 1$ we have
        \begin{align*}
            \partial {x_j} f^{(\ell)}_e
            &= \partial {x_j} (\coefficient {x_1^e} f^{(\ell)} \cdot \frac {x_1^e} {e!})
            = \partial {x_j} (\coefficient {x_1^e} f^{(\ell)}) \cdot \frac {x_1^e} {e!} = \\
            &= \coefficient {x_1^e} (\partial {x_j} f^{(\ell)}) \cdot \frac {x_1^e} {e!}
            = (\partial {x_j} f^{(\ell)})_e.
        \end{align*}
        Since $R$ is closed under partial derivatives,
        $\partial {x_j} f^{(\ell)}$ can be written as a polynomial combination
        of the original generators $f^{(1)}, \dots, f^{(m)}$.
        Reasoning as in the previous claim,
        $(\partial {x_j} f^{(\ell)})_e$ can be written as a polynomial combination of new generators
        $f^{(\ell)}_a$ for $0 \leq a \leq e$ and $1 \leq \ell \leq m$.
        The same clearly applies to $(\partial {x_j} f^{(\ell)})_{e-1}$.
        Consequently, $(\partial {x_j} f^{(\ell)})_e$ and $(\partial {x_j} f^{(\ell)})_{e-1}$
        are in $\widetilde R$,
        and thanks to~\cref{eq:a commuting rule} we have $\partial {x_j} f^{(\ell)}_e \in \widetilde R$ as well,
        as required.
    \end{proof}
    \noindent
    The two claims conclude the proof.
\end{proof}

We can finally show closure for \CDF~power series under strong composition.
\begin{restatable}[Closure under strong composition]{lemma}{CDFclosureUnderStrongComposition}
    \label{lem:CDF strong composition}
    Consider commuting variables
    $x = \tuple{x_1, \dots, x_d}$, $y = \tuple{y_1, \dots, y_k}$,
    and strongly $y$-composable power series $f \in \powerseries \Q {x, y}$, $g \in \powerseries \Q x^k$.
    If $f, g$ are \CDF, then $f \compose y g$ is \CDF.
\end{restatable}
%
\begin{proof}
    This is a consequence of the chain rule (\cref{lem:chain rule for power series}).
    Since $f$ is \CDF, by~\cref{lem:CDF characterisation} it belongs to a finitely generated algebra
    \begin{align*}
        f \in R := \poly \Q {f^{(1)}, \dots, f^{(m)}},
            \quad\text{for some } f^{(1)}, \dots, f^{(m)} \in \powerseries \Q {x, y},
    \end{align*}
    closed under $\partial {x_j}$'s and $\partial {y_i}$'s.
    The tuple of power series $g = \tuplesmall{g^{(1)}, \dots, g^{(k)}}$ is also \CDF,
    and we assume \wlg~that $$S := \poly \Q {g^{(1)}, \dots, g^{(k)}}$$ is already closed under $\partial {x_j}$'s.
    (Were this not the case, we could increase $k$ and add dummy variables $y_i$'s to accommodate additional generators).
    We now apply~\cref{lem:chain rule for power series} to $f, g$ and write
    \begin{align*}
        \partial {x_j} (f \compose y g) = \partial {x_j} f \compose y g + \sum_{i = 1}^k \partial {x_j} g^{(i)} \cdot (\partial {y_i} f \compose y g).
    \end{align*}
    This suggests to consider the finitely generated algebra
    \begin{align*}
        T := \poly \Q {f^{(1)} \compose y g, \dots, f^{(m)} \compose y g, g^{(1)}, \dots, g^{(k)}} \subseteq \powerseries \Q x.
    \end{align*}
    First of all, we need to ensure that the generators of $T$ are well defined.
    \begin{claim*}
        We can choose the generators $f^{(\ell)}$'s \st\-
        $f^{(\ell)}, g$ are strongly $y$-composable.
    \end{claim*}
    \noindent
    (Incidentally, this is the problematic step for showing closure under composable \CDF~power series:
    It is not clear how to ensure that the generators $f^{(\ell)}$'s of $f$ can be chosen
    so that if $f$ is locally polynomial in $y_I$, then $f^{(1)}, \dots, f^{(m)}$ are also locally polynomial in $y_I$.)
    \begin{proof}[Proof of the claim.]
        Let $I \subseteq \set{1, \dots, k}$ be the set of indices $i$ \st~$g^{(i)}(0) \neq 0$.
        Since $f, g$ are strongly $y$-composable,
        $f$ is polynomial \wrt~$y_I$,
        \ie, $f \in \poly {\powerseries \Q {x, y_{\setminus I}}} {y_I}$.
        We conclude by \cref{lem:cutting lemma}.
    \end{proof}
    Next, we show that the composition of $f$ and $g$ is indeed in $T$.
    \begin{claim*}
        $f \compose y g \in T$.
    \end{claim*}
    \begin{proof}[Proof of the claim.]
        Since $f \in R$ and $R$ is closed under $\partial {x_j}$'s and $\partial {y_i}$'s,
        we have $\partial {x_j} f, \partial {y_i} f \in R$ as well.
        Consequently,
        \begin{align*}
            \partial {x_j} f \compose y g, \partial {y_i} f \compose y g \in T' := \poly \Q {f^{(1)} \compose y g, \dots, f^{(m)} \compose y g}.
        \end{align*}
        By the same argument, $\partial {x_j} g^{(i)} \in S$.
        We conclude by noticing that the generators of $T$
        are obtained by taking the union of those of $T'$ and $S$.
    \end{proof}
    Finally, we show that $T$ is closed under partial derivatives.
    \begin{claim*}
        $T$ is closed under $\partial {x_j}$'s.
    \end{claim*}
    \begin{proof}[Proof of the claim.]
        By the product rule for power series~\cref{eq:Leibniz rule},
        it suffices to show that $\partial {x_j}$ applied to the generators of $T$ yields> a power series in $T$.
        The argument showing $\partial {x_j} g^{(i)} \in T$ is immediate since $\partial {x_j} g^{(i)} \in S$ and $S \subseteq T$.
        We conclude by showing $\partial {x_j} (f^{(\ell)} \compose y g) \in T$.
        This is another consequence of the chain rule~(\cref{lem:chain rule for power series}),
        this time applied to $f^{(\ell)}, g$:
        \begin{align*}
            \partial {x_j} (f^{(\ell)} \compose y g)
            = \underbrace {\partial {x_j} f^{(\ell)} \compose y g}_{\in T' \subseteq T}
            + \sum_{i = 1}^k \underbrace {\partial {x_j} g^{(i)}}_{\in S \subseteq T} \cdot (\underbrace {\partial {y_i} f^{(\ell)} \compose y g}_{\in T' \subseteq T}).
        \end{align*}
        Since $\partial {x_j} f^{(\ell)} \in R$, we have $\partial {x_j} f^{(\ell)} \compose y g \in T' \subseteq T$.
        Similarly, $\partial {y_i} f^{(\ell)} \compose y g \in T$.
        Also, $\partial {x_j} g^{(i)} \in S \subseteq T$.
        This suffices to show that $\partial {x_j} (f^{(\ell)} \compose y g) \in T$, as required.
    \end{proof}
    By the claims, we have shown that $f \compose y g$ belongs to a finitely generated algebra
    $T \subseteq \powerseries \Q x$ closed under $\partial {x_j}$'s.
    We conclude that $f \compose y g$ is \CDF~by~\cref{lem:CDF characterisation}.
\end{proof}

\subsection{Regular support restrictions}

We present below the formal semantics of constraint expressions introduced in~\cref{sec:CDF support restrictions}:
\begin{align*}
    \sem {z_j = n}
        &= \setof {\tuple {a_1, \dots, a_d} \in \N^d} {a_j = n} \\
    \sem {\equivmod {z_j} n m}
        &= \setof {\tuple {a_1, \dots, a_d} \in \N^d} {\equivmod {a_j} n m} \\
    \sem {\varphi \lor \psi}
        &= \sem \varphi \cup \sem \psi \\
    \sem {\varphi \land \psi}
        &= \sem \varphi \cap \sem \psi \\
    \sem {\lneg \varphi}
        &= \N^d \setminus \sem \varphi.
\end{align*}

While constraint expressions are convenient to describe subsets of $\N^d$,
we now introduce an algebraic formalism which is more suitable in proofs.
A set $S \subseteq \N^d$ is \emph{recognisable} if
there is a finite commutative monoid $\tuple{M, +, 0}$,
with a distinguished subset $F \subseteq M$,
and a homomorphism $h : \N^d \to M$ \st~$S = h^{-1} F$.
Recognisable subsets of $\N^d$ are in bijective correspondence with commutative regular languages
and form a strict subset of semilinear sets.
For instance, the semilinear set $\tuple {1, 1, 1}^* \subseteq \N^3$ is not recognisable.

\begin{lemma}
    \label{lem:regular constraints are recognisable}
    For every set $S \subseteq \N^d$,
    if $S$ is regular then $S$ is recognisable.
\end{lemma}
\noindent
In fact, the converse holds as well and thus regular and recognisable subsets of $\N^d$ coincide,
however we will not need the other direction in the sequel.
\begin{proof}
    Assume $S \subseteq \N^d$ is regular.
    There exists an expression $\varphi$ of dimension $d$ \st~$S = \sem \varphi$.
    We show that $\sem \varphi$ is recognisable by structural induction on $\varphi$.
    
    In the first base case, we have $\varphi \equiv x_j = n$.
    We construct a monoid with elements $0, 1, \dots, n, \infty$
    where element $i$ for $0 \leq i \leq n$ intuitively means ``coordinate $j$ equals $i$'',
    and $\infty$ means ``coordinate $j$ is $\geq n+1$''.
    The monoid operation is ordinary sum in $\N$ when the result is $\leq n$, and it is $\infty$ otherwise.
    The monoid homomorphism $h : \N^d \to M$ maps a vector $x \in \N^d$
    to its $j$-th coordinate $x_j$ if $x_j \leq n$, and $\infty$ otherwise.
    The accepting set is $F = \set n$.    

    In the second base case, we have $\varphi \equiv (\equivmod {x_j} n m)$.
    We assume \wlg~$0 \leq n \leq m - 1$.
    The construction is similar as in the previous case.
    We construct a monoid with elements $0, \dots, m-1$,
    representing equivalence classes modulo $m$.
    Addition in the monoid $a +_M b$ is $(a + b) \mod m$,
    the homomorphism is $h(x) = (x_j \mod m)$ and the accepting set is $F = \set n$.

    In the first inductive case, the constraint is of the form $\varphi \land \psi$.
    By inductive assumption, $\sem \varphi$ is recognised by $M_\varphi, h_\varphi, F_\varphi$
    and $\sem \psi$ by $M_\psi, h_\psi, F_\psi$.
    We construct a product monoid $M := M_\varphi \times M_\psi$
    where all data is defined component-wise:
    addition is $(a_0, b_0) +_M (a_1, b_1) := (a_0 +_{M_\varphi} a_1, b_0 +_{M_\psi} b_1)$,
    the neutral element is $0_M := \tuple{0_{M_\varphi}, 0_{M_\psi}}$,
    the accepting set is $F := F_\varphi \times F_\psi$,
    and the homomorphism is $h(x) := \tuple{h_\varphi(x), h_\psi(x)}$.
    We then have $h^{-1} F = h^{-1} (F_\varphi \times F_\psi) = h_\varphi^{-1} F_\varphi \cap h_\psi^{-1} F_\psi$.

    In the second inductive case, the constraint is of the form $\lneg \varphi$.
    If $M, h, F$ recognises $\sem \varphi$,
    then its complement $\N^d \setminus \sem \varphi$ is recognised by $M, h, M \setminus F$.

    In the last inductive case, the constraint is of the form $\varphi \lor \psi$.
    We can just reduce to the previous two cases by double negation,
    since $\sem {\varphi \lor \psi} = \sem {\lneg (\lneg \varphi \land \lneg \psi)}$.
\end{proof}

\CDFclosureUnderRegularRestrictions*
\begin{proof}
    Let $f^{(1)} \in \powerseries \Q x$ be a \CDF~power series,
    and thus by \cref{lem:CDF characterisation}
    $f^{(1)}$ belongs to a finitely generated algebra
    \begin{align*}
        R := \poly \Q {f^{(1)}, \dots, f^{(k)}}
    \end{align*}
    closed under $\partial {x_j}$'s.
    Let $S$ be a regular subset of $\N^d$,
    which is thus recognisable by virtue of~\cref{lem:regular constraints are recognisable}.
    There exist a finite commutative monoid $M$ and a surjective homomorphism $h : \N^d \to M$
    \st~$S = h^{-1} F$ for some $F \subseteq M$.
    Notice that $\setof {h^{-1} m} {m \in M}$ is a partition of $\N^d$.


    For every $m \in M$ and $1 \leq i \leq k$, let $\restrictsmall{f^{(i)}} m$ be the restriction of $f^{(i)}$
    to $S_m := h^{-1}m \subseteq \N^d$, i.e.,
    %
    \begin{align*}
        \restrictsmall {f^{(i)}} m := \restrictsmall {f^{(i)}} {S_m} = \sum_{n \in h^{-1}m} \coefficient {x^n} {f^{(i)}} \cdot \frac {x^n} {n!}.
    \end{align*}
    By the partition property, every power series $f \in \powerseries \Q x$ can be written as a finite sum of restrictions
    \begin{align}
        \label{eq:power series partition property}
        f = \sum_{m \in M} \restrictsmall f m.
    \end{align}
    First of all, restriction is a linear operation.
    \begin{claim*}
        For two power series $f, g \in \powerseries \Q x$ and $m \in M$, we have
        \begin{align}
            \label{eq:restriction - sum rule}
            \restrict {(f + g)} m = \restrict f m + \restrict g m.
        \end{align}
    \end{claim*}
    \begin{proof}[Proof of the claim]
        We have
        \begin{align*}
            \restrict {(f + g)} m
            &= \restrict {\left(\sum_{n \in \N^d} \coefficient {x^n} f \cdot\frac {x^n} {n!} + \sum_{n \in \N^d} \coefficient {x^n} g \cdot\frac {x^n} {n!}\right)} m = \\
            &= \restrict {\left(\sum_{n \in \N^d} (\coefficient {x^n} f + \coefficient {x^n} g) \cdot\frac {x^n} {n!}\right)} m = \\
            &= \sum_{n \in h^{-1} m} (\coefficient {x^n} f + \coefficient {x^n} g) \cdot\frac {x^n} {n!} = \\
            &= \sum_{n \in h^{-1} m} \coefficient {x^n} f \cdot\frac {x^n} {n!} + \sum_{n \in h^{-1} m} \coefficient {x^n} g \cdot\frac {x^n} {n!} = \\
            &= \restrict f m + \restrict g m. \quad \qedhere
        \end{align*}
    \end{proof}
    The following claim shows that restrictions can be pushed through products.
    \begin{claim*}
        For two power series $f, g \in \powerseries \Q x$,
        we have the following product rule for the restriction operation:
        \begin{align}
            \label{eq:restriction - product rule}
            \restrictsmall {(f \cdot g)} m = \sum_{a + b = m} \restrictsmall f a \cdot \restrictsmall g b,
            \quad \text{for all } m \in M.
        \end{align}
    \end{claim*}
    \begin{proof}[Proof of the claim.]
        By~\cref{eq:power series partition property} applied to $f$ and to $g$ we have
        \begin{align*}
            \restrictsmall {(f \cdot g)} m
                &= \restrict {\left(\left(\sum_{a \in M} \restrictsmall f a\right) \cdot \left(\sum_{b \in M} \restrictsmall g b\right)\right)} m = \\
                &= \restrict {\left(\sum_{a, b \in M} \restrictsmall f a \cdot \restrictsmall g b\right)} m = \\
                &= \sum_{a, b \in M} \restrict {(\restrictsmall f a \cdot \restrictsmall g b)} m = \\
                &= \sum_{a, b \in M} \restrict {\left(\sum_{h(d) = a} f_d \cdot x^d \cdot \sum_{h(e) = b} g_e \cdot x^e\right)} m = \\
                &= \sum_{a, b \in M} \restrict {\left(\sum_{h(d) = a, h(e) = b} f_d g_e \cdot x^{e+d}\right)} m = \\
                &= \sum_{a, b \in M} \sum_{h(d) = a, h(e) = b, h(e+d) = m} f_d g_e \cdot x^{e+d} = \\
                &= \sum_{a + b = m} \sum_{h(d) = a, h(e) = b} f_d g_e \cdot x^{e+d} = \\
                &= \sum_{a + b = m} \sum_{h(d) = a} f_d \cdot x^d \cdot \sum_{h(e) = b} g_e \cdot x^e = \\
                &= \sum_{a + b = m} \restrictsmall f a \cdot \restrictsmall g b. \qedhere
        \end{align*}
    \end{proof}

    \begin{claim*}
        Restrictions compose trivially:
        For all $a, b \in M$ with $a \neq b$,
        \begin{align}
            \label{eq:restriction - restriction rule}
            \restrict {(\restrictsmall f a)} a = \restrictsmall f a
                \quad \text{and} \quad
                    \restrict {(\restrictsmall f a)} b = 0.
        \end{align}
    \end{claim*}
    \begin{proof}[Proof of the claim.]
        The first equality is trivial, and the second one follows from the partition property
        (distinct monoid elements recognise disjoint subsets of $\N^d$).
    \end{proof}

    Consider the new finitely generated algebra
    \begin{align*}
        \widetilde R := \polyof \Q {\restrictsmall {f^{(i)}} m} {m \in M, 1 \leq i \leq k}.
    \end{align*}
    $\restrictsmall {f^{(1)}} S = \sum_{m \in F} \restrictsmall {f^{(1)}} m$,
    thus in $\widetilde R$.
    It remains to argue that $\widetilde R$ is closed under $\partial {x_j}$'s.
    In fact, it suffices to show this for an arbitrary generator $\restrictsmall {f^{(i)}} m$.
    We begin by an observation.
    \begin{claim*}
        We have the following exchange property
        \begin{align}
            \label{eq:partial derivative exchange property}
            \partial {x_j} (\restrictsmall {f^{(i)}} m) = \sum_{m' \in M_j} \restrictsmall {(\partial {x_j} f^{(i)})} {m'},
        \end{align}
        where $M_j := \setof{m' \in M}{m' + h(e_j) = m}$ and $e_j$ is the $j$-th unit vector.
    \end{claim*}
    \begin{proof}[Proof of the claim]
        Justification:
        \begin{align*}
            \partial {x_j} (\restrictsmall {f^{(i)}} m) 
                &= \partial {x_j} \sum_{n \in h^{-1}m} \coefficient {x^n} {f^{(i)}} \cdot \frac {x^n} {n!} \\
                &= \sum_{n \in h^{-1}m} \coefficient {x^n} {f^{(i)}} \cdot \partial {x_j} \frac {x^n} {n!} \\
                &= \sum_{n \in h^{-1}m, n_j \geq 1} \coefficient {x^n} {f^{(i)}} \cdot \frac {x^{n - e_j}} {(n - e_j)!} \\
                &= \sum_{(n + e_j) \in h^{-1}m} \coefficient {x^{n + e_j}} {f^{(i)}} \cdot \frac {x^n} {n!} \\
                &= \sum_{(n + e_j) \in h^{-1}m} \coefficient {x^n} {\partial {x_j} f^{(i)}} \cdot \frac {x^n} {n!} \\
                &= \sum_{h(n + e_j) = h(n) + h(e_j) = m} \coefficient {x^n} {\partial {x_j} f^{(i)}} \cdot \frac {x^n} {n!} \\
                &= \sum_{m' \in M_j} \sum_{n \in h^{-1}m'} \coefficient {x^n} {\partial {x_j} f^{(i)}}\cdot \frac {x^n} {n!} \\
                &= \sum_{m' \in M_j} \restrictsmall {(\partial {x_j} f^{(i)})} {m'}. \qedhere
        \end{align*}
    \end{proof}
    The exchange property from~\cref{eq:partial derivative exchange property}
    shows that $\partial {x_j} (\restrictsmall {f^{(i)}} m)$ is a linear combination of restrictions of $\partial {x_j} f^{(i)}$.
    Since $\partial {x_j} f^{(i)} \in R$, the latter can written as a polynomial combination of the generators $f^{(1)}, \dots, f^{(k)}$.
    Since $f^{(1)}, \dots, f^{(k)} \in \widetilde R$ by the partition property~\cref{eq:power series partition property},
    we also have $\partial {x_j} f^{(i)} \in \widetilde R$.
    The proof is concluded by showing that $\widetilde R$ is closed under restrictions.
    \begin{claim*}
        If $f \in \widetilde R$ then $\restrict f m \in \widetilde R$ for all $m \in M$.
    \end{claim*}
    \begin{proof}[Proof of the claim.]
        This is shown by induction on the structure of witnesses of membership of $f$ in $\widetilde R$.
        The base case for a generator $\restrictsmall {f^{(i)}} m$
        holds since restrictions compose by~\cref{eq:restriction - restriction rule}.
        The inductive step follows by the sum~\cref{eq:restriction - sum rule}
        and product~\cref{eq:restriction - product rule} rules.
    \end{proof}
    \noindent
    This concludes the proof of closure under regular support restrictions.
\end{proof}

\subsection{Implicit power series theorems}

The purpose of this section is to provide a proof of~\cref{lem:CDF constructible nilpotent}
regarding solutions of constructible systems of \CDF~equations.
We will obtain this as a consequence of the study of more general systems of equations, which we now introduce.

The next result shows that the class \CDF~does not change if we allow \CDF~kernels in~the definition of \CDF~power series
(\cf~\cite[Proposition 12]{BergeronSattler:TCS:1995}).
\begin{restatable}[\CDF~differentially constructible power series theorem]{lemma}{CDFdifferentiallyConstructibleTheorem}
    \label{lem:CDF differentially constructible power series theorem}
    Consider a matrix of power series $M \in \matrices k d {\powerseries \Q {x, y}}$
    and a tuple of power series $f \in \powerseries \Q x^k$ \st~$M, f$ are $y$-composable
    and satisfy a constructible system of partial differential equations
    \begin{align}
        \label{eq:CDF differentially constructible power series theorem}
        \partial x f = M \compose y f.
    \end{align}
    If $M$ is \CDF, then $f$ is \CDF.
\end{restatable}
For instance, the unique univariate power series $f$ with $f(0) = 0$
satisfying $\partial x f =  x \cdot e^f$ is \CDF.
The next closure property is a simple application of~\cref{lem:CDF differentially constructible power series theorem},
albeit a very useful one since many power series can be shown to be \CDF~thanks to it.
\begin{restatable}[\CDF~implicit power series theorem]{lemma}{CDFimplicitPowerSeriesTheorem}
    \label{lem:CDF implicit power series theorem}
    Consider $y$-composable tuples of power series
    $f \in \powerseries \Q {x, y}^k, g \in \powerseries \Q x^k$ \st
    %
    \begin{align}
        \label{eq:CDF implicit power series equations}
        f \compose y g = 0.
    \end{align}
    Assume that the Jacobian matrix $\partial y f \in \matrices k k {\powerseries \Q {x, y}}$
    is invertible in the ring of power series matrices
    $(\partial y f)^{-1} \in \powerseries \Q {x, y}^{k \times k}$,
    and moreover $(\partial y f)^{-1}$ is $y$-composable with $g$.
    If $f$ is \CDF, then $g$ is \CDF.
\end{restatable}
\begin{remark}
    We do not know whether the composability assumption on $(\partial y f)^{-1}, g$ can be dropped.
    In applications, this is often guaranteed by $g(0) = 0$.
    Under this condition, $f(0, g(0)) = f(0, 0) = 0$
    and in this case $g$ is the unique power series solution of~\cref{eq:CDF implicit power series equations}.
    This is a consequence of the \emph{implicit power series theorem}
    (\cf~\cite[Theorem 4]{Panholzer:JALC:2005},\cite[Theorem 16.5]{KuichSalomaa:1986}),
    a variation of the \emph{implicit function theorem} well-known from analysis
    (\cf~\cite[Theorem 9 on pg.~39]{BochnerMartin:Book:1948},\cite[Theorem 9.28]{Rudin:Analysis:1976});
    \cf~also~\cite[Theorem 2]{BorealeCollodiGorla:arXiv:2023},
    where a univariate \emph{implicit stream theorem},
    in the special case where $f$ is a polynomial,
    is presented in the context of the \emph{stream calculus}~\cite{Rutten:TCS:2003}.
\end{remark}

\begin{example}
    To illustrate~\cref{lem:CDF implicit power series theorem},
    let $d = k = 1$ and consider $y \cdot f - 1 = 0$ for some \CDF~$f \in \powerseries \Q x$.
    The Jacobian matrix is $\partial y (y \cdot f - 1) = f$
    and the invertibility assumption means $f(0) \neq 0$.
    Thus $f^{-1} \in \powerseries \Q x$ exists and solves the equation.
    Since $f$ is trivially $y$-composable with $f^{-1}$,
    by \cref{lem:CDF implicit power series theorem} $f^{-1}$ is \CDF.
\end{example}


We now prove~\cref{lem:CDF differentially constructible power series theorem}.
\begin{proof}[Proof of~\cref{lem:CDF differentially constructible power series theorem}]
    Let $M = \tuple{M_{ij}} \in \matrices k d {\powerseries \Q {x, y}}$ be a matrix of \CDF~power series.
    Write
    \begin{align*}
        \partial {x_j} f_i = g_{ij},
            \ \text{with }
                g_{ij} := M_{ij} \compose y f \in \powerseries \Q x,
                \ \forall 1 \leq i \leq k, 1 \leq j \leq d.
    \end{align*}
    It suffices to show that the $g_{ij}$'s are \CDF.
    We assume \wlg~that $M_{ij}$ satisfies equations of the form
    \begin{align*}
        \partial {x_h} M_{ij} &= p_{ijh} \compose z M
            \quad \forall 1 \leq h \leq d, \text{ and } \\
        \partial {y_\ell} M_{ij} &= q_{ij\ell} \compose z M,
            \quad \forall 1 \leq \ell \leq k,
    \end{align*}
    for some polynomials $p_{ijh}, q_{ij\ell} \in \poly \Q {x, y, z}$
    and fresh commuting variables $z = \tuple {z_{ij}}$
    with $1 \leq i \leq k$ and $1 \leq j \leq d$.
    (This can be guaranteed since 1) $M_{ij}$ is \CDF~and
    2) we can add auxiliary \CDF~power series to $M$ if necessary.)
    %
    We apply the chain rule from~\cref{lem:chain rule for power series},
    %
    \begin{align*}
        & \partial {x_h} g_{ij} \\
            &= \partial {x_h} (M_{ij} \compose y f) = \\
            &= \partial {x_h} M_{ij} \compose y f
                + \sum_{\ell = 1}^k \left(\partial {y_\ell} M_{ij} \compose y f \right) \cdot \partial {x_h} f_\ell = \\
            &= (p_{ijh} \compose z M) \compose y f
                + \sum_{\ell = 1}^k \left((q_{ij\ell} \compose z M) \compose y f \right) \cdot g_{\ell h} = \\
            &= (p_{ijh} \compose y f) \compose z (M \compose y f)
                + \sum_{\ell = 1}^k \left((q_{ij\ell} \compose y f) \compose z (M \compose y f) \right) \cdot g_{\ell h} = \\
            &= (p_{ijh} \compose y f) \compose z g
                + \sum_{\ell = 1}^k \left((q_{ij\ell} \compose y f) \compose z g \right) \cdot g_{\ell h} = \\
            &= p_{ijh} \compose {y, z} \tuple {f, g}
                + \sum_{\ell = 1}^k \left( q_{ij\ell} \compose {y, z} \tuple {f, g} \right) \cdot g_{\ell h}.
    \end{align*}
    Notice that $\partial {x_h} M_{ij}, f$ and $\partial {y_\ell} M_{ij}, f$ are still $y$-composable by~\cref{lem:composable closure}.
    Moreover, since $p_{ijh}$ is a polynomial and $f$ does not contain $z$,
    $p_{ijh} \compose y f$ is a polynomial in $z$,
    and thus $p_{ijh} \compose y f, M \compose y f$ are $z$-composable.
    Similarly for $(q_{ij\ell} \compose y f), (M \compose y f)$.
    The last equation follows from the fact that $f$ does not contain $z$.
    We have shown that $\partial {x_h} g_{ij}$ is a polynomial function of $x$, $f$, and $g$,
    as required for $g$ to be \CDF.
\end{proof}

Thanks to~\cref{lem:CDF differentially constructible power series theorem} we can now prove~\cref{lem:CDF implicit power series theorem}.
\begin{proof}[Proof of~\cref{lem:CDF implicit power series theorem}]
    Let $g$ be a solution of~\cref{eq:CDF implicit power series equations}.
    The composability assumption means that each entry $f^{(i)}$
    of $f = \tuplesmall{f^{(1)}, \dots, f^{(k)}}$ is $y$-composable with $g$,
    implying that $f^{(i)} \compose y g$ is defined.
    By deriving $f \compose y g$ \wrt~$x$ and applying the chain rule (\cref{lem:chain rule for power series}) in matrix form we obtain
    \begin{align*}
        \underbrace{\partial x (f \compose y g)}_{k \times d}
        = \underbrace {\partial x f}_{k \times d} \compose y g + (\underbrace {\partial y f}_{k \times k} \compose y g) \cdot \underbrace {\partial x g}_{k \times d} = 0.
    \end{align*}
    %
    %
    First of all we show that the \rhs~is well-defined.
    \begin{claim*}
        $\partial x f, g$ and $\partial y f, g$ are $y$-composable.
    \end{claim*}
    \begin{proof}[Proof of the claim.]
        We need to show that $\partial {x_j} f^{(i)}, g$ are $y$-composable.
        By assumption $f^{(i)}, g$ are $y$-composable and we can apply~\cref{lem:composable closure}.
        The argument showing composability of $\partial {y_\ell} f^{(i)}, g$ is analogous.
    \end{proof}
    Since by assumption the Jacobian matrix $\partial y f$
    is invertible as a power series matrix, we can write
    \begin{align*}
        \partial x g = P \compose y g,
        \ \text{where } P := - \underbrace{(\partial y f)^{-1}}_{\matrices k k {\powerseries \Q {x, y}}} \cdot \underbrace {\partial x f}_{\matrices k d {\powerseries \Q {x, y}}} \in \matrices k d {\powerseries \Q {x, y}}.
    \end{align*}
    Notice that $P, g$ are $y$-composable
    since $(\partial y f)^{-1}, g$ are $y$-composable by assumption
    and $\partial x f, g$ are $y$-composable by the claim.
    By Cramer's rule, the entries of $(\partial y f)^{-1}$
    are rational functions of the entries of $\partial y f$.
    Since $f$ is \CDF, by \cref{lem:CDF basic closure properties}
    $\partial x f$, $\partial y f$, and thus $(\partial y f)^{-1}$ are \CDF.
    Thus $P$ is a matrix of \CDF~power series
    and we conclude by~\cref{lem:CDF differentially constructible power series theorem}.
\end{proof}

The proof of~\cref{lem:CDF constructible nilpotent}
relies on two simple facts about matrices,
which we recall below in~\cref{lem:power series matrix invertibility criterion,lem:nilpotent lemma}.
The following lemma provides a simple condition
guaranteeing that $\partial y f$ is invertible as a power series matrix.
\begin{lemma}
    \label{lem:power series matrix invertibility criterion}
    A power series matrix $M \in \matrices k k {\powerseries \Q z}$
    has a power series matrix inverse $M^{-1} \in \matrices k k {\powerseries \Q z}$
    iff $M$ evaluated at the origin $M(0) \in \matrices k k \Q$ is invertible (in $\matrices k k \Q$).
\end{lemma}
\begin{proof}
    The crucial observation is that determinant and evaluation commute:
    \begin{align*}
        (\det M) (0) = \det (M(0)).
    \end{align*}
    We now have a chain of equivalences:
    $M(0)$ is invertible iff
    $\det (M(0)) \neq 0$ iff
    $(\det M) (0) \neq 0$ (the constant term of $\det M \in \powerseries \Q z$ is nonzero) iff
    $\det M$ is invertible in $\powerseries \Q z$ iff
    $M$ is invertible in $\matrices k k {\powerseries \Q z}$.
\end{proof}

\begin{lemma}
    \label{lem:nilpotent lemma}
    Consider a rational matrix $M \in \matrices k k \Q$.
    If $M$ is nilpotent, then $I - M$ is invertible.
\end{lemma}
\begin{proof}
    Since $M$ is nilpotent, there is $m \in \N$ \st~$M^m = 0$.
    Since the matrix series $M^* := \sum_{\ell = 0}^\infty M^\ell$ exists
    (in fact it equals $M^0 + M^1 + M^2 + \cdots + M^{m-1}$),
    and $(I-M) \cdot M^* = M^* - M \cdot M^* = I$,
    we conclude that $M^*$ is the inverse of $I - M$,
    and thus the latter is invertible.
\end{proof}


\constructiblePowerSeriesTheorem*
\begin{proof}
    Let $g \in \powerseries \Q x^k$ be a solution of a constructible system $y = f(x, y)$.
    The Jacobian of $\widetilde f := y - f$ is
    $\partial y \widetilde f = I - \partial y f$,
    where $I$ is the $k \times k$ identity matrix.
    By the well-posedness assumption, $\partial y f$ evaluated at the origin is nilpotent.
    By~\cref{lem:nilpotent lemma}, $\partial y \widetilde f$ evaluated at the origin is invertible.
    By~\cref{lem:power series matrix invertibility criterion},
    $\partial y \widetilde f$ has a power series matrix inverse $(\partial y \widetilde f)^{-1}$.
    The condition $g(0)=0$ guarantees $y$-composability of $g$ with $(\partial y \widetilde f)^{-1}$.
    We conclude by~\cref{lem:CDF implicit power series theorem} applied to $\widetilde f, g$.

    Existence and uniqueness of canonical solutions is classic (\cf, \eg,~\cite[Theorem 6]{Joyal:AM:1981}).
    We can also obtain uniqueness from the previous construction,
    since the \CDF~system for $g$ has a unique solution for every fixed initial condition $g(0)$.
\end{proof}

\subsection{Lie derivatives and exchange rule}

The purpose of this section is to show a formula for the computation of the coefficients of a \CDF~power series (\cref{lem:CDF exchange rule}).
This will be used in~\cref{app:computing CDF coefficients}
to obtain algorithms for computing coefficients and deciding zeroness of~\CDF~power series, and analyse their complexity. 
In order to do this, we recall the notion of Lie derivative.

For every $1 \leq j \leq d$,
define the \emph{$j$-th Lie derivative} of an autonomous \CDF~system of equations~\cref{eq:multivariate CDF - matrix form} as
\begin{align*}
    L_j := \sum_{h = 1}^k P_{hj} \cdot \partial {y_h} : \poly \Q y \to \poly \Q y.
\end{align*}
%
%
Since $L_j$ is a linear combination of derivations $\partial {y_h}$'s,
it is itself a derivation.
%
%
Extend $L$ to sequences of indices as $L_\e = 1$ (the identity operator)
and $L_{j \cdot w} = L_w L_j$ (operator composition),
for every $w \in \set{1, \dots, d}^*$.
%
%

We show a formula to express coefficients of $p \compose y f$
as a function of the initial condition $f(0)$.
This will be useful
in the zeroness algorithm (\cref{sec:CDF zeroness}).
\begin{restatable}{lemma}{CDFexchangeRule}
    \label{lem:CDF exchange rule}
    For any solution $f$ of~\cref{eq:multivariate CDF - matrix form},
    $w \in \set{1, \dots, d}^*$, $p \in \poly \Q y$:
    \begin{align*}
        \partial x^{\Parikh w} (p \compose y f)
            = L_w p \compose y f \text{ and }
        \coefficient {x^{\Parikh w}} (p \compose y f)
            = L_w p \compose y f(0).
    \end{align*}
\end{restatable}

\begin{proof}
    %
    The first equality is proved by induction.
    The base case $w = \e$ is clear,
    since both $\partial x^0$ and $L_\e$ are the identity.
    For the inductive step,
    \begin{align*}
        \partial x^{\Parikh {j \cdot w}} (p \compose y f)
        &= \partial x^{\Parikh w} \partial {x_j} (p \compose y f)
        = \partial x^{\Parikh w} (L_j p \compose y f) = \\
        &= L_w L_j p \compose y f
        = L_{j \cdot w} p \compose y f.
    \end{align*}
    For the second equality, we compute:
    \begin{align*}
        &\coefficient {x^{\Parikh w}} (p \compose y f)
            &&\text{(by~\cref{eq:coefficient extraction and derivation})} \\
        &= \coefficient {x^0} (\partial x^{\Parikh w} (p \compose y f)) = 
            &&\text{(by the first equality)} \\
        &= \coefficient {x^0} (L_w p \compose y f) =
            &&\text{($\coefficient {x^0}$ homomorphism)} \\
        &= L_w p \compose y \coefficient {x^0} f = \\
        &= L_w p  \compose y f(0). \qedhere
    \end{align*}
\end{proof}


\subsection{\CDF~and commutative \WBPP}

In this section we provide a proof of the fact that \CDF~power series coincide with \WBPP~series which are commutative.

\commutativeShuffleFinite*
\begin{proof}
    For the first direction, let $f \in \series \Q \Sigma$ be a commutative shuffle-finite series.
    Thanks to~\cref{lem:shuffle-finite working characterisation}, we have
    \begin{align*}
        f = f^{(1)} \in \poly \Q {f^{(1)}, \dots, f^{(k)}},
    \end{align*}
    where the latter ring is closed under $\derive {a_j}$, for every $1 \leq j \leq d$.
    Let $g^{(i)} := \stops {f^{(i)}} = \sum_{n \in \N^d} f^{(i)}_{w_n} \cdot \frac {x^n} {n!}$
    be the corresponding commutative series, for every $1 \leq i \leq k$.
    Here, for every $n \in \N^d$, we have fixed a word $w_n \in \Sigma^*$ \st~$\Parikh {w_n} = n$.
    Consider the power series ring
    \begin{align*}
        R := \poly \Q {g^{(1)}, \dots, g^{(k)}}.
    \end{align*}
    Clearly $\stops {f^{(1)}} \in R$.
    It remains to show that $R$ is closed under $\partial {x_j}$, for every $1 \leq j \leq d$.
    
    We begin with a combinatorial equality which will be useful in the rest of the proof.
    \begin{claim*}
        For every $x, y \in \N^d$, we have
        \begin{align}
            \label{eq:shuffle combinatorial identity}
            \sum_{\Parikh u = x, \Parikh v = y} (u \shuffle v) = \sum_{\Parikh w = x + y} {\Parikh w \choose x, y} \cdot w,
        \end{align}
        where the \lhs~sum is over $u, v \in \Sigma^*$ and the \rhs~sum is over $w \in \Sigma^*$.
    \end{claim*}
    \begin{proof}[Proof of the claim.]
        Since the shuffle product enjoys bilinearity, associativity, commutativity, and distributivity over sum, we have
        \begin{align*}
            &\sum_{\Parikh u = x, \Parikh v = y} (u \shuffle v) = \\
            &= \left(\sum_{\Parikh u = x} u\right) \shuffle \left(\sum_{\Parikh v = y} v\right) = \\
            &= \left(a_1^{x_1} \shuffle \cdots \shuffle a_d^{x_d}\right) \shuffle \left(a_1^{y_1} \shuffle \cdots \shuffle a_d^{y_d}\right) = \\
            &= (a_1^{x_1} \shuffle a_1^{y_1}) \shuffle \cdots \shuffle (a_d^{x_d} \shuffle a_d^{y_d}) = \\
            &= \left({x_1 + y_1 \choose x_1, y_1} a_1^{x_1 + y_1}\right) \shuffle \cdots \shuffle \left({x_d + y_d \choose x_d, y_d} a_d^{x_d + y_d}\right) = \\
            &= {x_1 + y_1 \choose x_1, y_1} \cdots {x_d + y_d \choose x_d, y_d} \cdot \left(a_1^{x_1 + y_1} \shuffle \cdots \shuffle a_d^{x_d + y_d}\right) = \\
            &= {x + y \choose x, y} \cdot \left(a_1^{x_1 + y_1} \shuffle \cdots \shuffle a_d^{x_d + y_d}\right) = \\
            &= {x + y \choose x, y} \cdot \sum_{\Parikh w = x + y} w = \\
            &= \sum_{\Parikh w = x + y} {\Parikh w \choose x, y} \cdot w. \qedhere
        \end{align*}
    \end{proof}

    In the next claim we show a commutation rule for partial derivative $\partial {x_j}$
    \wrt~$\stops \_$ and shift $\derive {a_j}$.
    \begin{claim*}
        For every commutative series $f \in \series \Q \Sigma$ and $1 \leq j \leq d$,
        \begin{align}
            \label{eq:partial stops derive}
            \partial {x_j} (\stops f) = \stops {\derive {a_j} f}.
        \end{align}
    \end{claim*}
    \begin{proof}[Proof of the claim]
        \begin{align*}
            \partial {x_j} (\stops f)
            &= \partial {x_j} \sum_{n \in \N^d} f_{w_n} \cdot \frac {x^n} {n!} = \\
            &= \sum_{n \in \N^d} f_{w_n} \cdot \partial {x_j} \frac {x^n} {n!} = \\
            &= \sum_{n \in \N^d, n_j \geq 1} f_{w_n} \cdot \frac {x^{n - e_j}} {(n - e_j)!} = \\
            &= \sum_{n \in \N^d} f_{w_{n + e_j}} \cdot \frac {x^n} {n!} = (*) \\
            &= \sum_{n \in \N^d} f_{a_j \cdot w_n} \cdot \frac {x^n} {n!} = \\
            &= \sum_{n \in \N^d} (\derive {a_j} f)_{w_n} \cdot \frac {x^n} {n!} = \\
            &= \stops {\derive {a_j} f}.
        \end{align*}
        The crucial equality is $(*)$ where we have used commutativity to deduce
        $f_{w_{n + e_j}} = f_{a_j \cdot w_n}$.
    \end{proof}

    In the next claim we show that $\stops \_$ is a homomorphism
    from the ring of commutative shuffle series to the ring of power series.
    \begin{claim*}
        For every commutative series $f, g \in \series \Q \Sigma$,
        \begin{align}
            \label{eq:stops sum}
            \stops {f + g} = \stops f + \stops g, \\
            \label{eq:stops shuffle}
            \stops {f \shuffle g} = \stops f \cdot \stops g.
        \end{align}
    \end{claim*}
    \begin{proof}[Proof of the claim]
        Additivity follows immediately from the definitions.
        We now show~\cref{eq:stops shuffle}.
        For every word $w \in \Sigma^*$, write $f_{\Parikh w} := f_w$ and $g_{\Parikh w} := g_w$,
        which is well defined since $f, g$ are commutative.
        \begin{align*}
            \stops {f \shuffle g}
            &= \stops {\left(\sum_{u \in \Sigma^*} f_u \cdot u\right) \shuffle \left(\sum_{v \in \Sigma^*} g_v \cdot v\right)} = \\
            &= \stops {\left(\sum_{a \in \N^d} f_a \cdot \sum_{\Parikh u = a} u\right) \shuffle \left(\sum_{b \in \N^d} g_b \cdot \sum_{\Parikh v = b} v\right)} = \\
            &= \stops {\sum_{a, b \in \N^d} f_a f_b \cdot \sum_{\Parikh u = a, \Parikh v = b} u \shuffle v}
                = \text{(by~\cref{eq:shuffle combinatorial identity})} \\
            &= \stops {\sum_{a, b \in \N^d} f_a f_b \cdot \sum_{\Parikh w = a + b} {a+b \choose a, b} \cdot w} = \\
            &= \stops {\sum_{a, b \in \N^d} f_a f_b \cdot {a+b \choose a, b} \cdot \sum_{\Parikh w = a + b} w} = \\
            &= \sum_{a, b \in \N^d} f_a f_b \cdot {a+b \choose a, b} \cdot \frac {x^{a+b}} {(a+b)!} = \\
            &= \left(\sum_{a \in \N^d} f_a \cdot \frac {x^a} {a!}\right) \cdot \left(\sum_{b \in \N^d} f_b \cdot \frac {x^b} {b!}\right) = \\
            &= \stops f \cdot \stops g. \qedhere
        \end{align*}
    \end{proof}

    We can now apply the previous two claims and compute
    \begin{align*}
        \partial {x_j} g^{(i)}
            &= \partial {x_j} (\stops {f^{(i)}}) = 
                &&\text{(by~\cref{eq:partial stops derive})} \\
            &= \stops {\derive {a_j} f^{(i)}} = \\
            &= \stops {p_j^{(i)} \compose y \tuple{f^{(1)}, \dots, f^{(k)}}} =
                &&\text{(by~\cref{eq:stops sum,eq:stops shuffle})} \\
            &= p_j^{(i)}(\stops {f^{(1)}}, \dots, \stops {f^{(k)}}) = \\
            &= p_j^{(i)}(g^{(1)}, \dots, g^{(k)}) \in \poly \Q {g^{(1)}, \dots, g^{(k)}}.
    \end{align*}
    This shows that $\partial {x_j} g^{(i)}$ belongs to $\poly \Q {g^{(1)}, \dots, g^{(k)}}$, as required.

    For the other direction, let $f \in \powerseries \Q x$ be a \CDF~power series.
    Thanks to~\cref{lem:CDF characterisation}, we have
    \begin{align*}
        f = f^{(1)} \in \poly \Q {f^{(1)}, \dots, f^{(k)}},
    \end{align*}
    where the latter ring is closed under $\partial {x_j}$, for every $1 \leq j \leq d$.
    Write $f^{(i)} = \sum_{n \in \N^d} f^{(i)}_n \cdot \frac {x^n} {n!}$
    and let $g^{(i)} := \pstos {f^{(i)}} = \sum_{w \in \Sigma^*} f^{(i)}_{\Parikh w} \cdot w$
    be the corresponding commutative series, for every $1 \leq i \leq k$.
    We claim that
    \begin{align*}
        \poly \Q {g^{(1)}, \dots, g^{(k)}}
    \end{align*}
    is closed under $\derive {a_j}$, for every $1 \leq j \leq d$.
    It suffices to show this for the generators.

    The following two claims show that $\pstos \_$ is a homomorphism from the power series to the series ring
    which respects their differential structure.
    \begin{claim*}
        For every power series $f \in \powerseries \Q x$ and $1 \leq j \leq d$, we have
        \begin{align}
            \label{eq:derive pstos partial}
            \derive {a_j} (\pstos f) = \pstos {\partial {x_j} f}.
        \end{align}
    \end{claim*}
    \begin{proof}[Proof of the claim.]
        We compute:
        \begin{align*}
            \derive {a_j} (\pstos f)
            &= \derive {a_j} \sum_{w \in \Sigma^*} f_{\Parikh w} \cdot w = \\
            &= \sum_{w \in \Sigma^*} f_{\Parikh w} \cdot \derive {a_j} w = \\
            &= \sum_{w \in \Sigma^*} f_{\Parikh {a_j \cdot w}} \cdot \derive {a_j} (a_j \cdot w) = \\
            &= \sum_{w \in \Sigma^*} f_{\Parikh w + e_j} \cdot w = \\
            &= \sum_{w \in \Sigma^*} (\partial {x_j} f)_{\Parikh w} \cdot w = \\
            &= \pstos {\partial {x_j} f}. \qedhere
        \end{align*}
    \end{proof}
    \begin{claim*}
        For every two power series $f, g \in \powerseries \Q x$, we have
        \begin{align}
            \label{eq:pstos sum}
            \pstos {f + g} = \pstos f + \pstos g, \\
            \label{eq:pstos shuffle}
            \pstos {f \cdot g} = \pstos f \shuffle \pstos g.
        \end{align}
    \end{claim*}
    \begin{proof}[Proof of the claim.]
        Additivity follows immediately from the definitions.
        Regarding~\cref{eq:pstos shuffle}, we have
        \begin{align*}
            \pstos f \shuffle \pstos g
            &= \left(\sum_{u \in \Sigma^*} f_{\Parikh u} \cdot u\right) \shuffle \left(\sum_{v \in \Sigma^*} g_{\Parikh v} \cdot v\right) = \\
            &= \sum_{u, v \in \Sigma^* } f_{\Parikh u} g_{\Parikh v} \cdot (u \shuffle v) = \\
            &= \sum_{x, y \in \N^d} f_x g_y \sum_{\Parikh u = x, \Parikh v = y} (u \shuffle v)
                = \text{(by~\cref{eq:shuffle combinatorial identity})} \\
            &= \sum_{x, y \in \N^d} f_x g_y \sum_{\Parikh w = x + y} {\Parikh w \choose x, y} \cdot w = \\
            &= \sum_{w \in \Sigma^*} \sum_{x + y = \Parikh w} {\Parikh w \choose x, y} f_x g_y \cdot w = \\
            &= \pstos {\sum_{n \in \N^d} \sum_{x + y = n} {n \choose x, y} f_x g_y \cdot x^n} = \\
            &= \pstos {f \cdot g}. \qedhere
        \end{align*}
    \end{proof}
    Thanks to the last two claims we have
    \begin{align*}
        &\derive {a_j} g^{(i)} = \\
        &= \derive {a_j} (\pstos {f^{(i)}}) =
            && \text{(by~\cref{eq:derive pstos partial})}\\
        &= \pstos {\partial {x_j} f^{(i)}} = \\
        &= \pstos {p_j^{(i)}(f^{(1)}, \dots, f^{(k)})} = 
            && \text{(by~\cref{eq:pstos sum,eq:pstos shuffle})} \\
        &= p_j^{(i)} \compose y \tuple{\pstos{f^{(1)}}, \dots, \pstos{f^{(k)}}} = \\
        &= p_j^{(i)} \compose y \tuple{g^{(1)}, \dots, g^{(k)}} \in \poly \Q {g^{(1)}, \dots, g^{(k)}}.
    \end{align*}
    This shows that $\derive {a_j} g^{(i)}$ is in the ring $\poly \Q {g^{(1)}, \dots, g^{(k)}}$, as required.
\end{proof}

\subsection{Number sequences and binomial convolution}

In this short section we recall a classical connection between power series and numeric sequences.
Let the set of number sequences $\N^d \to \Q$ be denoted for brevity by $\sequences \Q d$.
Consider two number sequences $f, g \in \sequences \Q d$.
Their sum $f + g$ is the sequence obtained by pointwise sum,
and the sequence $c \cdot f$ for a scalar $c \in \Q$ is obtained by multiplying by $c$ the value of $f$.
This endows the set of number sequences with the structure of a vector space.
The \emph{binomial convolution} of two sequences $f, g$ is the number sequence $f \binconv g$
\st~for every $n \in \N^d$,
\begin{align}
    \label{eq:binomial convolution}
    (f \binconv g)_n := \sum_{m \leq n} \binom n m f_m g_{m - n},
        \ \text{with } \binom n m := \binom {n_1} {m_1} \cdots \binom {n_d} {m_d}.
\end{align}
Binomial convolution is associative, commutative,
with identity the sequence $f$ \st~$f_0 = 1$ and zero everywhere else.
In this way we have endowed number sequences with the structure of a commutative ring.
Since binomial convolution is bilinear, it also endows number sequences with the structure of an algebra over $\Q$.

The \emph{$j$-th shift} of a sequence $f \in \sequences \Q d$
is the sequence $\shift j f$ which maps $n$ to $f_{n + e_j}$,
where $n + e_j$ is obtained from $n$ by increasing the $j$-th coordinate by $1$.
The shifts $\shift j$'s are pairwise commutative
and endow the sequences ring / algebra with a differential structure since they obey Leibniz rule
\begin{align}
    \shift j (f \binconv g) = \shift j f \binconv g + f \binconv \shift j g.
\end{align}

\begin{definition}
    A sequence $f \in \sequences \Q d$ is \emph{binomial-convolution finite}
    if it belongs to a finitely generated subalgebra of $\tuple{\sequences \Q d; {+}, {\binconv}}$
    closed under shifts $\shift j$'s for every $1 \leq j \leq d$.
\end{definition}

In analogy with other composition operators,
define the composition of a polynomial $p \in \poly \Q k$
and a tuple of number sequences $f \in \sequences \Q d^k$
as the number sequence $p \circ f$ obtained by replacing the $i$-th variable in $p$ by $f^{(i)}$
and interpreting polynomial multiplication by binomial convolution.
The following is our working definition of binomial-convolution finite sequences.
\begin{lemma}
    \label{lem:binomial-convolution-finite working characterisation}
    A number sequence $f^{(1)} \in \sequences \Q d$ is binomial-convolution finite iff
    there exist number sequences $f^{(2)}, \dots, f^{(k)} \in \sequences \Q d$,
    and polynomials $p_j^{(i)} \in \poly \Q k$ for every $1 \leq j \leq d$ and $1 \leq i \leq k$,
    \st~for every $1 \leq j \leq d$:
    \begin{align}
        \label{eq:binomial-convolution-finite equations}
        \left\{
            \begin{array}{rcl}
                \shift j f^{(1)} &=& p^{(1)}_j \circ \tuple{f^{(1)}, \dots, f^{(k)}}, \\
                &\vdots& \\
                \shift j f^{(k)} &=& p^{(k)}_j \circ \tuple{f^{(1)}, \dots, f^{(k)}}.
            \end{array}
        \right.
    \end{align}
\end{lemma}

\subsubsection{Connection with \CDF~power series}

At this point, the analogy between binomial-convolution finite sequences and \CDF~power series should be evident.
The \emph{exponential generating power series} of a number sequence $f \in \sequences \Q d$
is the power series
\begin{align*}
    \EGS f := \sum_{n \in \N} f_n \cdot \frac {x^n} {n!}.
\end{align*}
In the other direction, we can transform a power series $f \in \powerseries \Q d$
to a number sequence $g \in \sequences \Q d$
with the coefficient extraction operation $g_n := \coefficient {x^n} f$.
Clearly, these two operations are one the inverse of the other.
There are a number of notable identities connecting the $\EGS \_$ transform
and operations on number sequences.
\begin{center}
    \begin{tabular}{ C C }
        \sequences \Q d & \powerseries \Q d \\
        \hline
        f & \EGS f \\
        c \cdot f & c \cdot \EGS f \\
        f + g & \EGS f + \EGS g \\
        f \binconv g & \EGS f \cdot \EGS g \\
        \shift j f & \partial {x_j} \EGS f \\
        p \circ \tuple{f^{(1)}, \dots, f^{(k)}} & p \circ \tuple{\EGS {f^{(1)}}, \dots, \EGS {f^{(k)}}}
    \end{tabular}
\end{center}
The connection between binomial-convolution finite sequences
and \CDF~power series stated below
follows immediately from the identities above.
\begin{lemma}
    A number sequence $f \in \sequences \Q d$ is binomial-convolution finite iff
    its exponential power series $\EGS f \in \powerseries \Q x$ is \CDF.
\end{lemma}
This interpretation will be useful
when computing \CDF~power series coefficients in~\cref{app:computing CDF coefficients}.

\subsection{Computing coefficients of \CDF~power series}
\label{app:computing CDF coefficients}

In this section we provide more details
on the complexity of computing \CDF~power series coefficients.

\subsubsection{Coefficient growth}
\label{sec:CDF coefficient growth}

We show that coefficients and denominators of a \CDF~power series do not grow too fast.
We being with a bound on the growth of coefficients.
\begin{restatable}{lemma}{SmallCoefficientProperty}
    \label{lem:CDF small coefficient}
    Let $f \in \powerseries \Q d^k$ be a tuple of \CDF~power series
    satisfying a \CDF~system of degree $\leq D$, order $\leq k$, height $\leq H$,
    and with initial value $\inftynorm {f(0)} \leq H$.
    For every $n \in \N^d$,
    \begin{align*}
        \inftynorm {\coefficient {x^n} f} \leq (\onenorm nD + k)^{\bigO{\onenorm n \cdot k}} \cdot H^{\bigO {\onenorm n \cdot D}}.
    \end{align*}
\end{restatable}

\begin{remark}
    For univariate power series $f$ solving an algebraic differential equation,
    a bound $\abs {\coefficient {x^n} f / n!} \leq (n!)^{\bigO 1}$ 
    was previously known~\cite[Theorem 16]{Mahler:1976}.
    For univariate \CDF~power series,
    \cite[Theorem 3]{BergeronReutenauer:EJC:1990} shows the more precise bound $\abs {\coefficient {x^n}f} \leq \alpha^n \cdot n!$,
    later generalised to the multivariate case~\cite[Theorem 11]{BergeronSattler:TCS:1995}.
    None of these bounds provide details on the precise dependency of the exponent on the parameters of the system,
    such as degree $D$, height $H$, and order $k$,
    which we need in our complexity analysis.
\end{remark}
\begin{proof}
	By~\cref{lem:CDF exchange rule},
	the coefficient of $x^n$ of component $f^{(i)}$ of a solution $f$ 	of~\cref{eq:multivariate CDF - matrix form}
	can be expressed as
	$\coefficient {x^n} f^i = \coefficient {x^n} (y_i \compose y f) = L_w y_i (f(0))$,
	for any $w \in \Sigma^*$ \st~$\Parikh w = n$.
	The bound follows as a corollary of~\cref{eq:bound on evaluationing iterated derivations} in \cref{lem:iterated derivation bounds} regarding the complexity of $L_w y_i$.
\end{proof}

We now state a bound on the growth of denominators of \CDF~power series coefficients.
The following bound has appeared in~\cite[Theorem 3(iii)]{BergeronReutenauer:EJC:1990} in the univariate context,
later generalised to the multivariate context in~\cite[Theorem 14(iii)]{BergeronSattler:TCS:1995}.
Since no proof appeared in the multivariate context,
we provide one for completeness.
\begin{restatable}{lemma}{CDFdenominatorBound}
    \label{lem:CDF denominator bound}
    Let $f \in \powerseries \Q d^k$ with $f_0 = 0$ be a tuple of \CDF~power series
    solving~\cref{eq:multivariate CDF - matrix form}.
    Let $q \in \N$ be a common denominator of all rational constants
    appearing in the kernel $P_{ij} \in \poly \Q y$.
    We have
    \begin{align}
        \label{eq:CDF denominator bound}
        q^{\onenorm n} \cdot \coefficient {x^n} f \in \Z^k, \quad \forall n \in \N^d.
    \end{align}
\end{restatable}

\begin{proof}
    There are integer polynomials $Q_{ij} \in \poly \Z y$ \st~$P_{ij} = q^{-1} \cdot Q_{ij}$.
    By extracting the coefficient of $x^n$ on both sides of~\cref{eq:multivariate CDF - matrix form}
    and using~\cref{eq:coefficient extraction and derivation}, we have
    \begin{align}
        \label{eq:coefficient extraction - denominator bound}
        \coefficient {x^{n+e_j}} f^{(i)} = q^{-1} \cdot \coefficient {x^n} (Q_{ij} \compose y f).
    \end{align}
    We show~\cref{eq:CDF denominator bound} by induction on $\onenorm n$.
    The base case $\onenorm n = 0$ holds by the assumption $f_0 = 0 \in \N$.
    For the inductive step, consider a coordinate $1 \leq j \leq d$ and write
    \begin{align*}
        &q^{\onenorm {n+e_j}} \cdot \coefficient {x^{n+e_j}} f^{(i)} \\
        &= q^{\onenorm n + 1} \cdot q^{-1} \cdot \coefficient {x^n} (Q_{ij} \compose y f) = \\
        &= q^{\onenorm n} \cdot \coefficient {x^n} (Q_{ij} \compose y f).
    \end{align*}
    By the correspondence between power series product
    and binomial convolution of number sequences,
    $\coefficient {x^n} (Q_{ij} \compose y f)$ is a sum with integer coefficients
    of products of the form $\coefficient {x^{n_1}}f^{(i_1)} \cdots \coefficient {x^{n_\ell}}f^{(i_\ell)}$
    where $n_1 + \cdots + n_\ell = n$.
    By inductive assumption, for each factor we have
    \begin{align*}
        q^{\onenorm {n_1}} \cdot \coefficient {x^{n_1}}f^{(i_1)}, \dots, q^{\onenorm {n_\ell}} \cdot \coefficient {x^{n_\ell}}f^{(i_\ell)} \in \Z,
    \end{align*}
    and thus for their product
    \begin{align*}
        q^{\onenorm n} \cdot \coefficient {x^{n_1}}f^{(i_1)} \cdots \coefficient {x^{n_\ell}}f^{(i_\ell)} \in \Z
    \end{align*}
    as well.
\end{proof}

\subsubsection{Complexity of computing \CDF~coefficients}

We now provide complexity bounds on computing coefficients of \CDF~power series
satisfying systems of \CDF~equations with integer coefficients.

\CDFtableLemma*
\begin{proof}
    %
    By the stars and bars argument in combinatorics
    the number of $d$-variate monomials of total degree equal to $N$ is $\binom {N + d - 1} {d - 1}$.
    The number of $d$-variate monomials of total degree $\leq N$
    can be counted by adding an extra variable and count the number of monomials of total degree equal to $N$,
    which consequently is $\binom {N + d} d$.
    
    We build a table of size $\leq k \cdot \binom {N + d} d$
    where we store all values $f_n := \coefficient {x^n} f \in \Z^k$ required during the computation,
    for all $\onenorm n \leq N$.
    Indeed, for each sequence $f^{(i)}$ (with $1 \leq i \leq k$)
    the number of coefficients $f^{(i)}_n = \coefficient {x^n} f^{(i)}$ with $\onenorm n \leq N$
    equals the number of $d$-variate monomials of total degree $\leq N$.
    %
    %
    Generalising~\cref{eq:binomial convolution},
    we can compute a specific entry $n \in \N^d$ of a convolution product
    $g := g^{(1)} \binconv \cdots \binconv g^{(e)}$ with the formula
    \begin{align}
        \label{eq:multinomial convolution}
        g_n = \sum_{m^{(1)} + \cdots + m^{(e)} = n} \binom n {m^{(1)}, \dots, m^{(e)}} g^{(1)}_{m^{(1)}} \cdots g^{(e)}_{m^{(e)}},
    \end{align}
    where the vectorial multinomial coefficient is defined as
    \begin{align}
        \label{eq:multinomial coefficient}
        \binom n {m^{(1)}, \dots, m^{(e)}} := \frac {n!} {m^{(1)}! \cdots m^{(e)}!} \in \N.
    \end{align}
    \begin{claim*}
        The multinomial coefficient~\cref{eq:multinomial coefficient} can be computed with $\bigO N$ arithmetic operations.
    \end{claim*}
    \begin{proof}
        We expand the vectorial notation in~\cref{eq:multinomial coefficient}
        by writing $n = \tuple{n_1, \dots, n_d}, m^{(\ell)} = \tuple{m^{(\ell)}_1, \dots, m^{(\ell)}_d} \in \N^d$
        and we have
        \begin{align*}
            \binom n {m^{(1)}, \dots, m^{(e)}} = \frac {n_1!} {m^{(1)}_1! \cdots m^{(e)}_1!} \cdots \frac {n_d!} {m^{(1)}_d! \cdots m^{(e)}_d!}
        \end{align*}
        where, for every $1 \leq j \leq d$,
        we have the balancing condition $m^{(1)}_j + \cdots + m^{(e)}_j = n_j$.
        Computing the numerator requires $n_1 + \cdots + n_d = \onenorm n$ products,
        and the same is true for the denominator, thanks to the balancing condition.
        Thus $2 \cdot \onenorm n$ products suffice and at the end we perform one integer division.
    \end{proof}
    \begin{claim*}
        For every $n \in \N^d$ and $e \in \N$ with $\onenorm n \leq N$ and $e \leq D$,
        we can compute~\cref{eq:multinomial convolution}
        with $\bigO {(N + D) \cdot \binom {N + d \cdot D - 1} {d \cdot D - 1}}$ arithmetic operations.
    \end{claim*}
    \begin{proof}[Proof of the claim]
        Each summand of~\cref{eq:multinomial convolution} requires $e - 1 \leq D - 1$ binary products to compute
        $g^{(1)}_{m^{(1)}} \cdots g^{(e)}_{m^{(e)}}$
        and $\bigO N$ arithmetic operations for the multinomial coefficient (by the previous claim).
        We now estimate the number of summands of~\cref{eq:multinomial convolution}.
        Every tuple of vectors $\tuplesmall{m^{(1)}, \dots, m^{(e)}} \in (\N^d)^e$
        corresponds to a monomial in $d \cdot e$ variables.
        Since these vectors sum up to $n \in \N^d$,
        the total degree of the monomial is $\onenorm n$.
        It follows that there are $\leq \binom {\onenorm n + d \cdot e - 1} {d \cdot e - 1}$ such summands.
        The latter quantity is monotone in $\onenorm n$ and $e$,
        and thus it is bounded by $\leq \binom {N + d \cdot D - 1} {d \cdot D - 1}$.
    \end{proof}

    Suppose a certain portion of the table has already been filled in,
    and consider an index $n \in \N^d$
    \st~$\onenorm n < N$, $f_m \in \Z^k$ is already in the table for every $m \leq n$,
    but $f^{(i)}_{n + e_j}$ is not in the table for some $1 \leq j \leq d$ and $1 \leq i \leq k$.
    Extract the term $n \in \N^d$ on both sides of~\cref{eq:binomial-convolution-finite equations}:
    \begin{align}
        \label{eq:binomial-convolution-finite recursive relation}
        f^{(i)}_{n + e_j} = \tuple{p^{(i)}_j \circ \tuple{f^{(1)}, \dots, f^{(k)}}}_n,
        \quad \forall 1 \leq i \leq k, 1 \leq j \leq d.
    \end{align}
    This means that $f^{(i)}_{n + e_j}$ is a polynomial function of previous values $f^{(h)}_m$'s
    with $1 \leq h \leq k$ and $m \leq n$.
    Polynomial $p^{(i)}_j \in \poly \Z k$ has total degree $\leq D$ and it is $k$-variate,
    therefore it contains at most $\binom {D + k} k$ monomials.
    Thanks to the previous claim, entry $n \in \N^d$ in the convolution product arising from each such monomial
    can be computed with $$\bigO {(N+D) \cdot \binom {N + d \cdot D - 1} {d \cdot D - 1}}$$ arithmetic operations.
    Combining these two facts yields the next claim.
    \begin{claim*}
        We can compute entry $f^{(i)}_n \in \Z$ of the table with a number of arithmetic operations at most
        \begin{align*}
            \bigO{\binom {D + k} k \cdot (N+D) \cdot \binom {N + d \cdot D - 1} {d \cdot D - 1}}.
        \end{align*}
    \end{claim*}
    By the previous claim and the fact that the table has $\leq k \cdot \binom {N + d} d$ entries,
    we can compute the whole table with a number of arithmetic operations bounded by
    \begin{align*}
        &\leq \bigO{k \cdot \binom {N + d} d \cdot \binom {D + k} k \cdot (N+D) \cdot \binom {N + d \cdot D - 1} {d \cdot D - 1}} \leq \\
        &\leq \bigO{k \cdot (N+d)^d (D+k)^k \cdot (N+D) \cdot (N + d \cdot D - 1)^{d \cdot D - 1}} \leq \\
        &\leq (N + d \cdot D + k)^{\bigO {d \cdot D + k}}.
    \end{align*}
    By~\cref{lem:CDF small coefficient}, numbers in the table are exponentially bounded in $N$, $k$, and $D$,
    and in fact each of them can be stored with a number of bits bounded by
    \begin{align*}  
        \bigO {N \cdot k \cdot \log (N \cdot D + k) + N \cdot D \cdot \log H}.
    \end{align*}
    Since an arithmetic operation (sum or multiplication) on two $b$-bits integer numbers
    can be performed in $\bigO {b^2}$ deterministic time,
    the total time necessary to build the table is
    \begin{align*}
        &(N + d \cdot D + k)^{\bigO {d \cdot D + k}} \cdot (N \cdot k \cdot \log (N \cdot D + k) + N \cdot D \cdot \log H)^{\bigO 1} \\
        &\leq (N + d \cdot D + k)^{\bigO{d \cdot D + k}} \cdot (\log H)^{\bigO 1}. \qedhere
    \end{align*}
\end{proof}

\subsubsection{Complexity of the \CDF~zeroness problem}

In the next lemma we argue that a \CDF~power series is zero
iff just a few initial coefficients are zero.

\CDFnonzeroWitnessBound*
\begin{proof}
    We construct the ideal chain~\cref{eq:WBPP ideal chain} as in the case of \WBPP,
    where $\Delta_w$ is replaced by $L_w$,
    and $\alpha$ by $p$.
    Let $M \in \N$ be the index when this chain stabilises, $I_M = I_{M+1} = I_{M+2} = \cdots$.
    Consider now an arbitrary word $w \in \Sigma^*$.
    Since $L_w p \in I_M$, we can write
    \begin{align*}
        L_w p = q_1 \cdot L_{w_1} p + \cdots + q_m \cdot L_{w_m} p,
    \end{align*}
    for some short words $w_1, \dots, w_m \in \Sigma^{\leq M}$
    and polynomial multipliers $q_1, \dots, q_m \in \poly \Q y$.
    We now compose with $f$ on the right and by~\cref{lem:CDF exchange rule} we have
    \begin{align*}
        \partial x^{\Parikh w} g
        &= L_w p \circ f = \\
        &= \underbrace {(q_1 \circ f)}_{=: \; g_1} \cdot (L_{w_1} p \circ f) + \cdots + \underbrace{(q_m \circ f)}_{=: \;g_m} \cdot (L_{w_m} p \circ f) = \\
        &= g_1 \cdot \partial x^{\Parikh {w_1}} g + \cdots + g_m \cdot \partial x^{\Parikh {w_m}} g.
    \end{align*}
    We now extract the constant term on both sides.
    This is a homomorphic operation $\coefficient {x^0} \_$, and thus we have
    \begin{align*}
        \coefficient {x^0} (\partial x^{\Parikh w} g)
            = \coefficient {x^0} g_1 \cdot \coefficient {x^0} (\partial x^{\Parikh {w_1}} g) + \cdots + \coefficient {x^0} g_m \cdot \coefficient {x^0} (\partial x^{\Parikh {w_m}} g).
    \end{align*}
    By the commutation rule~\cref{eq:coefficient extraction and derivation},
    \begin{align*}
        \coefficient {x^{\Parikh w}} g
            = \coefficient {x^0} g_1 \cdot \coefficient {x^{\Parikh {w_1}}} g + \cdots + \coefficient {x^0} g_m \cdot \coefficient {x^{\Parikh {w_m}}} g.
    \end{align*}
    We have shown that the coefficient $\coefficient {x^{{\Parikh w}}} g$ is a homogeneous linear combination
    of the coefficients $\coefficient {x^n} g$'s of total degree $\onenorm n \leq M$.
    Since $w$ was arbitrary, if all those coefficients $\coefficient {x^n} g$'s are zero,
    then $g$ is zero.
    We now invoke the upper bound on $M$ from~\cref{thm:chain length bound}.
\end{proof}

\CDFzeronessTWOEXPTIME*
\begin{proof}
    Without loss of generality the problem reduces to checking zeroness
    of the first component $g := f^{(1)}$ of a tuple of \CDF~power series $f \in \powerseries \Q x^k$.
    Moreover, we can assume that $f \in \powerseries \Z x^k$
    and that the corresponding \CDF~system of equations~\cref{eq:multivariate CDF - matrix form}
    contains only integer polynomials $P_{ij} \in \poly \Z y$.
    (Were this not the case, a simple scaling by multiplying each $P_{ij}$ with a common denominator of all rational constants of the system
    provides a new integral system where the answer to the zeroness problem is the same.)
    By \cref{lem:CDF nonzero witness degree bound}, it suffices to check $\coefficient {x^n} f^{(1)} = 0$
    for all monomials $x^n$ of degree $\onenorm n \leq N := D^{k^{\bigO {k^2}}}$.
    By~\cref{lem:CDF table} we can compute a table containing all those integer coefficients in time
    \begin{align*}
        (N + d \cdot D + k)^{\bigO{d \cdot D + k}} \cdot (\log H)^{\bigO 1}.
    \end{align*}
    %
    By the definition of $N$, this gives a zeroness algorithm running in deterministic time
    \begin{align*}
        &(D^{k^{\bigO {k^2}}} + d \cdot D + k)^{\bigO{d \cdot D + k}} \cdot (\log H)^{\bigO 1} \\
        &\leq (d \cdot D)^{d \cdot D \cdot k^{\bigO {k^2}}} \cdot (\log H)^{\bigO 1},
    \end{align*}
    which amounts to a \TWOEXPTIME~zeroness algorithm.
\end{proof}